\documentclass[11pt]{article}
\usepackage{amsfonts,amsmath,amssymb,amsthm}
\usepackage{paralist}
\usepackage{algpseudocode}
\usepackage{algorithm,algorithmicx}
\usepackage{color}
\usepackage{graphicx}
\usepackage{comment}
\usepackage{hyperref}
\usepackage{multirow}
\usepackage{float}
\usepackage{subfigure}
 
\newlength{\defbaselineskip}
\newtheorem{theorem}{Theorem}

\newtheorem{definition}{Definition}
\newtheorem{lemma}{Lemma}
\newtheorem{proposition}{Proposition}

\DeclareMathOperator{\poly}{poly}
\DeclareMathOperator{\bigO}{\mathcal{O}}

\DeclareMathOperator{\nnz}{nnz}
\DeclareMathOperator{\A}{\mathcal{A}}
\DeclareMathOperator{\R}{\mathbb{R}}

\DeclareMathOperator{\C}{\mathcal{C}}

\DeclareMathOperator{\E}{\mathcal{E}}

\DeclareMathOperator{\Exp}{\mathbf{E}}
\DeclareMathOperator{\Var}{\mathbf{Var}}

\allowdisplaybreaks

\begin{document}

\title{Quantile Regression for Large-scale Applications%
\footnote{A conference version of this paper appears under the same title in the Proc. of the 2013 ICML~\cite{YMM13_ICML}.}
}

\author{
  Jiyan Yang
  \thanks{
    ICME,
    Stanford University,
    Stanford, CA 94305.
    Email: jiyan@stanford.edu
  }
  \and
  Xiangrui Meng
  \thanks{
    LinkedIn,
    2029 Stierlin Ct,
    Mountain View, CA 94043.
    Email: ximeng@linkedin.com
  }
  \and
  Michael W. Mahoney
  \thanks{
       International Computer Science Institute and the Department of Statistics, 
       University of California at Berkeley, Berkeley, CA 94720.
       Email: mmahoney@icsi.berkeley.edu
  }
}

\date{}
\maketitle


\begin{abstract} 
Quantile regression is a method to estimate the quantiles of the 
conditional distribution of a response variable, and as such it permits a 
much more accurate portrayal of the relationship between the response 
variable and observed covariates than methods such as Least-squares or 
Least Absolute Deviations regression.
It can be expressed as a linear program, and, with appropriate 
preprocessing, interior-point methods can be used to find a solution for
moderately large problems.
Dealing with very large problems, \emph{e.g.}, involving data up to and 
beyond the terabyte regime, remains a challenge.
Here, we present a randomized algorithm that runs in nearly 
linear time in the size of the input and that, with constant probability, 
computes a $(1+\epsilon)$ approximate solution to an arbitrary quantile 
regression problem.
As a key step, our algorithm computes a low-distortion subspace-preserving
embedding with respect to the loss function of quantile regression.
Our empirical evaluation illustrates that our algorithm is competitive with 
the best previous work on small to medium-sized problems, and that in 
addition it can be implemented in MapReduce-like environments and  
applied to terabyte-sized~problems.
\end{abstract}

\section{Introduction}

Quantile regression is a method to estimate the quantiles of the conditional 
distribution of a response variable, expressed as functions of observed 
covariates~\cite{KB78}, in a manner analogous to the way in which 
Least-squares regression estimates the conditional mean. 
The Least Absolute Deviations regression (\emph{i.e.}, $\ell_1$ regression) 
is a special case of quantile regression that involves computing the median 
of the conditional distribution. 
In contrast with $\ell_1$ regression and the more popular $\ell_2$ or 
Least-squares regression, quantile regression involves minimizing 
asymmetrically-weighted absolute residuals. 
Doing so, however, permits a much more accurate portrayal of the 
relationship between the response variable and observed covariates, and it 
is more appropriate in certain non-Gaussian settings.
For these reasons, quantile regression has found applications in many 
areas, \emph{e.g.}, survival analysis and 
economics~\cite{Buc94,KH01,Buhai05}.
As with $\ell_1$ regression, the quantile regression problem can be 
formulated as a linear programming problem, and thus simplex or interior-point 
methods can be applied~\cite{KD93,PK97,Por97}.
Most of these methods are efficient only for problems of small to 
moderate size, and thus to solve very large-scale quantile regression 
problems more reliably and efficiently, we need new computational techniques.


In this paper, we provide a fast algorithm to compute a $(1+\epsilon)$ 
relative-error approximate solution to the over-constrained quantile 
regression problem.
Our algorithm constructs a low-distortion subspace embedding of the form 
that has been used in recent developments in randomized algorithms for 
matrices and large-scale data problems, and our algorithm runs in time that
is nearly linear in the number of nonzeros in the input data.

In more detail, recall that a quantile regression problem can be specified 
by a (design) matrix $A \in \R^{n \times d}$, a (response) vector 
$b \in \R^n$, and a parameter $\tau \in (0,1)$, in which case the quantile 
regression problem can be solved via the optimization problem 
\begin{align}
 \label{ori_pb}
   \textrm{minimize}_{x \in \R^d}   \quad  \rho_\tau(b-Ax)  ,
\end{align}
where $\rho_\tau(y) = \sum_{i = 1}^d \rho_\tau(y_i)$, for $y \in \R^d$, where
\begin{align}
     \label{eq:loss_func}
     \rho_{\tau}(z) = 
      \begin{cases}
           \tau z,  &  z\geq0;  \\
           (\tau-1)z,  &  z<0  ,
           \end{cases}
\end{align} 
for $z \in \R$,
is the corresponding loss function.
In the remainder of this paper, we will use $A$ to denote the augmented 
matrix $\begin{bmatrix} b & -A \end{bmatrix}$, and we will consider 
$A \in \R^{n\times d}$.
With this notation, the quantile regression problem of Eqn.~\eqref{ori_pb} can 
equivalently be expressed as a constrained optimization problem with a single linear 
constraint,
\begin{align}
 \label{pb}
     \textrm{minimize}_{x\in \C} \quad \rho_\tau(Ax),
\end{align}
where $\C = \{ x \in \R^d \,|\, c^Tx = 1 \}$ and $c$ is a unit vector with the 
first coordinate set to be $1$.
The reasons we want to switch from Eqn.~\eqref{ori_pb} to Eqn.~\eqref{pb} are as follows.
    Firstly, it is for notational simplicity in the presentation of our theorems and algorithms. 
    Secondly, all the results about low-distortion or $(1\pm \epsilon)$-subspace embedding in this paper
    holds for any $x \in \R^d$,
     \begin{equation*}
 (1/ \kappa_1) \|Ax\|_1 \leq \|\Pi Ax\|_1 \leq  \kappa_2 \|Ax\|_1.
  \end{equation*}
    In particular, we can consider $x$ in some specific subspace of $\R^d$.
    For example, in our case, $x \in \C$.
    Then, the equation above is equivalent to the following,
     \begin{equation*}
 (1/ \kappa_1) \|b-Ax\|_1 \leq \|\Pi b - \Pi Ax\|_1 \leq  \kappa_2 \|b - Ax\|_1.
  \end{equation*}
  Therefore, using notation $Ax$ with $x$ in some constraint is a more general form of expression.
We will focus on very over-constrained problems with size $n \gg d$. 

Our main algorithm depends on a technical result, presented as 
Lemma~\ref{lem:samp_lem}, which is of independent interest.
Let $A \in \R^{n \times d}$ be an input matrix, and let 
$S \in \R^{s \times n}$ be a random sampling matrix constructed based on 
the following importance sampling probabilities,
 \begin{align*}
   p_i = \min\{1, s \cdot \| U_{(i)} \|_1 /  \|U\|_1 \},
\end{align*}
where $\| \cdot \|_1$ is the element-wise $\ell_1$-norm, and where $U_{(i)}$ 
is the $i$-th row of an $\ell_1$ well-conditioned basis $U$ for the range of 
$A$ (see Definition~\ref{def:basis} and Proposition~\ref{lem:const_samp}).
Then, Lemma~\ref{lem:samp_lem} states that, for a sampling complexity $s$ 
that depends on $d$ but is independent of $n$, 
\begin{equation*}
 (1-\epsilon) \rho_\tau(Ax)   \leq  \rho_\tau(SAx)  \leq (1+\epsilon) \rho_\tau(Ax)
\end{equation*}
will be satisfied for every $x \in \R^d$.

Although one could use, \emph{e.g.}, the algorithm of~\cite{DDHKM09} to 
compute such a well-conditioned basis $U$ and then ``read off'' the 
$\ell_1$-norm of the rows of $U$, doing so would be much slower than the 
time allotted by our main algorithm.
As Lemma~\ref{lem:samp_lem} enables us to leverage
the fast quantile regression theory and the algorithms developed for $\ell_1$ regression,
we provide two sets of additional results, most of which are built from the previous work.
First, we describe three algorithms (Algorithm~\ref{spc_alg},
Algorithm~\ref{cond_alg}, and Algorithm~\ref{new_cond_alg}) for 
computing an implicit representation of a well-conditioned basis; and 
second, we describe an algorithm (Algorithm~\ref{embed_alg}) for 
approximating the $\ell_1$-norm of the rows of the well-conditioned basis from 
that implicit representation.
For each of these algorithms, we prove quality-of-approximation bounds in quantile regression problems,
and we show that they run in nearly ``input-sparsity'' time, \emph{i.e.}, in 
$\bigO(\nnz(A) \cdot \log n)$ time, where $\nnz(A)$ is the number of nonzero 
elements of $A$, plus lower-order terms.
These lower-order terms depend on the time to solve the subproblem we 
construct; and they depend on the smaller dimension $d$ but not on the 
larger dimension $n$.
Although of less interest in theory, these lower-order terms can be 
important in practice, as our empirical evaluation will demonstrate.

We should note that, of the three algorithms for computing a 
well-conditioned basis, the first two appear in~\cite{MM12} and are stated 
here for completeness; and the third algorithm, which is new to this paper, 
is \emph{not} uniformly better than either of the two previous algorithms 
with respect to either condition number or the running time.
(In particular, Algorithm~\ref{spc_alg} has slightly better running time, 
and Algorithm~\ref{cond_alg} has slightly better conditioning properties.)
Our new conditioning algorithm is, however, only slightly worse than the 
better of the two previous algorithms with respect to each of those two
measures.
Because of the trade-offs involved in implementing quantile regression 
algorithms in practical settings, our empirical results show that by using 
a conditioning algorithm that is only slightly worse than the best previous 
conditioning algorithms for each of these two criteria, our new conditioning 
algorithm can lead to better results than either of the previous algorithms 
that was superior by only one of those criteria.

Given these results, our main algorithm for quantile regression is presented 
as Algorithm~\ref{qr_alg}. 
Our main theorem for this algorithm, Theorem~\ref{qr_thm}, states 
that, with constant probability, this algorithm returns a 
$(1+\epsilon)$-approximate solution to the quantile regression problem; and 
that this solution can be obtained in $\bigO(\nnz(A) \cdot \log n)$ time, 
plus the time for solving the subproblem (whose size is 
$\bigO(\mu d^{3} \log(\mu/\epsilon)/\epsilon^2) \times d$, where 
$\mu = \frac{\tau}{1-\tau}$, independent of $n$, when $\tau \in [1/2,1)$).

We also provide a detailed empirical evaluation of our main algorithm for 
quantile regression, including characterizing the quality of the solution 
as well as the running time, as a function of the high dimension $n$, the
lower dimension $d$, the sampling complexity $s$, and the quantile parameter 
$\tau$.
Among other things, our empirical evaluation demonstrates that the output of 
our algorithm is highly accurate, in terms of not only objective function 
value, but also the actual solution quality (by the latter, we mean a norm of the
difference between the exact solution to the full problem and the solution to
the subproblem constructed by our algorithm), when compared with the exact quantile regression, as measured in three different norms.
More specifically, our algorithm yields 2-digit accuracy solution
by sampling only, \emph{e.g.}, about $0.001\%$ of a problem with 
size $2.5e9 \times 50$.\footnote{We use this notation throughout; \emph{e.g.}, by $2.5e9 \times 50$, 
we mean that $n=2.5 \times 10^{9}$ and $d=50$.}
Our new conditioning algorithm outperforms other conditioning-based methods, 
and it permits much larger small dimension $d$ than previous conditioning 
algorithms.
In addition to evaluating our algorithm on moderately-large data that can 
fit in RAM, we also show that our algorithm can be implemented in 
MapReduce-like environments and applied to computing the solution of
terabyte-sized quantile regression problems.


The best previous algorithm for moderately large quantile regression 
problems is due to \cite{PK97} and \cite{Por97}.
Their algorithm uses an interior-point method on a smaller problem that has 
been preprocessed by randomly sampling a subset of the data.
Their preprocessing step involves predicting the sign of each 
$A_{(i)} x^* - b_i$, where $A_{(i)}$ and $b_i$ are the $i$-th row of the input 
matrix and the $i$-th element of the response vector, respectively, and $x^*$ is an optimal solution to the 
original problem. 
When compared with our approach, they compute an optimal solution, while we 
compute an approximate solution; but in worst-case analysis it can be shown that with
high probability our algorithm is guaranteed to work, while their algorithm do not come with such guarantees.
Also, the sampling complexity of their algorithm depends on the higher
dimension $n$, while the number of samples required by our algorithm depends 
only on the lower dimension $d$; but our sampling is with respect to a 
carefully-constructed nonuniform distribution, while they sample uniformly 
at random. 

For a detailed overview of recent work on using randomized algorithms to 
compute approximate solutions for least-squares regression and related 
problems, see the recent review~\cite{Mah-mat-rev_BOOK}.
Most relevant for our work is the algorithm of~\cite{DDHKM09} that 
constructs a well-conditioned basis by ellipsoid rounding and a 
subspace-preserving sampling matrix in order to approximate the solution of 
general $\ell_p$ regression problems, for $p \in [1,\infty)$, in roughly 
$\bigO(nd^5 \log n)$; the algorithms of~\cite{SW11} and~\cite{CDMMMW12} 
that use the ``slow'' and ``fast' versions of Cauchy Transform to obtain a 
low-distortion $\ell_1$ embedding matrix and solve the over-constrained 
$\ell_1$ regression problem in $\bigO(nd^{1.376+})$ and $\bigO(nd\log n)$ 
time, respectively; 
and the algorithm of~\cite{MM12} that constructs low-distortion embeddings 
in ``input sparsity'' time and uses those embeddings to construct a well-conditioned basis
and approximate the 
solution of the over-constrained $\ell_1$ regression problem in 
$\bigO(\nnz(A) \cdot \log n + \poly(d) \log(1/\epsilon)/\epsilon^2)$ time. 
In particular, we will use the two conditioning methods in~\cite{MM12}, as 
well as our ``improvement'' of those two methods, for constructing 
$\ell_1$-norm well-conditioned basis matrices in nearly input-sparsity time.
In this work, we also demonstrate that such well-conditioned basis in $\ell_1$ sense
can be used to solve over-constrained quantile regression problem.

\section{Background and Overview of Conditioning Methods}

\subsection{Preliminaries}
\label{pre}

We use $\| \cdot \|_1$ to denote the element-wise $\ell_1$ norm for both 
vectors and matrices; and we use $[n]$ to denote the set $\{1,2,\ldots,n\}$. 
For any matrix $A$, $A_{(i)}$ and $A^{(j)}$ denote the $i$-th row and the 
$j$-th column of $A$, respectively; and $\A$ denotes the column space of $A$.
For simplicity, we assume $A$ has full column rank; and we always assume that
$\tau \geq \frac{1}{2}$. 
All the results hold for $\tau < \frac{1}{2}$ by simply switching the 
positions of $\tau$ and $1-\tau$.

Although $\rho_\tau( \cdot )$, defined in Eqn.~\ref{eq:loss_func}, is not a norm, since the loss function does 
not have the positive linearity, it satisfies some ``good'' properties, as 
stated in the following lemma:
\begin{lemma}
 \label{properties}
Suppose that $\tau \geq \frac{1}{2}$.  
Then, for any $x, y \in \R^d, a\geq 0$, the following hold:
 \begin{enumerate}
    \item  $\rho_\tau(x + y) \leq \rho_\tau(x) + \rho_\tau(y)$;
    \item  $ (1-\tau)\|x\|_1 \leq \rho_\tau(x) \leq   \tau \|x\|_1$;
    \item $ \rho_\tau(ax) = a\rho_\tau(x)$; and
    \item $ |\rho_\tau(x) - \rho_\tau(y)| \leq \tau \|x-y\|_1$.
 \end{enumerate}
\end{lemma}
 \begin{proof}
  It is trivial to prove every equality or inequality for $x, y$ in one dimension. Then by the definition of $\rho_\tau( \cdot )$ for vectors, the inequalities and equalities hold for general $x$ and $y$.
\end{proof}

To make our subsequent presentation self-contained, here we will provide
a brief review of recent work on subspace embedding algorithms.
We start with the definition of a low-distortion embedding matrix for $\A$
in terms of $\| \cdot \|_1$, see \emph{e.g.},~\cite{MM12}.

\begin{definition}[Low-distortion $\ell_1$ Subspace Embedding]
 \label{def:low_dist}
Given  $A \in \R^{n \times d}$, $\Pi \in \R^{r\times n}$ is a low-distortion embedding of $\A$ if $r = \poly(d)$ and
for all $x \in \R^d$,
 \begin{equation*}
 (1/ \kappa_1) \|Ax\|_1 \leq \|\Pi Ax\|_1 \leq  \kappa_2 \|Ax\|_1.
  \end{equation*}
 where $\kappa_1$ and $\kappa_2$ are low-degree polynomials of $d$.
\end{definition}

\noindent
The following stronger notion of a $(1\pm\epsilon)$-distortion 
subspace-preserving embedding will be crucial for our method.
In this paper, the ``measure functions'' we will consider are $\| \cdot \|_1$ 
and $\rho_\tau(\cdot)$. 

\begin{definition}[$(1\pm\epsilon)$-distortion Subspace-preserving Embedding]
Given $A \in \R^{n \times d}$ and a measure function of vectors $f(\cdot)$,
 $S \in \R^{s\times n}$ is a $(1 \pm \epsilon)$-distortion 
subspace-preserving matrix of $(\A, f(\cdot))$ if $s = \poly (d)$ and for all 
$x \in \R^d$,
\begin{equation*}
   (1 - \epsilon) f(Ax)  \leq  f(SAx)  \leq (1 + \epsilon) f(Ax).
\end{equation*}
 Furthermore, if $S$ is a sampling matrix (one nonzero element per row in $S$),
   we call it a $(1\pm\epsilon)$-distortion subspace-preserving sampling matrix.
 \end{definition}

\noindent
In addition, the following notion, originally introduced by~\cite{Cla05}, and stated more precisely in~\cite{DDHKM09}, 
of a basis that is well-conditioned for the $\ell_1$ norm will also be 
crucial for our method.

\begin{definition}[$(\alpha, \beta)$-conditioning and well-conditioned basis]
 \label{def:basis}
Given $A \in \R^{n \times d}$,
$A$ is $(\alpha, \beta)$-conditioned if $ \|A\|_1 \leq \alpha$ and for all $x\in \R^q, \|x\|_\infty \leq \beta \|Ax\|_1$. 
Define $\kappa(A)$ as the minimum value of $\alpha \beta$ such that $A$ is $(\alpha, \beta)$-conditioned.
We will say that a basis $U$ of $A$ is a well-conditioned basis if 
$\kappa=\kappa(U)$ is a polynomial in $d$, independent of~$n$.
\end{definition}

For a low-distortion embedding matrix for $(\A, \|\cdot\|_1)$, we next state 
a fast construction algorithm that runs in ``input-sparsity'' time by 
applying the Sparse Cauchy Transform.
This was originally proposed as Theorem 2 in \cite{MM12}.

\begin{lemma}[Fast construction of Low-distortion $\ell_1$ Subspace Embedding Matrix, from \cite{MM12}]
 \label{lem:spc_low_dist}
Given $A \in \R^{n \times d}$ with full column rank, let $\Pi_1 = SC \in \R^{r_1 \times n}$, where $S \in \R^{r_1 \times n}$ has each column chosen independently and uniformly from the $r_1$ standard basis vector of $\R^{r_1}$, and where $C \in \R^{n \times n}$ is a diagonal matrix with diagonals chosen independently from Cauchy distribution. Set $r_1 = \omega d^5\log^5 d$ with $\omega$ sufficiently large. Then, with a constant probability, we have
 \begin{equation}
    1 / \bigO(d^2 \log^2 d) \cdot \|Ax\|_1 \leq \|\Pi_1Ax\|_1 \leq \bigO(d \log d) \cdot \|Ax\|_1,   \:\:  \forall x \in \R^d.
 \end{equation}
In addition, $\Pi_1 A$ can be computed in $\bigO(\textup{nnz}(A))$ time.
\end{lemma}

\noindent
{\bf Remark.}
This result has very recently been improved.
In \cite{WZ13}, the authors show that one can achieve a $\bigO(d^2 \log^2 d)$ distortion $\ell_1$ subspace embedding matrix
with embedding dimension $\bigO(d \log d)$ in $\nnz(A)$ time by
 replacing Cauchy variables in the above lemma with exponential variables. 
Our theory can also be easily improved by using this improved lemma.

Next, we state a result for the fast construction of a 
$(1\pm\epsilon)$-distortion subspace-preserving sampling matrix for 
$(\A, \| \cdot \|_1)$, from Theorem 5.4 in \cite{CDMMMW12}, with $p = 1$, 
as follows.

\begin{lemma}[Fast construction of 
 $\ell_1$ Sampling Matrix, from Theorem 5.4 in \cite{CDMMMW12}]
 \label{lem:l1_samp}
Given a matrix $A \in \R^{n \times d}$ and a matrix $R \in \R^{d\times d}$ 
such that $AR^{-1}$ is a well-conditioned basis for $\A$ with condition 
number $\kappa$, 
it takes $\bigO(\nnz(A) \cdot \log n)$ time to compute a sampling matrix $S \in \R^{s\times n}$ with
  $s = \bigO(\kappa d \log(1/\epsilon)/\epsilon^2)$ such that with a constant probability, for any $x \in \R^d$,
   $$ (1-\epsilon) \|Ax\|_1 \leq \|SAx\|_1 \leq (1+\epsilon) \|Ax\|_1. $$
 \end{lemma}

\noindent
We also cite the following lemma for finding a matrix $R$, such that $AR^{-1}$ is a 
well-condition basis, which is based on ellipsoidal rounding proposed 
in~\cite{CDMMMW12}.

\begin{lemma}[Fast Ellipsoid Rounding, from \cite{CDMMMW12}]
 \label{lem:er}
 Given an $n \times d$ matrix $A$, by applying an ellipsoid rounding method, it takes at most $\bigO(nd^3 \log n)$ time to find a matrix $R \in \R^{d\times d}$ such that $\kappa(AR^{-1}) \leq 2d^2$.
\end{lemma}

Finally, two important ingredients for proving subspace preservation are 
$\gamma$-nets and tail inequalities. 
Suppose that $Z$ is a point set and $\| \cdot \|$ is a metric on $Z$. 
A subset $Z_\gamma$ is called a $\gamma$-net for some $\gamma>0$ if for every
$x\in Z$ there is a $y \in Z_\gamma$ such that $\|x-y\|\leq \gamma$. 
It is well-known that the unit ball of a $d$-dimensional subspace has a
$\gamma$-net with size at most $(3/\gamma)^d$~\cite{BLM89}. 
Also, we will use the standard Bernstein inequality to prove concentration
results for the sum of independent random variables.

\begin{lemma}[Bernstein inequality, \cite{BLM89}]
   Let $X_1, \ldots, X_n$ be independent random variables with zero-mean. Suppose that $|X_i| \leq M$, for $i \in [n]$, then for any positive number $t$, we have
   \begin{equation*}
     \Pr\left[ \sum_{i \in [n]} X_i > t \right]  \leq \exp \left( -\frac{ t^2/2 }{ \sum_{i \in [n]} \Exp X_j^2 + Mt/3} \right). 
   \end{equation*}
\end{lemma}

\subsection{Conditioning methods for $\ell_1$ regression problems}
 \label{cond}

Before presenting our main results, we start here by outlining the theory for 
conditioning for overconstrained $\ell_1$ (and $\ell_p$) regression problems.

The concept of a well-conditioned basis $U$ (recall 
Definition~\ref{def:basis}) plays an important role in our algorithms, and 
thus in this subsection we will discuss several related conditioning methods.
By a conditioning method, we mean an algorithm for finding, for an input matrix $A$, a 
well-conditioned basis, \emph{i.e.}, either finding a well-conditioned matrix 
$U$ or finding a matrix $R$ such that $U=AR^{-1}$ is well-conditioned.
There exist many approaches that have been proposed for conditioning.
The two most important properties of these methods for our subsequent 
analysis are: (1) the condition number $\kappa=\alpha\beta$; and (2) the 
running time to construct $U$ (or $R$).
The importance of the running time should be obvious; but the condition 
number directly determines the number of rows that we need to select, and
thus it has an indirect effect on running time (via the time required to 
solve the subproblem).
See Table~\ref{cond_table} for a summary of the basic properties of the 
conditioning methods that will be discussed in this subsection.

 \begin{table}[t]
  \centering
   \begin{tabular}{c|ccc}
     name  &  running time & $\kappa$ &  type  \\
    \hline
      SC \cite{SW11}    & $\bigO(nd^2 \log d)$   &  $\bigO(d^{5/2} \log^{3/2} n)$ & QR  \\
      FC \cite{CDMMMW12}   &  $\bigO(nd \log d)$  &  $\bigO(d^{7/2} \log^{5/2} n)$  & QR \\ 
      Ellipsoid rounding \cite{Cla05} &  $\bigO(nd^5 \log n)$  & $d^{3/2} (d+1)^{1/2}$ & ER \\
      Fast ellipsoid rounding \cite{CDMMMW12} &  $\bigO(nd^3 \log n)$  & $2d^2$ & ER \\
      SPC1 \cite{MM12}  &  $\bigO(\textup{nnz}(A))$  &   $\bigO(d^{\frac{13}{2}} \log^{\frac{11}{2}} d)$ & QR  \\
      SPC2 \cite{MM12}   &  $\bigO(\textup{nnz}(A) \cdot \log n)$ + \texttt{ER\_small}  &   $6d^2$ & QR+ER  \\
      SPC3 (proposed in this article)  &  $\bigO(\textup{nnz}(A) \cdot \log n)$  + \texttt{QR\_small}  &   $\bigO(d^{\frac{19}{4}} \log^{\frac{11}{4}} d)$ & QR+QR
   \end{tabular}
   \caption{Summary of running time, condition number, and type of 
            conditioning methods proposed recently.
            QR and ER refer, respectively, to methods based on the QR 
            factorization and methods based on Ellipsoid Rounding, as 
            discussed in the text.
            \texttt{QR\_small} and \texttt{ER\_small} denote the running 
            time for applying QR factorization and Ellipsoid Rounding,
            respectively, on a small matrix with size independent of $n$.}
    \label{cond_table}
 \end{table}

In general, there are three basic ways for finding a matrix $R$ such that 
$U=AR^{-1}$ is well-conditioned: those based on the QR factorization; those 
based on Ellipsoid Rounding; and those based on combining the two basic 
methods. 
\begin{itemize}
\item
\textbf{Via QR Factorization (QR).}
To obtain a well-conditioned basis, one can first construct a low-distortion 
$\ell_1$ embedding matrix.  By Definition~\ref{def:low_dist}, this means 
finding a $\Pi \in \R^{r \times d}$, 
such that for any $x \in \R^d$,
 \begin{equation}
   \label{low_dist}
     (1/ \kappa_1) \|Ax\|_1 \leq \|\Pi Ax\|_1 \leq \kappa_2 \|Ax\|_1,
  \end{equation}
where $r \ll n$ and is independent of $n$ and
the factors $\kappa_1$ and $\kappa_2$ here will be low-degree 
polynomials of $d$ (and related to $\alpha$ and $\beta$ of 
Definition~\ref{def:basis}).
For example, $\Pi$ could be the Sparse Cauchy Transform described in Lemma~\ref{lem:spc_low_dist}.
After obtaining $\Pi$, by calculating a matrix $R$ such that $\Pi AR^{-1}$ 
has orthonormal columns, the matrix $AR^{-1}$ is a well-conditioned basis 
with $\kappa \leq d \sqrt{r} \kappa_1\kappa_2$.  See Theorem 4.1 in \cite{MM12} for more details.
Here, the matrix $R$ can be obtained by a QR factorization (or, alternately,
the Singular Value Decomposition).
As the choice of $\Pi$ varies, the condition number of $AR^{-1}$, \emph{i.e.}, $\kappa(AR^{-1})$, and the
corresponding running time will also vary, and there is in general a trade-off 
among these.

For simplicity, the acronyms for these types of conditioning methods will 
come from the name of the corresponding transformations:
SC stands for Slow Cauchy Transform from \cite{SW11}; 
FC stands for Fast Cauchy Transform from \cite{CDMMMW12}; and
SPC1 (see Algorithm~\ref{spc_alg}) will be the first method based on the 
Sparse Cauchy Transform (see Lemma~\ref{lem:spc_low_dist}).
We will call the methods derived from this scheme QR-type methods.
\item
\textbf{Via Ellipsoid Rounding (ER).}
Alternatively, one can compute a well-conditioned basis by applying ellipsoid 
rounding.
This is a deterministic algorithm that computes a $\eta$-rounding of a 
centrally symmetric convex set $\C = \{x \in \R^d | \|Ax\|_1 \leq 1\}$.
By $\eta$-rounding here we mean finding an ellipsoid 
$\E = \{ x \in \R^d | \|Rx\|_2 \leq 1\}$, satisfying
$\E/\eta \subseteq \C  \subseteq \E$,
 which implies 
$\|Rx\|_2 \leq \|Ax\|_1 \leq \eta \|Rx\|_2$, $\forall x \in \R^d$.
With a transformation of the coordinates, it is not hard to show the following,
\begin{align}
 \label{er}
  \|x\|_2 \leq \|AR^{-1}x\|_1 \leq \eta \|x\|_2.
\end{align}
From this, it is not hard to show the following inequalities,
\begin{align*}
  \|AR^{-1}\|_1 &\leq \sum_{j \in [d]}  \|AR^{-1}e_j\|_1 \leq \sum_{j \in [d]} \eta \|e_j\|_2  \leq d\eta,  \\
  \|AR^{-1}x\|_1 &\geq \|x\|_2 \geq  \|x\|_\infty.
\end{align*}
This directly leads to a well-conditioned matrix $U=AR^{-1}$ with 
$\kappa \leq d \eta$.
Hence, the problem boils down to finding a $\eta$-rounding with $\eta$ small in a reasonable time.

By Theorem 2.4.1 in \cite{Lov86}, one can find a $(d(d+1))^{1/2}$-rounding in polynomial time.
This result was used by \cite{Cla05} and \cite{DDHKM09}.
As we mentioned in the previous section, Lemma~\ref{lem:er},
in \cite{CDMMMW12}, a new fast ellipsoid rounding algorithm was proposed.
For an $n\times d$ matrix $A$ with full rank,
it takes at most $\bigO(nd^3\log n)$ time to find a matrix $R$
such that $AR^{-1}$ is a well-conditioned basis with $\kappa \leq 2d^2$.
We will call the methods derived from this scheme ER-type methods.
\item
\textbf{Via Combined QR+ER Methods.}
Finally, one can construct a well-conditioned basis by combining QR-like and
ER-like methods.
For example, after we obtain $R$ such that $AR^{-1}$ is a well-conditioned 
basis, as described in Lemma~\ref{lem:l1_samp}, one can then construct a 
$(1\pm \epsilon)$-distortion subspace-preserving sampling matrix $S$ in 
$\bigO(\nnz(A) \cdot \log n)$ time.
We may view that the price we pay for obtaining $S$ is very low in terms of running time. 
Since $S$ is a sampling matrix with constant distortion factor and 
since the dimension of $SA$ is independent of $n$, we can apply 
additional operations on that smaller matrix in order to obtain a better 
condition number, without much additional running time, in theory at least, 
if $n \gg \poly(d)$, for some low-degree $\poly(d)$.

Since the bottleneck for ellipsoid rounding is its running time, when 
compared to QR-type methods, one possibility is to apply ellipsoid rounding 
on $SA$.  Since the bigger dimension of $SA$ only depends on $d$, the running 
time for computing $R$ via ellipsoid rounding will be acceptable if 
$n \gg \poly(d)$.
As for the condition number, for any general $\ell_1$ subspace embedding 
$\Pi$ satisfying 
Eqn.~\eqref{low_dist}, \emph{i.e.}, which preserves the $\ell_1$ norm up to
some factor determined by $d$, including $S$, if we apply ellipsoid rounding 
on $\Pi A$, then the resulting $R$ may still satisfy Eqn.~\eqref{er} with some 
$\eta$.
In detail, viewing $R^{-1}x$ as a vector in $\R^d$, from Eqn.~\eqref{low_dist}, we have
\begin{equation*}
 (1/\kappa_2) \|\Pi AR^{-1}x\|_1 \leq \|AR^{-1}x\|_1 \leq  \kappa_1 \|\Pi AR^{-1}x\|_1.
\end{equation*}
In Eqn.~\eqref{er}, replace $A$ with $\Pi A$, combining the inequalities above, we have
 \begin{equation*}
 (1/\kappa_2) \|x\|_2 \leq \|AR^{-1}x\|_1 \leq \eta\kappa_1 \|x\|_2.
\end{equation*}
With appropriate scaling, one can show that $AR^{-1}$ a well-conditioned matrix 
with $\kappa = d \eta \kappa_1\kappa_2$.
Especially, when $S$ has constant distortion, say $(1 \pm 1/2)$,
the condition number is preserved at sampling complexity $\bigO(d^2)$, while
the running time has been reduced a lot, when compared to the vanilla ellipsoid rounding method.
(See Algorithm~\ref{cond_alg} (SPC2) below for a version of this method.)

A second possibility is to view $S$ as a sampling matrix satisfying 
Eqn.~\eqref{low_dist} with $\Pi = S$. 
Then, according to our discussion of the QR-type methods, if we compute the 
QR factorization of $SA$, we may expect the resulting $AR^{-1}$ to be a 
well-conditioned basis with lower condition number $\kappa$.
As for the running time, QR factorization on a smaller matrix will be 
inconsequential, in theory at least.
(See Algorithm~\ref{new_cond_alg} (SPC3) below for a version of this method.)
\end{itemize}
In the remainder of this subsection, we will describe three related 
methods for computing a well-conditioned basis that we will use in our 
empirical evaluations.
Recall that Table~\ref{cond_table} provides a summary of these three methods and the 
other methods that we will use.

We start with the algorithm obtained when we use Sparse Cauchy Transform 
from \cite{MM12} as the random projection $\Pi$ in a vanilla QR-type method.
We call it SPC1 since we will describe two of its variants below.
Our main result for Algorithm~\ref{spc_alg} is given in Lemma~\ref{lem:spc}.
Since the proof is quite straightforward, we omit it here.

\begin{algorithm}[t]
  \caption{SPC1: vanilla QR type method with Sparse Cauchy Transform}
    \label{spc_alg}
  \begin{algorithmic}[1]
    \Require  $A \in \R^{n \times d}$ with full column rank.
    
    \Ensure $R^{-1} \in \R^{d \times d}$ such that $AR^{-1}$ is a well-conditioned basis with $\kappa \leq \bigO(d^{\frac{13}{2}} \log^{\frac{11}{2}} d)$. 
        
    \State Construct a low-distortion embedding matrix $\Pi_1 \in \R^{r_1 \times n}$ of $(\A, \| \cdot \|_1)$ via Lemma~\ref{lem:spc_low_dist}.
        
    \State  Compute $R\in \R^{ d \times d}$ such that $AR^{-1}$ is a well-conditioned basis for $\A$ via QR factorization of $\Pi_1 A$.
    
    \end{algorithmic}
 \end{algorithm}

\begin{lemma}
\label{lem:spc}
Given $A \in \R^{n \times d}$ with full rank, Algorithm~\ref{spc_alg} takes $\bigO(\nnz(A) \cdot \log n)$ time to compute a matrix $R \in \R^{d \times d}$ such that with a constant probability, $AR^{-1}$ is a well-conditioned basis for $\A$ with $\kappa \leq \bigO(d^{\frac{13}{2}} \log^{\frac{11}{2}} d)$.
\end{lemma}

Next, we summarize the two Combined Methods described above in
Algorithm~\ref{cond_alg} and Algorithm~\ref{new_cond_alg}.
Since they are variants of SPC1, we call them SPC2 and SPC3, respectively.
Algorithm~\ref{cond_alg} originally appeared as first four steps of Algorithm 2 
in~\cite{MM12}.
Our main result for Algorithm~\ref{cond_alg} is given in Lemma~\ref{lem:cond}; 
since the proof of this lemma is very similar to the proof of Theorem 7 
in~\cite{MM12}, we omit it here.
Algorithm~\ref{new_cond_alg} is new to this paper.
Our main result for Algorithm~\ref{new_cond_alg} is given in 
Lemma~\ref{lem:new_cond}.

\begin{algorithm}[t]
  \caption{SPC2: QR + ER type method with Sparse Cauchy Transform}
    \label{cond_alg}
  \begin{algorithmic}[1]
    \Require  $A \in \R^{n \times d}$ with full column rank.
    
    \Ensure $R^{-1} \in \R^{d \times d}$ such that $AR^{-1}$ is a well-conditioned basis with $\kappa \leq 6d^2$. 
        
    \State Construct a low-distortion embedding matrix $\Pi_1 \in \R^{r_1 \times n}$ of $(\A, \| \cdot \|_1)$ via Lemma~\ref{lem:spc_low_dist}.
        
    \State  Construct $\tilde R\in \R^{ d \times d}$ such that $A \tilde R^{-1}$ is a well-conditioned basis for $\A$ via QR factorization of $\Pi_1 A$.
    
    \State Compute a $(1 \pm 1/2)$-distortion sampling matrix $\tilde S\in \R^{\poly(d) \times n}$ of $( \mathcal{A}, \| \cdot \|_1)$
                 via Lemma~\ref{lem:l1_samp}.
    
    \State Compute $R \in \R^{d\times d}$ by ellipsoid rounding for $\tilde SA$ via Lemma~\ref{lem:er}.
    
    \end{algorithmic}
 \end{algorithm}
 
\begin{lemma}
\label{lem:cond}
Given $A \in \R^{n \times d}$ with full rank, Algorithm~\ref{cond_alg} takes $\bigO(\nnz(A) \cdot \log n)$ time to compute a matrix $R \in \R^{d \times d}$ such that with a constant probability, $AR^{-1}$ is a well-conditioned basis for $\A$ with $\kappa \leq 6d^2$.
\end{lemma}

\begin{algorithm}[t]
  \caption{SPC3: QR + QR type method with Sparse Cauchy Transform}
  \label{new_cond_alg}
  \begin{algorithmic}[1]
    \Require  $A \in \R^{n \times d}$ with full column rank.
    
    \Ensure $R^{-1} \in \R^{d \times d}$ such that $AR^{-1}$ is a well-conditioned basis with 
    $\kappa \leq \bigO(d^{\frac{19}{4}} \log^{\frac{11}{4}} d)$. 
        
    \State Construct a low-distortion embedding matrix $\Pi_1 \in \R^{r_1 \times n}$ of $(\A, \| \cdot \|_1)$ via Lemma~\ref{lem:spc_low_dist}.   
        
    \State  Construct $\tilde R\in \R^{ d \times d}$ such that $A \tilde R^{-1}$ is a well-conditioned basis for $\A$ via QR factorization of $\Pi_1 A$.
    
    \State Compute a $(1 \pm 1/2)$-distortion sampling matrix $\tilde S\in \R^{\poly(d) \times n}$ of $( \mathcal{A}, \| \cdot \|_1)$ via Lemma~\ref{lem:l1_samp}.
    
    \State Compute $R \in \R^{d\times d}$ via the QR factorization of $\tilde S A$.
    
    \end{algorithmic}
 \end{algorithm}

\begin{lemma}
\label{lem:new_cond}
Given $A \in \R^{n \times d}$ with full rank, Algorithm~\ref{new_cond_alg} takes $\bigO(\nnz(A) \cdot \log n)$ time to compute a matrix $R \in \R^{d \times d}$ such that with a constant probability, $AR^{-1}$ is a well-conditioned basis for $\A$ with $\kappa \leq \bigO(d^{\frac{19}{4}} \log^{\frac{11}{4}} d)$.
\end{lemma}
\begin{proof}
 By Lemma~\ref{lem:spc_low_dist}, in Step 1, $\Pi$ is a low-distortion embedding satisfying Eqn.~\eqref{low_dist}
 with $\kappa_1 \kappa_2 = \bigO(d^3 \log^3 d)$, and $r_1 = \bigO(d^5 \log^5 d)$.
  As a matter of fact, as we discussed in Section~\ref{cond}, the resulting $AR^{-1}$ in Step 2 is a well-conditioned basis with $\kappa =  \bigO(d^{\frac{13}{2}} \log^{\frac{11}{2}} d)$.
  In Step 3, by Lemma~\ref{lem:l1_samp}, the sampling complexity required for obtaining a $(1\pm 1/2)$-distortion sampling matrix is $\tilde s = \bigO(d^{\frac{15}{2}} \log^{\frac{11}{2}} d)$.
  Finally, if we view $\tilde S$ as a low-distortion embedding matrix with $r = \tilde s$ and $\kappa_2\kappa_1 = 3$,
  then the resulting $R$ in Step 4 will satisfy that $AR^{-1}$ is a well-conditioned basis with $\kappa =  \bigO(d^{\frac{19}{4}} \log^{\frac{11}{4}} d)$.
  
  For the running time, it takes $\bigO(\nnz(A))$ time for completing Step 1. In Step 2, the running time is $r_1 d^2 = \poly(d)$.
  As Lemma~\ref{lem:l1_samp} points out, the running time for constructing $\tilde S$ in Step 3 is $\bigO(\nnz(A) \cdot \log n)$. 
  Since the large dimension of $\tilde SA$ is a low-degree polynomial of $d$, the QR factorization of it costs $\tilde s d^2 = \poly(d)$ time in Step 4.
  Overall, the running time of Algorithm~\ref{new_cond_alg} is $\bigO(\nnz(A) \cdot \log n)$.   
\end{proof}

\noindent
Both Algorithm~\ref{cond_alg} and Algorithm~\ref{new_cond_alg} have 
additional steps (Steps 3 \& 4), when compared with Algorithm~\ref{spc_alg}, 
and this leads to some improvements, at the cost of additional computation time.
For example, in Algorithm~\ref{new_cond_alg} (SPC3), we obtain a 
well-conditioned basis with smaller $\kappa$ when comparing to 
Algorithm~\ref{spc_alg} (SPC1).
As for the running time, it will be still $\bigO(\nnz(A) \cdot \log n)$,
since the additional time is for constructing sampling matrix and 
solving a QR factorization of a matrix whose dimensions are determined by $d$.
Note that when we summarize these results in
Table~\ref{cond_table}, we explicitly list the additional running 
time for SPC2 and SPC3, in order to show the tradeoff between these
SPC-derived methods.
We will evaluate the performance of all these methods on quantile 
regression problems in Section~\ref{mid_emp} (except for FC, since it is 
similar to but worse than SPC1, and ellipsoid rounding, since on the full 
problem it is too expensive).
  
\noindent
\textbf{Remark.}
For all the methods we described above,
the output is \emph{not} the well-conditioned matrix $U$, but instead it is the matrix $R$, the inverse of which 
transforms $A$ into $U$.

\noindent
\textbf{Remark.}
As we can see in Table~\ref{cond_table}, with respect to conditioning
quality, SPC2 has the lowest condition number $\kappa$, followed by 
SPC3 and then SPC1, which has the worst condition number.
On the other hand, with respect to running time, SPC1 is the fastest, 
followed by SPC3, and then SPC1, which is the slowest.
(The reason for this ordering of the running time is that SPC2 and SPC3 need 
additional steps and ellipsoid rounding takes longer running time that doing
a QR decomposition.)

\section{Main Theoretical Results}

In this section, we present our main theoretical results on $(1\pm \epsilon)$-distortion 
subspace-preserving embeddings and our fast randomized algorithm for 
quantile regression. 


\subsection{Main technical ingredients}
\label{tech}


In this subsection, we present the main technical ingredients underlying our
main algorithm for quantile regression.
We start with a result which says that if we sample sufficiently many 
(but still only $\poly(d)$) rows according to an appropriately-defined 
non-uniform importance sampling distribution (of the form given in
Eqn.~\eqref{eqn:samp_prob} below), then we obtain a 
$(1\pm\epsilon)$-distortion embedding matrix with respect to the loss 
function of quantile regression.
Note that the form of this lemma is based on ideas 
from~\cite{DDHKM09,CDMMMW12}.

\begin{lemma}[Subspace-preserving Sampling Lemma]
 \label{lem:samp_lem}
Given $A \in \R^{n \times d}$, let $U \in \R^{n\times d}$ be a
well-conditioned basis for $\A$ with condition number $\kappa$. 
For $s>0$, define 
\begin{equation}
\label{eqn:samp_prob}
\hat p_i \geq \min\{1, s \cdot  \|U_{(i)}\|_1 / \|U\|_1  \}   , 
\end{equation}
and let $S \in \R^{n\times n}$
be a random diagonal matrix with $S_{ii} = 1/\hat p_i$ with probability 
$\hat p_i$, and 0 otherwise. 
Then when $\epsilon < 1/2  $ and 
$$
 s \geq \frac{\tau}{1-\tau}  \frac{27\kappa}{\epsilon^2} \left( d\log \left( \frac{\tau}{1-\tau} \frac{18}{\epsilon} \right) + \log\left(\frac{4}{\delta} \right) \right)  , 
$$
with probability at least $1-\delta$, for every $x \in \R^d$, 
 \begin{align}
  \label{bd}
    (1-\varepsilon)\rho_\tau(Ax) \leq \rho_\tau(SAx) \leq (1+\varepsilon)\rho_\tau(Ax).
 \end{align}
\end{lemma}
\begin{proof}
Since $U$ is a well-conditioned basis for the range space of $A$, to prove Eqn.~\eqref{bd} it is equivalent to prove the following holds for all $y \in \R^d$,
 \begin{align}
  \label{bbd2}
    (1-\varepsilon)\rho_\tau(Uy) \leq \rho_\tau(SUy) \leq (1+\varepsilon)\rho_\tau(Uy).
 \end{align}
To prove that Eqn.~\eqref{bbd2} holds for any $y \in \R^d$, firstly, we show that Eqn.~\eqref{bbd2} holds for any fixed $y \in \R^d$; and, secondly, we apply a standard $\gamma$-net argument to show that \eqref{bbd2} holds for every $y \in \R^d$.

Assume that $U$ is $(\alpha, \beta)$-conditioned with $\kappa = \alpha \beta$.
For $i \in [n]$, let $v_i = U_{(i)} y$. Then $\rho_\tau(SUy) = \sum_{i \in [n]}  \rho_\tau(S_{ii} v_i) = \sum_{i \in [n]}  S_{ii} \rho_\tau(v_i)$ since $S_{ii} \geq 0$. Let  $w_i =  S_{ii} \rho_\tau(v_i) - \rho_\tau(v_i)$ be a random variable, and we have
\begin{align*}
w_i =
\begin{cases}
 (\frac{1}{\hat p_i}-1)\rho_\tau(v_i),  &  \mbox{with  probability }  \hat p_i;  \\
  -\rho_\tau(v_i),  &   \mbox{with probability }  1-\hat p_i.
\end{cases} 
\end{align*}
Therefore, $\Exp[w_i]= 0, \Var[w_i] = (\frac{1}{\hat p_i}-1) \rho_\tau(v_i)^2, |w_i| \leq \frac{1}{\hat p_i}\rho_\tau(v_i)$. 
Note here we only consider $i$ such that $\|U_{(i)}\|_1 / \|U\|_1<1$ since otherwise we have $\hat p_i = 1$, and the corresponding term will not contribute to the variance. According to our definition,
$\hat p_i \geq s \cdot  \| U_{(i)} \|_1 /  \| U \|_1 = s \cdot t_i$.
Consider the following,
$$   \rho_\tau(v_i) = \rho_\tau(U_{(i)}y) \leq  \tau \|U_{(i)}y\|_1 \leq \tau  \|(U_{(i)})\|_1 \|y\|_\infty.  $$
Hence,
\begin{align*}
  |w_i|  &\leq  \frac{1}{\hat p_i} \rho_\tau(v_i) \leq \frac{1}{\hat p_i}  \tau  \|U_{(i)}\|_1 \|y\|_\infty \leq \frac{\tau}{s} \|U\|_1 \|y\|_\infty   \\
 & \leq \frac{1}{s} \frac{\tau}{1-\tau} \alpha\beta \rho_\tau(Uy) := M.
\end{align*}
Also,
\begin{equation*}
  \sum_{i \in [n]} \Var[w_i] \leq \sum_{i \in [n]}  \frac{1}{\hat p_i} \rho_\tau(v_i)^2 \leq  M \rho_\tau(Uy).
 \end{equation*}
Applying the Bernstein inequality to the zero-mean random variables $w_i$ gives
 \begin{equation*}
    \Pr\left[ \left | \sum_{i\in [n]} w_i \right | > \varepsilon  \right] \leq 2\exp \left(  \frac{ - \varepsilon^2}{2\sum_i \Var[w_i] + \frac{2}{3}M\epsilon} \right).
 \end{equation*}
Since $\sum_{i\in [n]} w_i =  \rho_\tau(SUy) - \rho_\tau(Uy) $, setting $\varepsilon$ to $\varepsilon \rho_\tau(Uy)$ and plugging all the results we derive above, we have
 \begin{align*}
    &\Pr\left[ \left | \rho_\tau(SUy) - \rho_\tau(Uy) \right | > \varepsilon\rho_\tau(Uy)  \right]  \leq  2\exp \left(  \frac{ - \varepsilon^2\rho^2_\tau(Uy)}{2M\rho_\tau(Uy) + \frac{2\varepsilon}{3}M \rho_\tau(Uy)}\right).  
  \end{align*}
 Let's simplify the exponential term on the right hand side of the above expression: 
 \begin{align*}   
  \frac{ - \varepsilon^2\rho^2_\tau(Uy)}{2M\rho_\tau(Uy) + \frac{2\varepsilon}{3}M \rho_\tau(Uy)}
       = \frac{ -s\varepsilon^2}{\alpha\beta}  \frac{1-\tau}{\tau} \frac{1}{ 2 + \frac{2\varepsilon}{3}} 
           \leq   \frac{ -s\varepsilon^2}{3\alpha\beta}  \frac{1-\tau}{\tau} .  
   \end{align*}
Therefore, when $s \geq  \frac{\tau}{1 - \tau}  \frac{27\alpha\beta}{\epsilon^2} \left( d\log \left( \frac{3}{\gamma} \right) + \log\left(\frac{4}{\delta} \right) \right)$, with probability at least $1 - ( \gamma/3  )^d \delta/2$,
  \begin{align}
    \label{bbd_fixed}
    (1- \epsilon/3) \rho_\tau(Uy) \leq \rho_\tau(SUy) \leq (1+ \epsilon/3)\rho_\tau(Uy),
  \end{align}
where $\gamma$ will be specified later.

We will show that, for all $ z \in \textup{range}(U)$,
\begin{equation}
 \label{ep_bound}
   (1-\epsilon) \rho_\tau(z) \leq \rho_\tau(Sz) \leq (1+\epsilon) \rho_\tau(z).
\end{equation} 
By the positive linearity of $\rho_\tau(\cdot)$, it suffices to show the bound above holds for all $z$ with $\|z\|_1 = 1$.

Next, let $Z = \{ z \in \textup{range}(U) \,|\, \|z\|_1 \leq 1\}$ and construct a
$\gamma$-net of $Z$, denoted by $Z_\gamma$, such that for any $z \in Z$, there
exists a $z_\gamma \in Z_\gamma$ that satisfies $\|z - z_\gamma\|_1 \leq \gamma$.
By \cite{BLM89}, the number of elements in $Z_\gamma$ is at most $(3/\gamma)^d$.
Hence, with probability at least $1-\delta/2$, Eqn.~\eqref{bbd_fixed} holds for all $z_\gamma \in Z_\gamma$.

We claim that, with suitable choice $\gamma$, with probability at least $1-\delta/2$, $S$ will be a $(1 \pm 2/3)$-distortion embedding matrix of $(\A, \| \cdot \|_1)$.
To show this, firstly, we state a similar result for $\| \cdot \|_1$ from Theorem 6 in \cite{DDHKM09} with $p = 1$ as follows.
\begin{lemma}[$\ell_1$ Subspace-preserving Sampling Lemma]
 \label{lem:l1_embed}
 Given $A \in \R^{n \times d}$, let $U \in \R^{n\times d}$ be an $(\alpha, \beta)$-conditioned basis for $\A$. For $s>0$, define
$$\hat p_i \geq \min\{1, s \cdot \| U_{(i)} \|_1 / \|U\|_1 \} ,$$ 
and let $S \in \R^{n\times n}$ be a random diagonal matrix with $S_{ii} = 1/\hat p_i$ with probability $\hat p_i$, and 0 otherwise. 
Then when $\epsilon < 1/2$ and
$$ s \geq  \frac{32 \alpha \beta}{\epsilon^2}  \left(d \log\left( \frac{12}{\epsilon} \right) + \log\left( \frac{2}{\delta} \right) \right) ,$$ 
with probability at least $1-\delta$, for every $x \in \R^d$,
 \begin{align}
    (1-\varepsilon) \|Ax\|_1 \leq \|SAx\|_1 \leq (1+\varepsilon) \|Ax\|_1.
 \end{align}
\end{lemma}

Note here we change the constraint $\epsilon \leq 1/7$ and the original theorem to $\epsilon \leq 1/2$ above. One can easily show that the result still holds with such setting.
If we set $\epsilon = 2/3$ and the failure probability to be at most $\delta/2$, 
the construction of $S$ defined above satisfies conditions of Lemma~\ref{lem:l1_embed} when the expected sampling complexity $s \geq  \bar s :=  72 \alpha \beta  \left(d \log\left( 18 \right) + \log\left( \frac{4}{\delta} \right) \right) $. Then our claim for $S$ holds. Hence we only need to make sure with suitable choice of $\gamma$, we have $s \geq \bar s$.

For any $z$ with $\|z\|_1 = 1$,
 we have
\begin{eqnarray*}
|\rho_\tau(Sz) - \rho_\tau(z)| 
   &\leq & |\rho_\tau(Sz) - \rho_\tau(S z_\gamma)| + |\rho_\tau(S z_\gamma) - \rho_\tau(z_\gamma)| + |\rho_\tau(z_\gamma) - \rho_\tau(z)| \\
 &\leq & \tau \|S (z - z_\gamma)\|_1 + (\epsilon/3) \rho_\tau(z_\gamma) + \tau \|z_\gamma - z\|_1 \\
  &\leq & \tau |\|S(z - z_\gamma)\|_1 - \|(z-z_\gamma)\|_1| +  (\epsilon/3) \rho_\tau(z) + (\epsilon/3)   \rho_\tau(z_\gamma-z) + 2 \tau \|z_\gamma - z\|_1 \\
  &\leq & 2\tau/3 \|z-z_\gamma\|_1 + (\epsilon/3) \rho_\tau(z) + \tau (\epsilon/3)  \|z_\gamma-z\|_1 + 2 \tau \|z_\gamma - z\|_1 \\
  &\leq &  (\epsilon/3) \rho_\tau(z) + \tau \gamma (2/3 + \epsilon/3 + 2) \\
  &\leq & \left( \epsilon/3 + \frac{\tau}{1-\tau} \gamma (2/3 + \epsilon/3 + 2) \right)  \rho_\tau(z)   \\
  &\leq &  \epsilon \rho_\tau(z),
\end{eqnarray*}
where we take $\gamma = \frac{1-\tau}{6\tau} \epsilon$, and the expected sampling size becomes 
 \begin{equation*}
  s  =  \frac{\tau}{1-\tau}  \frac{27\alpha\beta}{\epsilon^2} \left( d\log \left( \frac{\tau}{1-\tau} \frac{18}{\epsilon} \right) + \log\left(\frac{4}{\delta} \right) \right).  
 \end{equation*} 
When $\epsilon < 1/2$, we will have $s > \bar s$. Hence the claim for $S$ holds and Eqn.~\eqref{ep_bound} holds for every $ z \in \textup{range}(U)$.

Since the proof is involved with two random events with failure probability at most $\delta/2$, by a simple union bound,  Eqn.~\eqref{bbd2} holds with probability at least $1 - \delta$. Our results follows since $\kappa = \alpha \beta$.
\end{proof}

\noindent
\textbf{Remark.}
It is not hard to see that for any matrix $S$ satisfying Eqn.~\eqref{bd}, the
rank of $A$ is preserved.

\noindent
\textbf{Remark.}
Given such a subspace-preserving sampling matrix, it is not hard
to show that, by solving the sub-sampled problem induced by $S$, \emph{i.e.},
by solving $\min_{x \in \C} \rho_\tau(SAx)$, then one obtains a
$(1+\epsilon)/(1-\epsilon)$-approximate solution to the original problem. 
For more details, see the proof for Theorem~\ref{qr_thm}.

In order to apply Lemma~\ref{lem:samp_lem} to quantile regression, we need to compute the sampling
probabilities in Eqn.~\eqref{eqn:samp_prob}.
This requires two steps: first, find a well-conditioned basis $U$; and 
second, compute the $\ell_1$ row norms of $U$. 
For the first step, we can apply any method described in the previous 
subsection.
(Other methods are possible, but Algorithms~\ref{spc_alg}, ~\ref{cond_alg}, 
and~\ref{new_cond_alg} are of particular interest due to their nearly 
input-sparsity running time.  
We will now present an algorithm that will perform the second step of approximating the $\ell_1$ row norms of $U$ in the 
allotted $\bigO(\nnz(A) \cdot \log n)$ time.

Suppose we have obtained $R^{-1}$ such that $AR^{-1}$ is a well-conditioned basis.
Consider, next, computing $\hat p_i$ from $U$ (or from $A$ and $R^{-1}$), 
and note that forming $U$ explicitly is expensive both when $A$ is dense and 
when $A$ is sparse. 
In practice, however, we will not need to form $U$ explicitly, and we will
not need to compute the exact value of the $\ell_1$-norm of each row of $U$.
Indeed, it suffices to get estimates of $\|U_{(i)}\|_1$, in which case we 
can adjust the sampling complexity $s$ to maintain a small approximation 
factor. 
Algorithm~\ref{embed_alg} provides a way to compute the estimates of the 
$\ell_1$ norm of each row of $U$ fast and construct the sampling matrix. 
The same algorithm was used in~\cite{CDMMMW12} except for the choice of desired sampling complexity $s$
and we present the entire algorithm for completeness.
Our main result for Algorithm~\ref{embed_alg} is presented in 
Proposition~\ref{lem:const_samp}.

\begin{algorithm}[tb]
  \caption{Fast Construction of $(1\pm \epsilon)$-distortion Sampling Matrix of $(\A, \rho_\tau(\cdot))$}
  \label{embed_alg}

  \begin{algorithmic}[1]
    \Require $A \in \R^{n \times d}, R \in \R^{d \times d}$ such that $AR^{-1}$ is well-conditioned with condition number $\kappa$, $\epsilon \in (0, 1/2)$, $\tau \in [1/2,1)$.
      
    \Ensure Sampling matrix $S \in \R^{n \times n}$.
    
    \State Let $\Pi_2 \in \R^{d \times r_2}$ be a matrix of independent Cauchys with $r_2 = 15 \log (40n)$.
    
    \State Compute $R^{-1} \Pi_2$ and construct $\Lambda = A R^{-1} \Pi_2 \in \R^{n \times r_2}$.
    
    \State For $i \in [n]$, compute $\lambda_i = \textrm{median}_{j\in [r_2]} |\Lambda_{ij}|$.
    
    \State For $s =  \frac{\tau}{1-\tau}  \frac{81\kappa}{\epsilon^2} \left( d\log \left( \frac{\tau}{1-\tau} \frac{18}{\epsilon} \right) + \log 80 \right) $ and $i \in [n]$, compute probabilities
         \begin{equation*}
               \hat p_i = \min \left\{1, s \cdot \frac{\lambda_i}{\sum_{i \in [n]} \lambda_i} \right\}.  
         \end{equation*}
       
     \State Let $S \in \R^{n\times n}$ be diagonal with independent entries 
                 \begin{align*}
   S_{ii} =
    \begin{cases}
    \frac{1}{\hat p_i},  &   \mbox{with probability }  \hat p_i;  \\
      0,  &     \mbox{with probability }  1-\hat p_i.
    \end{cases} 
    \end{align*}
  \end{algorithmic}
\end{algorithm}

\begin{proposition}[Fast Construction of $(1 \pm \epsilon)$-distortion Sampling Matrix]
 \label{lem:const_samp}
Given a matrix $A \in \R^{n \times d}$, and a matrix $R \in \R^{d\times d}$ such that $AR^{-1}$ is a well-conditioned basis for $\A$ with
condition number $\kappa$, Algorithm~\ref{embed_alg} takes
$\bigO(\nnz(A) \cdot \log n)$ time to compute a sampling matrix $S  \in \R^{\hat s\times n}$ (with only one nonzero per row), such that with probability at least $0.9$,
$S$ is a $(1\pm \epsilon)$-distortion sampling matrix. That is for all
$x \in \R^d$,
    \begin{equation}
      \label{bd2}
      (1- \epsilon) \rho_\tau(Ax)  \leq \rho_\tau(SAx)  \leq (1+\epsilon)\rho_\tau(Ax). 
     \end{equation}
Further, with probability at least $1 - o(1)$, $\hat s= \bigO \left( \mu \kappa d\log \left(  \mu / \epsilon \right)  / \epsilon^2 \right)$,
where $\mu = \frac{\tau}{1-\tau}$.
\end{proposition}
\begin{proof}
In this lemma, slightly different from the previous notation,
we will use $s$ and $\hat s$ to denote the actual number of rows selected
and the input parameter for defining the sampling probability, respectively.
From Lemma~\ref{lem:samp_lem}, a $(1 \pm \epsilon)$-distortion
sampling matrix $S$ could be constructed by calculating the $\ell_1$ norms of
the rows of $AR^{-1}$.
Indeed, we will estimate these row norms and adjust the sampling complexity $s$.
According to Lemma 12 in \cite{CDMMMW12}, with probability at least 0.95, the
$\lambda_i, i\in [n]$ we compute in the first three steps of
Algorithm~\ref{embed_alg} satisfy
 \begin{align*}
      \frac{1}{2} \|U_{(i)}\|_1 \leq \lambda_i \leq \frac{3}{2} \|U_{(i)}\|_1,
 \end{align*}
 where $U = AR^{-1}$.
 Conditioned on this high probability event, we set $ \hat p_i \geq \min
 \left\{1, \hat s \cdot \frac{\lambda_i}{\sum_{i \in [n]} \lambda_i} \right\} $.
 Then we will have $ \hat p_i \geq \min \left\{1, \frac{\hat s}{3} \cdot \frac{\|
     U_{(i)} \|_1}{ \| U \|_1} \right\} $.
 Since $\hat s/3$ satisfies the sampling complexity required in
 Lemma~\ref{lem:samp_lem} with $\delta = 0.05$,
 the corresponding sampling matrix $S$ is constructed as desired.
 These are done in Step 4 and Step 5.
 Since the algorithm involves two random events, by a simple union
 bound, with probability at least 0.9, $S$ is a $(1 \pm \epsilon)$-distortion
 sampling matrix.

 By the definition of sampling probabilities, $\Exp[s] = \sum_{i \in [n]} \hat
 p_i \leq \hat s $.
 Note here $s$ is the sum of some random variables and it is tightly
 concentrated around its expectation.
 By a standard Bernstein bound, with probability $1 - o(1)$, $s \leq 2\hat s = \bigO
 \left( \mu \kappa d\log \left( \mu / \epsilon \right) / \epsilon^2
 \right)$, where $\mu = \frac{\tau}{1-\tau}$, as claimed.

 Now let's compute the running time in Algorithm~\ref{embed_alg}.
 The main computational cost comes from Steps 2, 3 and 5.
 The running time in other steps will be dominated by it.
 It takes $d^2r_2$ time to compute $R^{-1}\Pi_2$;
 then it takes $\bigO(\nnz(A) \cdot r_2)$ time to compute $AR^{-1}\Pi_2$; and finally it
 takes $\bigO(n)$ time to compute all the $\lambda_i$ and construct $S$.
 Since $r_2 = \bigO(\log n)$, in total, the running time is $\bigO( (d^2 + \nnz(A))
 \log n + n) = \bigO(\nnz(A) \cdot \log n)$.
 \end{proof}

\noindent
{\bf Remark.}
Such technique can also be used to fast approximate the $\ell_2$ row norms
of a well-conditioned basis by post-multiplying a matrix consisted of Gaussian variables; see \cite{DMMW12}.

\noindent
{\bf Remark.}
In the text before Proposition~\ref{lem:const_samp},
$s$ denotes an input parameter for defining the importance sampling probabilities.
However, the actual sample size might be less than that.
Since Proposition~\ref{lem:const_samp} is about the construction of the sampling matrix $S$,
we let $\hat s$ denote the actual number of row selected.
Also, as stated, the output of Algorithm~\ref{embed_alg} is a $n \times n$ matrix;
but if we zero-out the all-zero rows, then the actual size of $S$ is indeed $\hat s$ by $d$ as described in Proposition~\ref{lem:const_samp}.
Throughout the following text, by sampling size $s$, we mean the desired sampling size
which is the parameter in the algorithm.

\subsection{Main algorithm}
\label{main}
 
In this subsection, we state our main algorithm for computing an approximate 
solution to the quantile regression problem.
Recall that, to compute a relative-error approximate solution, it suffices to
compute a $(1 \pm \epsilon)$-distortion sampling matrix $S$. 
To construct $S$, we first compute a well-conditioned basis $U$ with
Algorithm~\ref{spc_alg}, \ref{cond_alg}, or~\ref{new_cond_alg} (or some 
other conditioning method), and then we apply Algorithm~\ref{embed_alg} to 
approximate the $\ell_1$ norm of each row of $U$. 
These procedures are summarized in Algorithm~\ref{qr_alg}.
The main quality-of-approximation result for this algorithm by using 
Algorithm~\ref{cond_alg} is stated in Theorem~\ref{qr_thm}.

\begin{algorithm}[tb]
  \caption{Fast Randomized Algorithm for Quantile Regression}
  \label{qr_alg}
  \begin{algorithmic}[1]
    \Require $A \in \R^{n \times d}$ with full column rank, $\epsilon \in (0, 1/2)$, $\tau \in [1/2,1)$.
    
    \Ensure An approximated solution $\hat x \in \R^d$ to problem $\textrm{minimize}_{x \in \C} \: \rho_\tau(Ax)$. 
        
    \State Compute $R \in \R^{d \times d}$ such that $AR^{-1}$ is a well-conditioned basis for $\A$ via Algorithm~\ref{spc_alg}, \ref{cond_alg}, or \ref{new_cond_alg}.
      
    \State Compute a $(1 \pm \epsilon)$-distortion embedding $S\in \R^{n \times n}$ of $(\mathcal{A}, \rho_\tau(\cdot))$ via Algorithm~\ref{embed_alg}.
    
     \State Return $\hat x \in \R^d$ that minimizes $\rho_\tau(SAx)$ with respect to $x \in \C$.

  \end{algorithmic}
\end{algorithm}

\begin{theorem}[Fast Quantile Regression]
\label{qr_thm}
Given $ A \in \R^{n \times d}$ and $\varepsilon \in (0,1/2)$, 
if Algorithm~\ref{cond_alg} is used in Step 1,
Algorithm~\ref{qr_alg} returns a vector $\hat x$ that,
with probability at least 0.8,
satisfies
\begin{equation*}
  \rho_\tau(A\hat x) \leq \left( \frac{1+\varepsilon}{1-\varepsilon} \right) \rho_\tau(Ax^*),
\end{equation*}
where $x^*$ is an optimal solution to the original problem.
In addition, the algorithm to construct $\hat x$ runs in time
\begin{equation*}
 \bigO(\textup{nnz}(A) \cdot \log n) + \phi \left(\bigO(\mu d^{3} \log(\mu/\epsilon)/\epsilon^2),d \right),
\end{equation*}
where $\mu = \frac{\tau}{1-\tau}$ and $\phi(s,d)$ is the time to solve a quantile regression problem of size $s \times d$.
\end{theorem}
\begin{proof}
 In Step 1, by Lemma~\ref{lem:cond}, the matrix $R \in \R^{d \times
  d}$ computed by Algorithm~\ref{cond_alg} satisfies that with probability at
least 0.9, $A R^{-1}$ is a well-condition basis for $\A$ with $\kappa =
6d^2$.
The probability bound can be attained by setting the corresponding constants
sufficiently large.
In Step 2, when we apply Algorithm~\ref{embed_alg} to construct the sampling
matrix $S$, by Proposition~\ref{lem:const_samp}, with probability at least 0.9, $S$
will be a $(1 \pm \epsilon)$-distortion sampling matrix of $(\A, \rho_\tau
(\cdot ))$. 
Solving the subproblem $\min_{ x \in \C} \rho_\tau(SAx) $ gives a
$(1+\epsilon)/(1-\epsilon)$ solution to the original problem Eqn.~\eqref{pb}. 
This is because
\begin{equation}
\rho_\tau(A\hat x) 
   \leq \frac{1}{1-\varepsilon} \rho_\tau(SA\hat x) 
    \leq \frac{1}{1-\varepsilon} \rho_\tau(SAx^*) 
\leq \frac{1+\varepsilon}{1-\varepsilon} \rho_\tau(Ax^*),
\end{equation}
where the first and third inequalities come from Eqn.~\eqref{bd2} and the second inequality comes from the fact that $\hat x$ is the minimizer of the subproblem.
Hence the solution $\hat x$ returned by Step 3 satisfies our claim. The whole algorithm involves two random events, the overall success probability is at least 0.8.
 
Now let's compute the running time for Algorithm~\ref{qr_alg}. In Step 1, by Lemma~\ref{lem:cond}, the running time for Algorithm~\ref{cond_alg} to compute $R$ is $\bigO(\nnz A)$. By Proposition~\ref{lem:const_samp}, the running for Step 2 is $\bigO(\nnz(A)\cdot \log n)$. Furthermore, as stated in Proposition~\ref{lem:const_samp} and $\kappa(AR^{-1}) = 2d^2$, with probability $1 - o(1)$, the actual sampling complexity is $\bigO \left( \mu  d^3\log \left(  \mu / \epsilon \right)  / \epsilon^2 \right)$, where $\mu = \tau/(1-\tau)$, and it takes $\phi \left( \bigO \left( \mu d^3\log \left(  \mu / \epsilon \right)  / \epsilon^2 \right), d \right)$ time to solve the subproblem in Step 3. This follows the overall running time of Algorithm~\ref{qr_alg} as claimed.
\end{proof}
   
\noindent
\textbf{Remark.}
As stated, Theorem~\ref{qr_thm} uses Algorithm~\ref{cond_alg} in Step 3; 
we did this since it leads to the best known running-time results in worst-case analysis, but 
our empirical results will indicate that due to various trade-offs the situation is more complex in 
practice.

\noindent
\textbf{Remark.}
Our theory provides a bound on the solution quality, as measured by the 
objective function of the quantile regression problem, and it does not 
provide bounds for the difference between the exact solution vector and 
the solution vector returned by our algorithm.
We will, however, compute this latter quantity in our empirical evaluation.
   

\section{Empirical Evaluation on Medium-scale Quantile Regression}
\label{mid_emp}

In this section and the next section, we present our main empirical results.
We have evaluated an implementation of 
Algorithm~\ref{qr_alg} using several different conditioning methods in 
Step 1.
We have considered both simulated data and real data, and we have 
considered both medium-sized data as well as terabyte-scale data.
In this section, we will summarize our results for medium-sized data.
The results on terabyte-scale data can be found in Section~\ref{large_emp}.

\paragraph{Simulated skewed data.}
For the synthetic data, in order to increase the difficulty for sampling, we 
will add imbalanced measurements to each coordinates of the solution vector.
A similar construction for the test data was appeared in \cite{CDMMMW12}.
Due to the skewed structure of the data, we will call this data set 
``skewed data'' in the following discussion.  
This data set is generated in the following~way. 

\begin{enumerate}
 \item Each row of the design matrix $A$ is a canonical vector.  Suppose the number of measurements on the $j$-th column are $c_j$, where $c_j = q c_{j-1}$, for $j = 2, \ldots, d$.  Here $1 < q \leq 2$. $A$ is a $n \times d$ matrix.
 \item The true vector $x^*$ with length $d$ is a vector with independent Gaussian entries. Let $b^* = Ax^*$.
 \item The noise vector $\epsilon$ is generated with independent Laplacian entries. We scale $\epsilon$ such that $\|\epsilon\| / \|b^*\|  = 0.2$. The response vector is given by
 $b_i = \begin{cases}  500 \epsilon_i &  \textrm{with probability } 0.001;  \\   b^*_i + \epsilon_i   &   \textrm{otherwise}.  \end{cases}$ 
 \end{enumerate}
When making the experiments, we require $c_1 \geq 161$. This implies that if we choose $s/n \geq 0.01$, and perform the uniform sampling, with probability at least 0.8, at least one row in the first block (associated with the first coordinate) will be selected,  due to $1 - (1 - 0.01)^{161} \geq 0.8$. Hence, if we choose $ s \geq 0.01n$, we may expect uniform sampling produce acceptable estimation.

\paragraph{Real census data.}
For the real data,
we consider a data set consisting of a $5\%$ sample of the U.S.
2000 Census data\footnote{U.S. Census, \url{http://www.census.gov/census2000/PUMS5.html}}, consisting of annual salary and related features on people 
who reported that they worked 40 or more weeks in the previous year and 
worked 35 or more hours per week.
The size of the design matrix is $5 \times 10^6$ by 11.

\vspace{3mm}
The remainder of this section will consist of six subsections, the first 
five of which will show the results of experiments on the skewed data,
and then Section~\ref{census_data}, which will show the results on census data.
In more detail, Section~\ref{relerr_s}, \ref{relerr_n},
\ref{relerr_d}, and~\ref{relerr_tau} will summarize the 
performance of the methods in terms of solution quality as the parameters 
$s$, $n$, $d$, and $\tau$, respectively, are varied; and
Section~\ref{running_time} will show how the running time changes as 
$s$, $n$, and $d$ change.

Before showing the details, we provide a quick summary of the numerical results.
We show high quality of approximation on both objective value and solution 
vector by using our main algorithm, \emph{i.e.}, Algorithm~\ref{qr_alg}, with 
various conditioning methods.
Among all the conditioning methods, SPC2 and SPC3 show higher accuracy than 
other methods.
They can achieve 2-digit accuracy by only sampling 1\% of the rows for 
moderately-large dataset.
Also, we show that using conditioning yields much higher accuracy, especially
when approximating the solution vector, as we can see in Figure~\ref{err_s}.
Next, we demonstrate that the empirical results are consistent to our theory,
that is, when we fix the lower dimension of the dataset, $d$, and fix the 
conditioning method we use, we always achieve the same accuracy, regardless 
how large the higher dimension $n$ is, as shown in Figure~\ref{err_n}.  
In Figure~\ref{err_d}, we explore the relationship between the accuracy and 
the lower dimension $d$ when $n$ is fixed.
The accuracy is monotonically decreasing as $d$ increases.
We also show that our algorithms are reliable for $\tau$ ranging from 0.05 
to 0.95 as shown in Figure~\ref{err_tau}, and the magnitude of the relative 
error remains almost the same.
As for the running time comparison, in Figure~\ref{time_s}, 
Figure~\ref{time_n} and Figure~\ref{time_d}, we show that running time of 
Algorithm~\ref{qr_alg} with different conditioning method is consistent with 
our theory. 
Moreover, SPC1 and SPC3 have a much better scalability than other methods,
including the standard solver \texttt{ipm} and best previous sampling 
algorithm \texttt{prqfn}.
For example, for $n = 1e6$ and $d = 280$, we can get at least 1-digit 
accuracy in a reasonable time, while we can only solve problem with size 
$1e6$ by $180$ exactly by using the standard solver in that same amount of 
time.

\subsection{Quality of approximation when the sampling size $s$ changes}
\label{relerr_s}

As discussed in Section~\ref{cond}, we can use one of several methods for 
the conditioning step, \emph{i.e.}, for finding the well-conditioned basis 
$U = AR^{-1}$ in the Step 1 of Algorithm~\ref{qr_alg}. 
Here, we will consider the empirical performance of \emph{six} methods for 
doing this conditioning step, namely: SC, SPC1, SPC2, SPC3, NOCO, and UNIF. 
The first four methods (SC, SPC1, SPC2, SPC3) are described in 
Section~\ref{cond}; 
NOCO stands for ``no conditioning,'' meaning the matrix $R$ in the 
conditioning step is taken to be identity; and, 
UNIF stands for the uniform sampling method, which we include here for 
completeness.
Note that, for all the methods, we compute the row norms of the well-conditioned basis 
exactly instead of estimating them with Algorithm \ref{embed_alg}.
The reason is that this permits a cleaner evaluation of the quantile 
regression algorithm, as this may reduce the error due to the estimating 
step.  We have, however, observed similar results if we approximate the row norms 
well.

Rather than determining the sample size from a given tolerance $\epsilon$, we
let the sample size $s$ vary in a range as an input to the algorithm.
Also, for a fixed data set, we will show the results when $\tau=0.5,0.75,0.95$.
In our figure, we will plot the first and the third quartiles of the relative errors of the
objective value and solution measured in three different norms from $50$ independent
trials. 
We restrict the $y$ axis in the plots to the range of $[0,100]$ to show 
more details.
We start with a test on skewed data with size $1e6 \times 50$.
(Recall that, by $1e6 \times 50$, we mean that 
$n=1 \times 10^{6}$ and $d=50$.)
The resulting plots are shown in Figure~\ref{err_s}.

\begin{figure}[h!tbp]
 \begin{center}
 \begin{tabular}{ccc}
\subfigure[$\tau = 0.5$, $|f-f^*|/|f^*|$]{
   \includegraphics[width=0.3\textwidth] {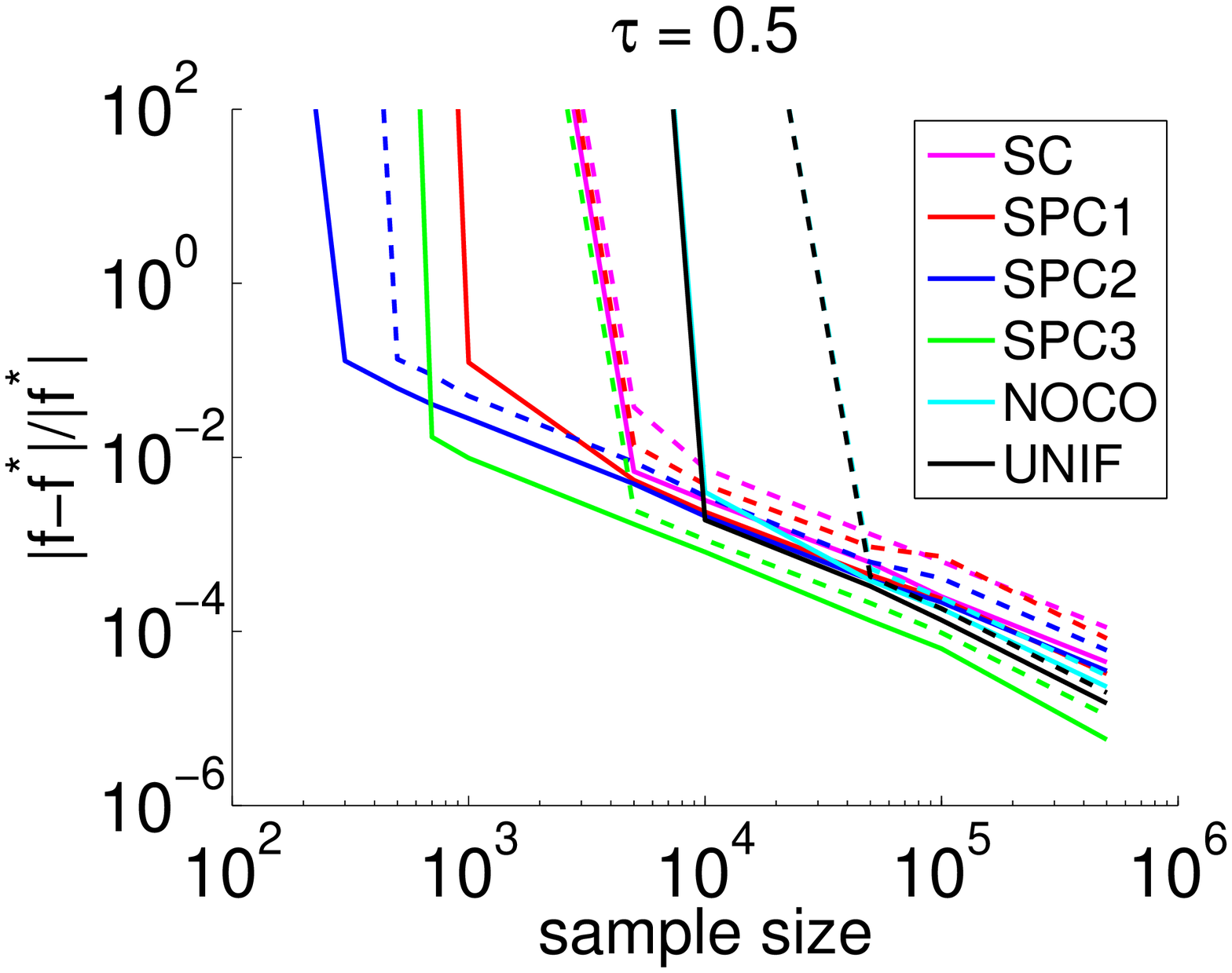}
 }
 &
 \subfigure[$\tau = 0.75$,  $|f-f^*|/|f^*|$]{
   \includegraphics[width=0.3\textwidth] {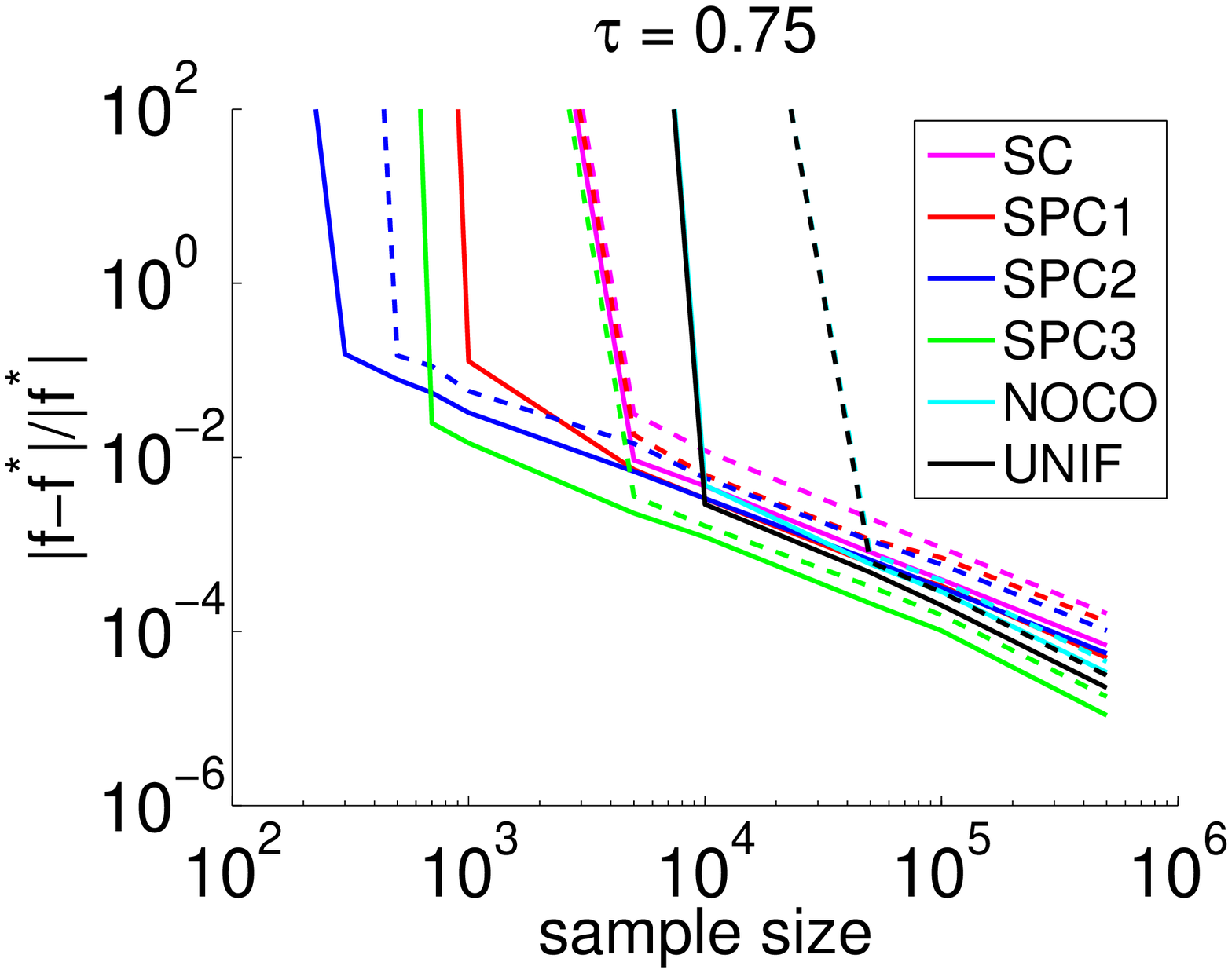}
 }
 &
\subfigure[$\tau = 0.95$,  $|f-f^*|/|f^*|$]{
   \includegraphics[width=0.3\textwidth] {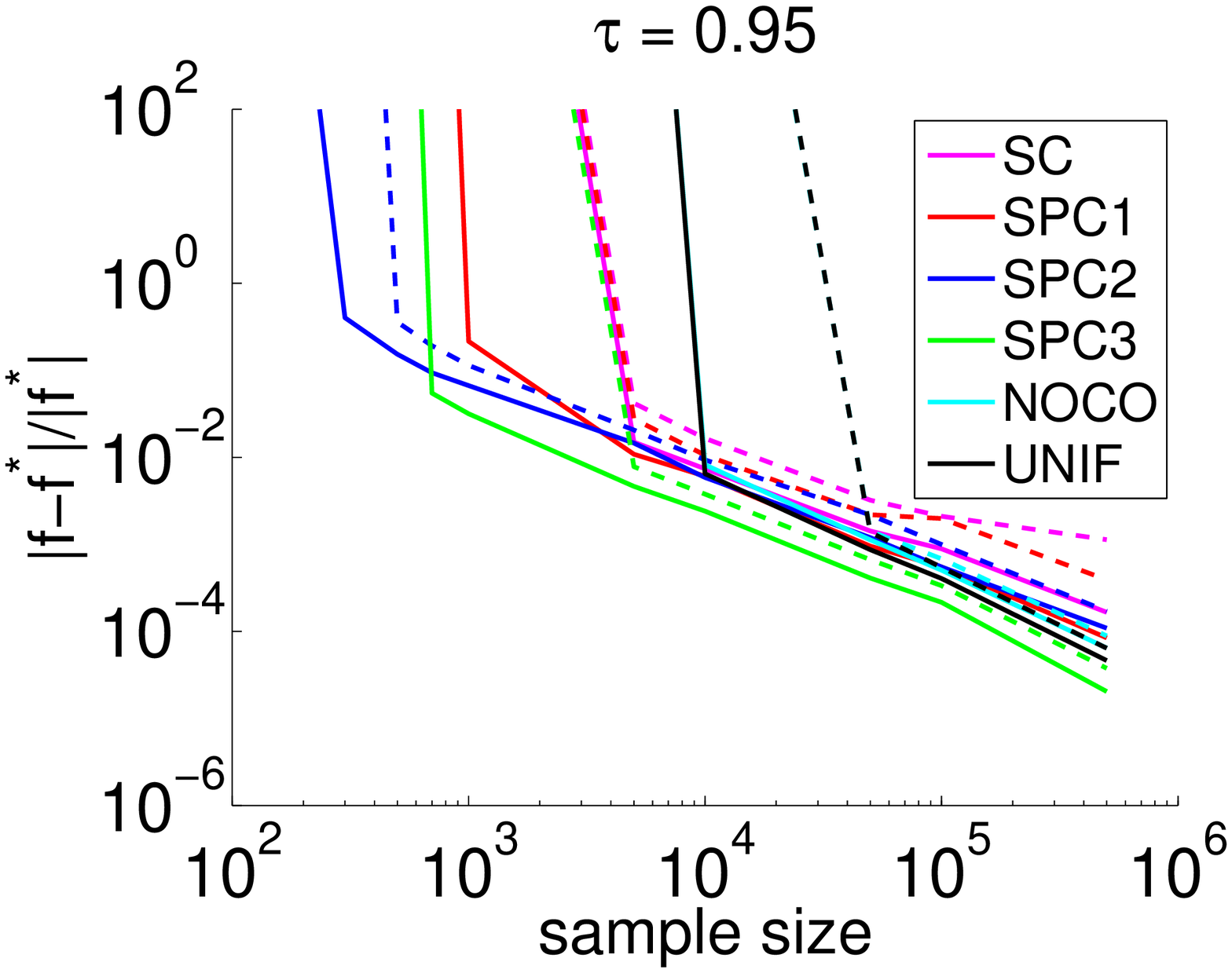}
 }
 \\
 \subfigure[$\tau = 0.5$, $\|x-x^*\|_2/\|x^*\|_2$]{
   \includegraphics[width=0.3\textwidth] {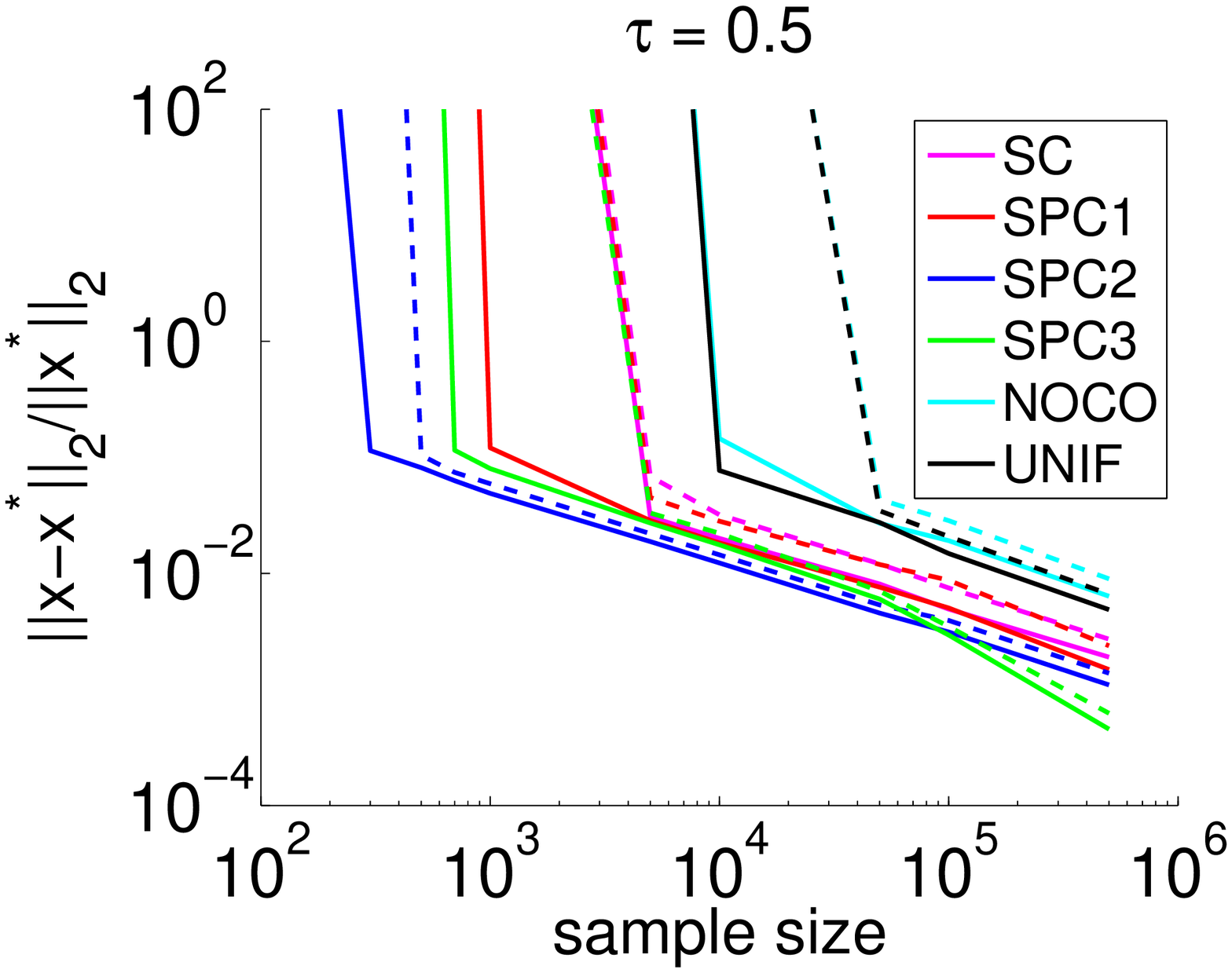}
 }
 &
 \subfigure[$\tau = 0.75$, $\|x-x^*\|_2/\|x^*\|_2$]{
   \includegraphics[width=0.3\textwidth] {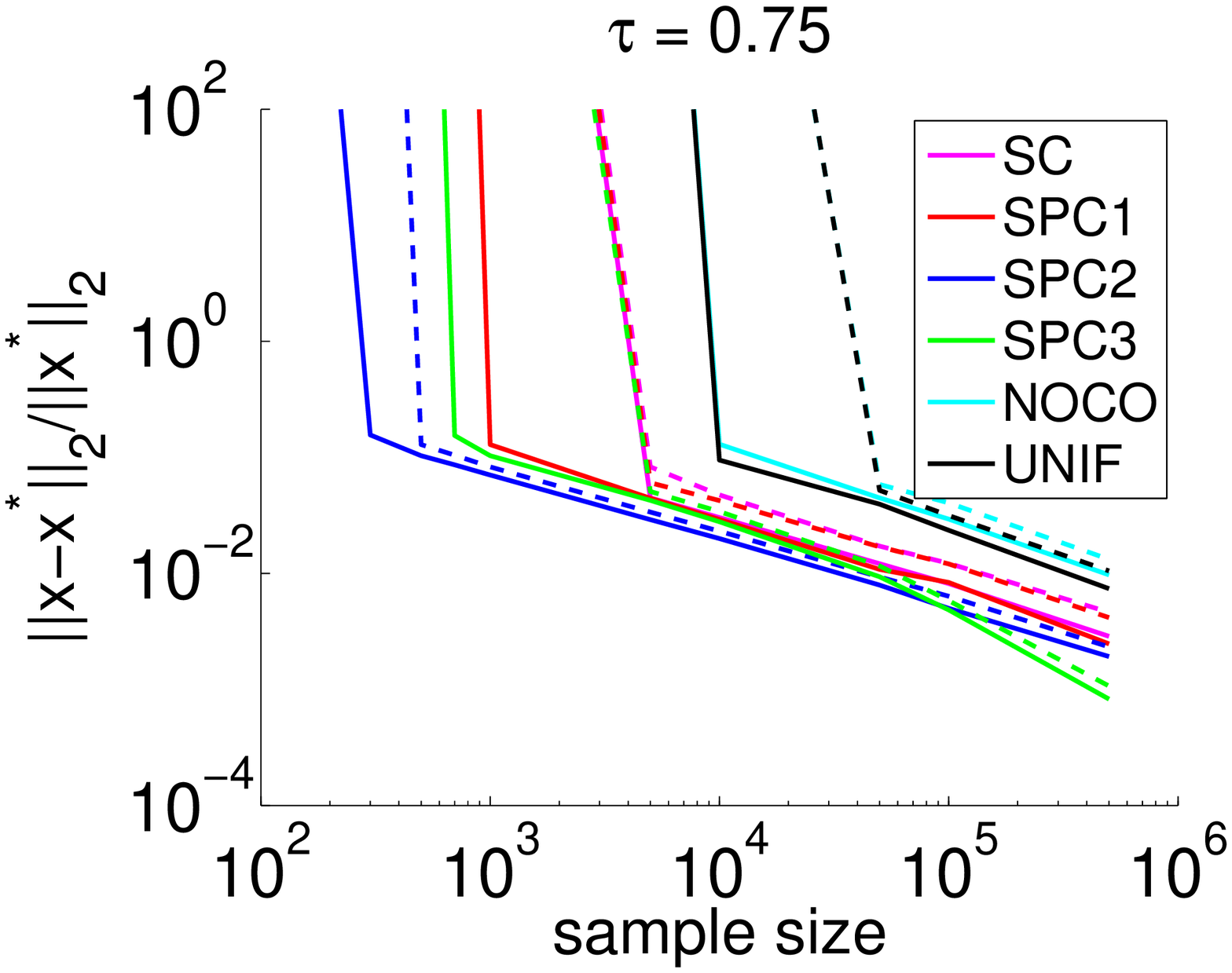}
 }
 &
 \subfigure[$\tau = 0.95$, $\|x-x^*\|_2/\|x^*\|_2$]{
   \includegraphics[width=0.3\textwidth] {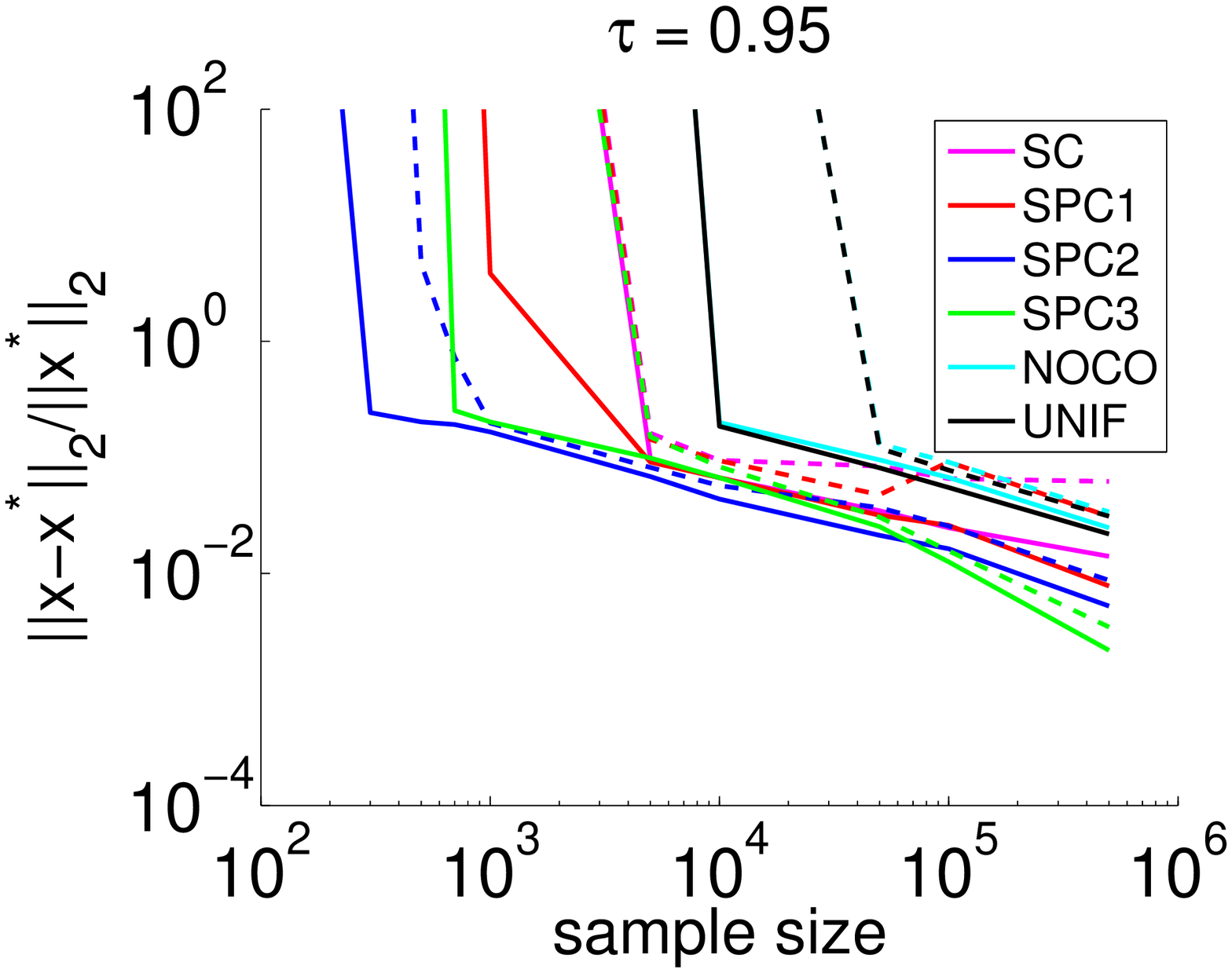}
 }
 \\
\subfigure[$\tau = 0.5$, $\|x-x^*\|_1/\|x^*\|_1$]{
   \includegraphics[width=0.3\textwidth] {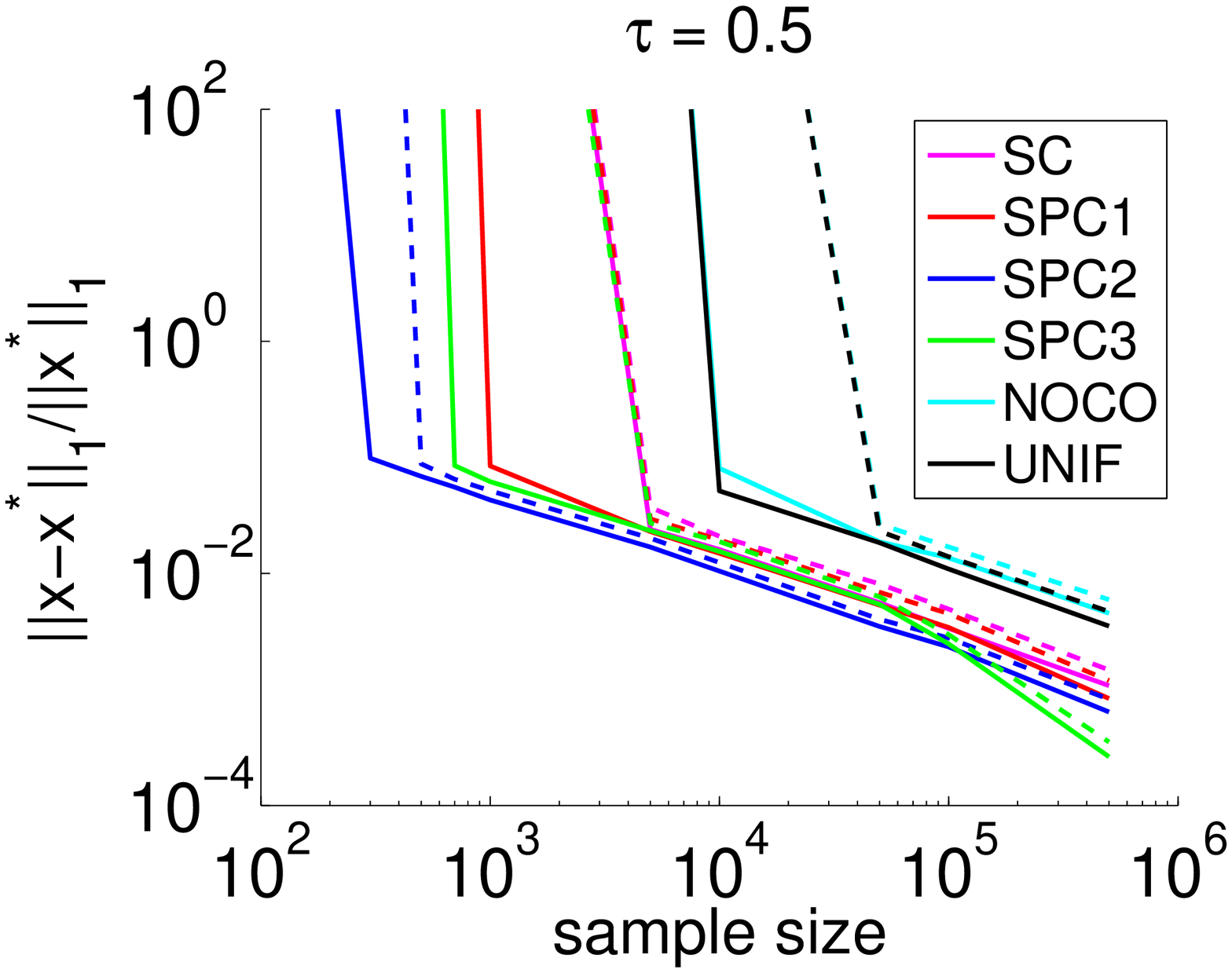}
 }
 &
 \subfigure[$\tau = 0.75$, $\|x-x^*\|_1/\|x^*\|_1$]{
   \includegraphics[width=0.3\textwidth] {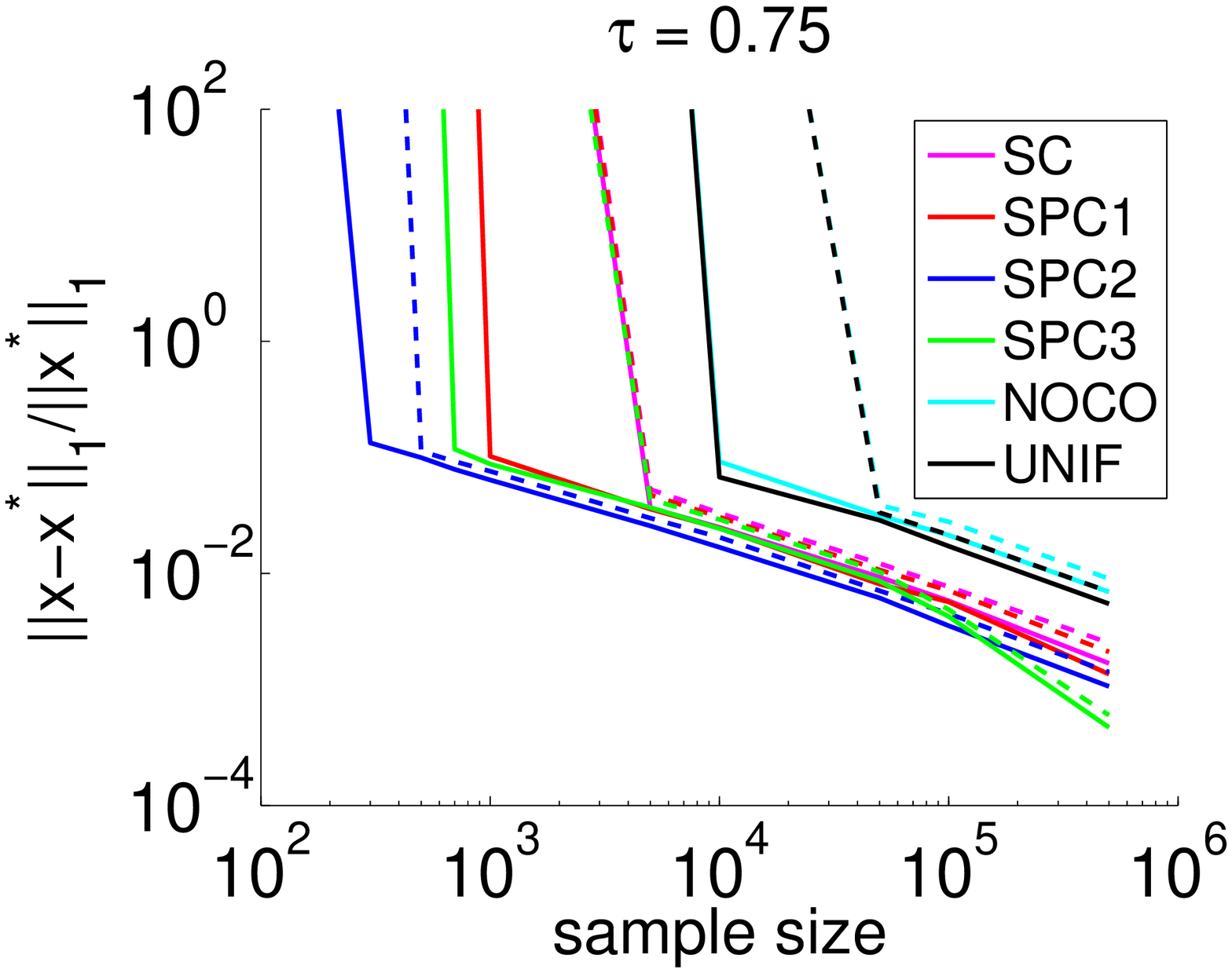}
 }
 &
 \subfigure[$\tau = 0.95$, $\|x-x^*\|_1/\|x^*\|_1$]{
   \includegraphics[width=0.3\textwidth] {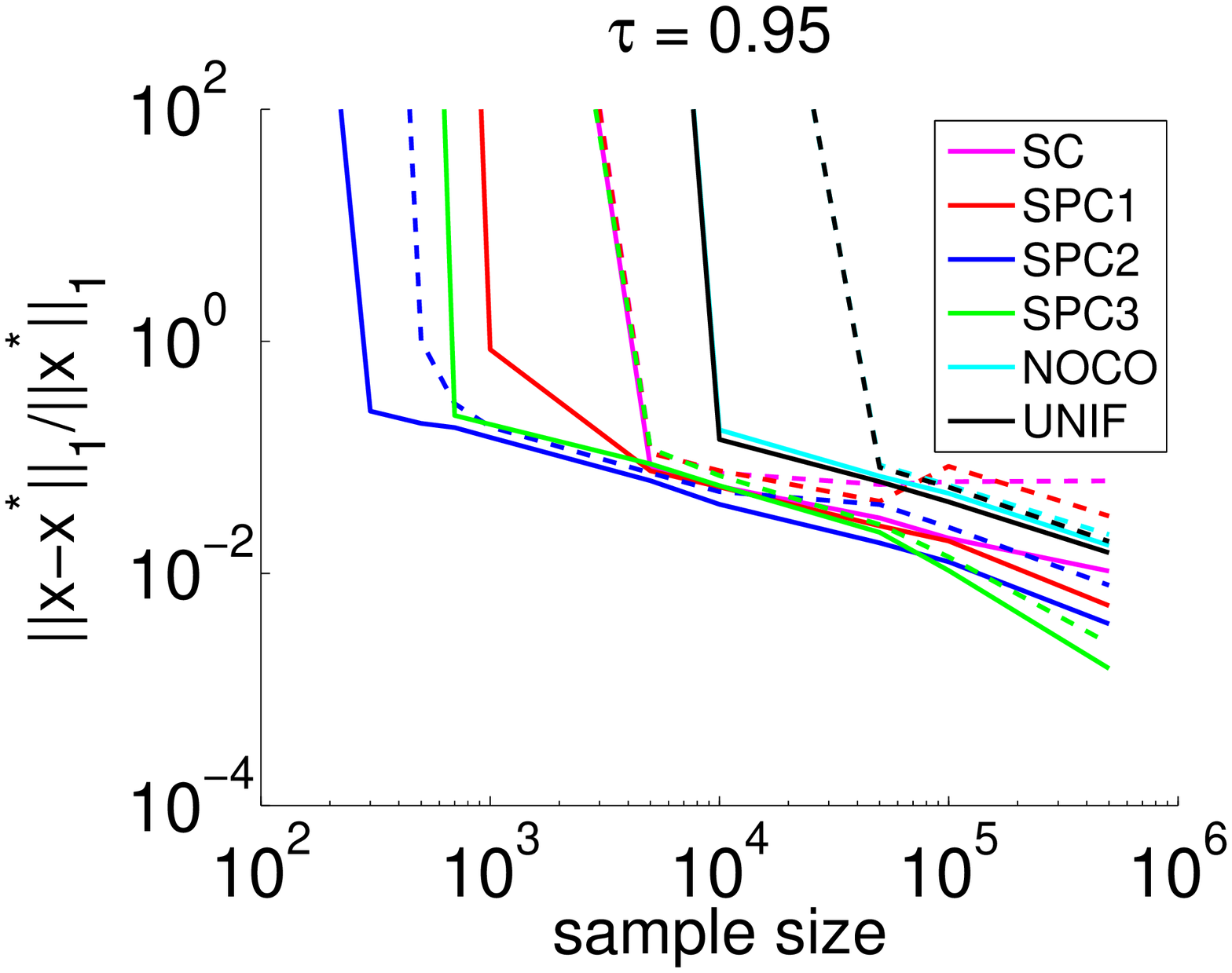}
 }
 \\
 \subfigure[$\tau = 0.5$, $\|x-x^*\|_\infty/\|x^*\|_\infty$]{
   \includegraphics[width=0.3\textwidth] {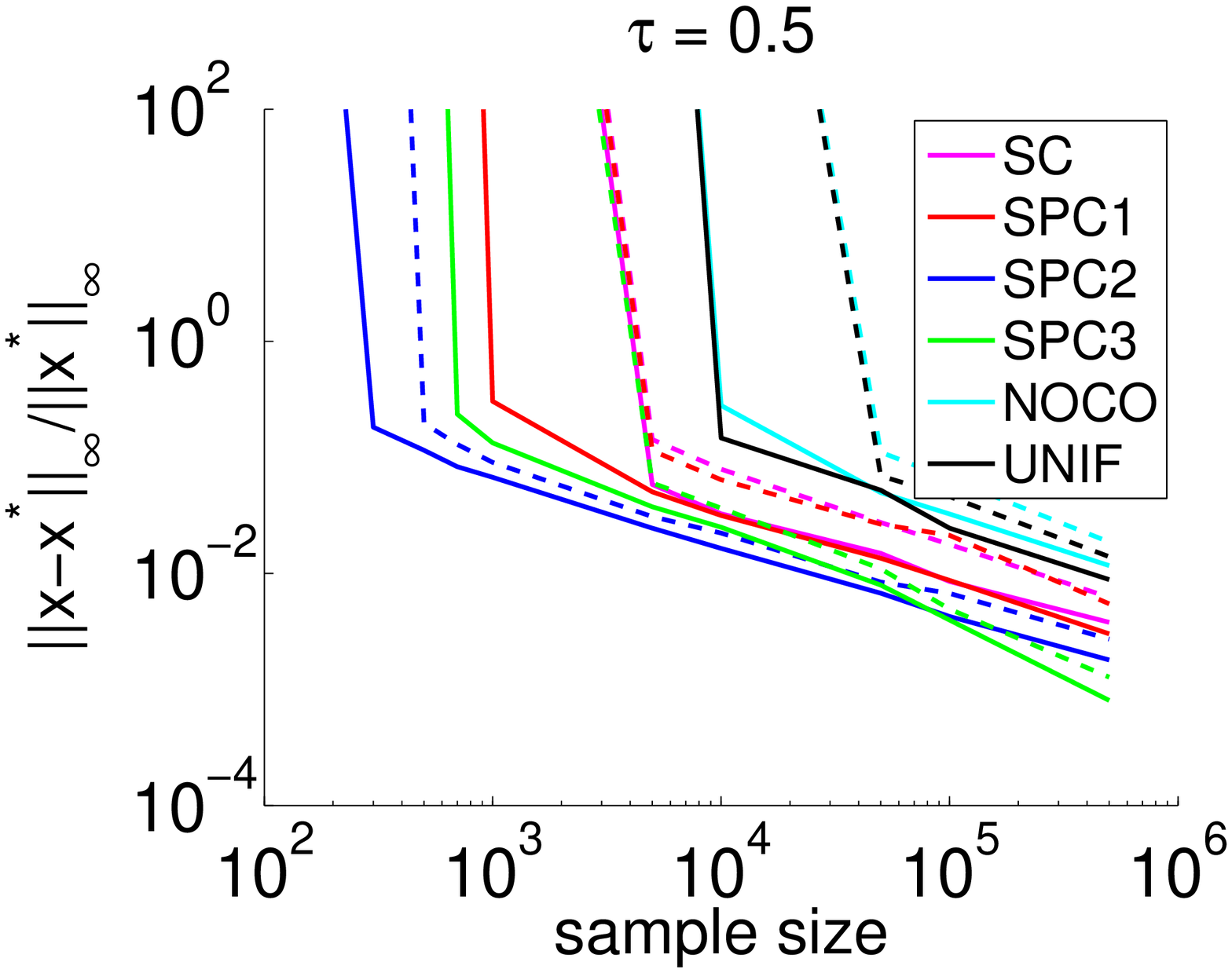}
 } 
 &
\subfigure[$\tau = 0.75$, $\|x-x^*\|_\infty/\|x^*\|_\infty$]{
   \includegraphics[width=0.3\textwidth] {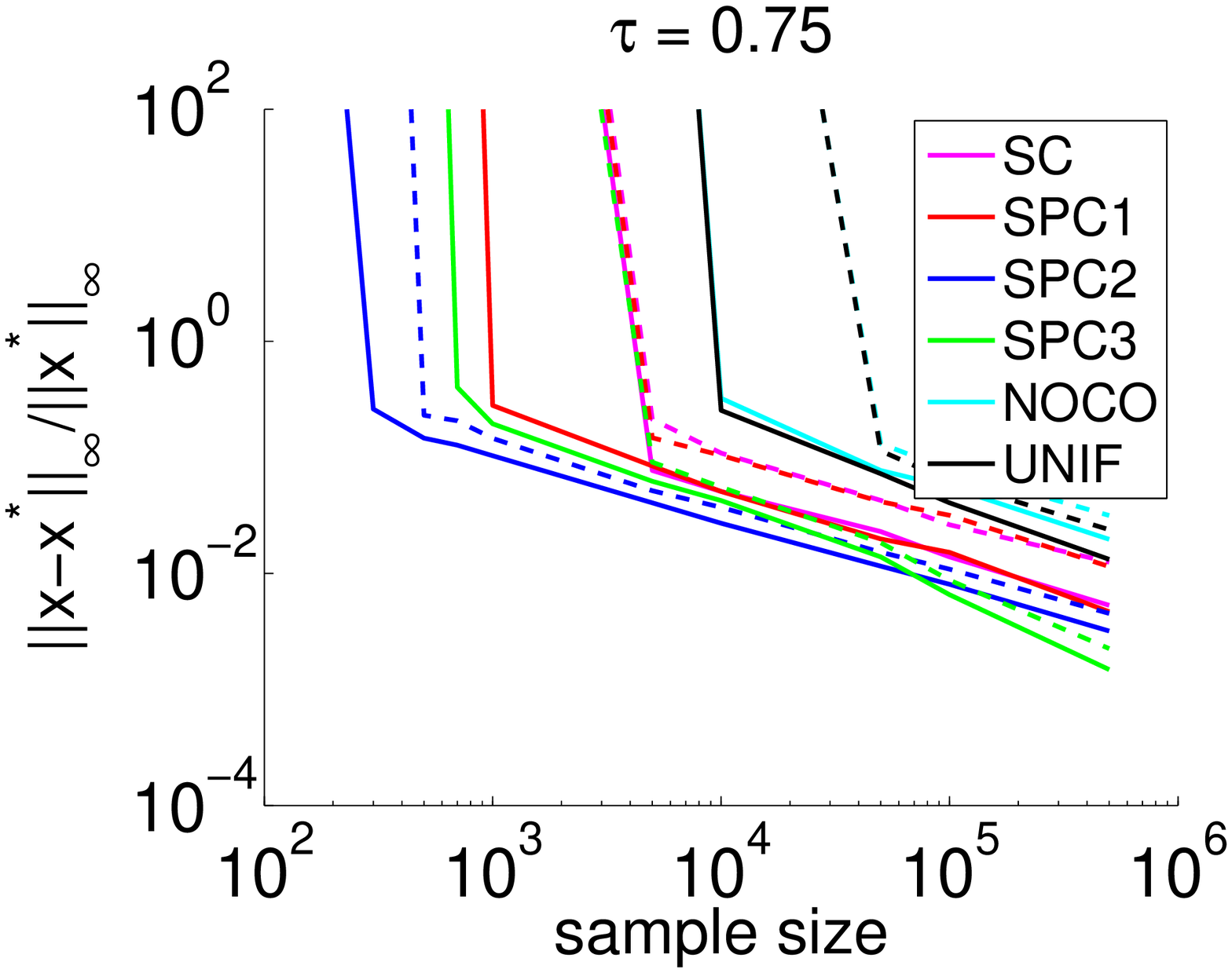}
 }
 &
 \subfigure[$\tau = 0.5$, $\|x-x^*\|_\infty/\|x^*\|_\infty$]{
   \includegraphics[width=0.3\textwidth] {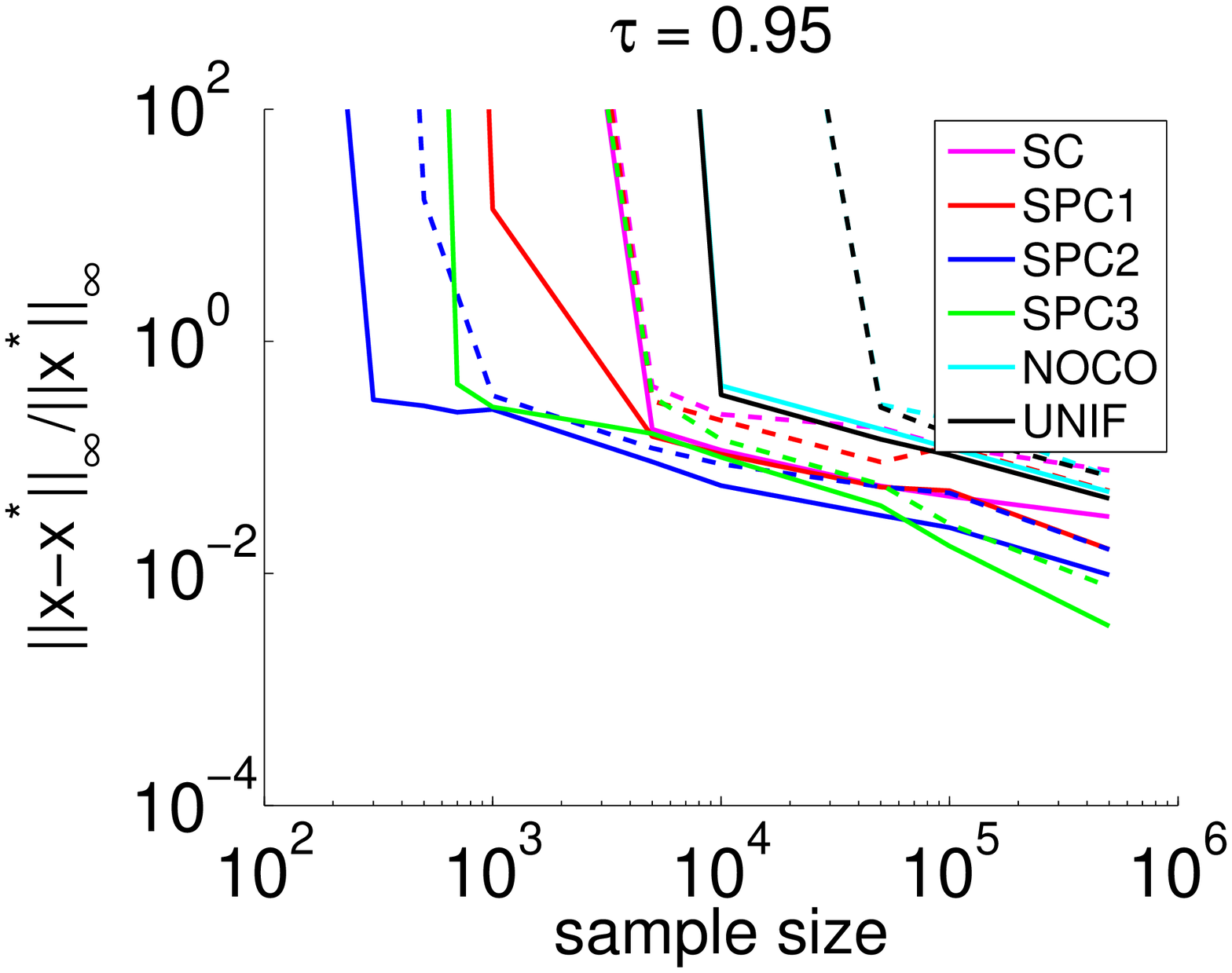}
 }
 \end{tabular}
 \end{center}
 \caption{
  The first (solid lines) and the third (dashed lines) quartiles of the relative errors of the objective value
  (namely, $|f-f^*|/|f^*|$) and solution vector (measured in three different norms, namely, the $\ell_2, \ell_1$ and $\ell_{\infty}$ norms),
  by using 6 different methods, among 50 independent trials.
  The test is on skewed data with size $1e6$ by 50.
   The three different columns correspond to $\tau = 0.5, 0.75, 0.95$, respectively.
    }
   \label{err_s}
\end{figure}

From these plots, if we look at the sampling size required for generating at 
least 1-digit accuracy, then SPC2 needs the fewest samples, followed by 
SPC3, and then SPC1.
This is consistent with the order of the condition numbers of these methods.
For SC, although in theory it has good condition number properties, in practice it 
performs worse than other methods.
Not surprisingly,  NOCO and UNIF are not reliable when $s$ is very small, 
\emph{e.g.}, less than $1e4$.

When the sampling size $s$ is large enough, the accuracy of each conditioning method is 
close to the others in terms of the objective value. 
Among these, SPC3 performs slightly better than others.
When estimating the actual solution vectors, the conditioning-based methods 
behave substantially better than the two naive methods.
SPC2 and SPC3 are the most reliable methods since they can yield the least 
relative error for every sample size $s$.
NOCO is likely to sample the outliers, and UNIF cannot get accurate answer 
until the sampling size $s \ge 1e4$.
This accords with our expectations.
For example, when $s$ is less than $1e4$, as we pointed out in the remark 
below the description of the skewed data, it is very likely that none of 
the rows in the first block corresponding to the first coordinate will be 
selected.
Thus, poor estimation will be generated due to the imbalanced measurements in 
the design matrix.
Note that from the plots we can see that if a method fails with some 
sampling complexity $s$, then for that value of $s$ the relative errors will 
be huge (\emph{e.g.}, larger than $100$, which is clearly a trivial result).
Note also that all the methods can generate at least 1-digit accuracy if $s$ 
is large enough.

It is worth mentioning the performance difference among SPC1, SPC2 and SPC3.
From Table~\ref{cond_table}, we show the tradeoff between running time and 
condition number for the three methods.
As we pointed out, SPC2 always needs the least sampling complexity to 
generate 2-digit accuracy, followed by SPC3 and then SPC1.
When $s$ is large enough, SPC2 and SPC3 perform substantially better than 
SPC1.
As for the running time, SPC1 is the fastest, followed by SPC3, and then SPC2.
Again, all of these follow the theory about our SPC methods.
We will present a more detailed discussion for the running time in 
Section~\ref{running_time}.

Although our theory doesn't say anything about the quality of the solution
vector itself (as opposed to the value of the objective function), we 
evaluate this here.
To measure the approximation to the solution vectors, we use three norms 
(the $\ell_1$, $\ell_2$, and $\ell_{\infty}$ norms). 
From Figure~\ref{err_s}, we see that the performance among these method is 
qualitatively similar for each of the three norms, but the relative error 
is higher when measured in the $\ell_{\infty}$ norm.
In more detail, see Table~\ref{sol_table}, where we show the exact quartiles 
of the relative error on vectors for each methods for $s = 5e4$ and 
$\tau = 0.75$. 
Not surprisingly, NOCO and UNIF are not among the reliable methods when $s$ 
is small (and they get worse when $s$ is even smaller).
Note that the relative error for each method doesn't change substantially 
when $\tau$ takes different values.
We present a more detailed discussion of the $\tau$ dependence 
in Section~\ref{relerr_tau}.

(We note also that, for subsequent figures in subsequent subsections, we 
obtained similar qualitative trends for the errors in the approximate 
solution vectors when the errors were measured in different norms.
Thus, due to this similarity and to save space, in subsequent figures, we 
will only show errors for $\ell_2$ norm.)

\begin{table}[ht]
\begin{center}
\begin{sc}
\small
\begin{tabular}{c|ccc}
   &  $\|x - x^*\|_2/\|x^*\|_2$ & $\|x - x^*\|_1/\|x^*\|_1$ & $\|x - x^*\|_\infty/\|x^*\|_\infty$ \\
\hline
SC & [0.0121, 0.0172] & [0.0093, 0.0122] & [0.0229, 0.0426]  \\
SPC1 & [0.0108,  0.0170]  &  [0.0081, 0.0107] & [0.0198, 0.0415] \\
SPC2 & [0.0079,  0.0093]  &  [0.0061, 0.0071] & [0.0115, 0.0152] \\
SPC3 & [0.0094,  0.0116]  &  [0.0086, 0.0103] & [0.0139, 0.0184] \\
NOCO & [0.0447,  0.0583] &  [0.0315, 0.0386] & [0.0769, 0.1313] \\
UNIF  &  [0.0396,  0.0520]  &  [0.0287, 0.0334] & [0.0723, 0.1138]
\end{tabular}
\end{sc}
\end{center}
\caption{The first and the third quartiles of relative errors of the 
solution vector, measured in $\ell_1$, $\ell_2$, and $\ell_\infty$ norms.
The test data set is the skewed data, with size $1e6 \times 50$, the 
sampling size $s = 5e4$, and $\tau = 0.75$.
}
\label{sol_table}
\end{table}


\subsection{Quality of approximation when the higher dimension $n$ changes}
\label{relerr_n}

Next, we describe how the performance of our algorithm varies when higher 
dimension $n$ changes.  
(We present the results when the lower dimension $d$ changes in 
Section~\ref{relerr_d}.)
Figures~\ref{comp_n} and~\ref{err_n} summarize our results.

Figure~\ref{comp_n} shows the performance of the relative error of the 
objective value and solution vector by using the six different methods, 
as $n$ is varied, for fixed values of $\tau = 0.75$ and $d = 50$. 
For each row, the three figures come from three data sets with $n$ taking 
value in $1e5, 5e5, 1e6$.
(Recall that, in these experiments, we only list the plots showing the 
relative error on vectors measured in $\ell_2$ norm.
Since the plots for the $\ell_1$ and $\ell_\infty$ norm are similar, we 
omit them.)
We see that, when $d$ is fixed, the basic structure in the plots that we observed before is 
preserved when $n$ takes three different values.
In particular, the minimum sampling complexity $s$ needed for each method 
for yielding high accuracy does not vary a lot.
When $s$ is large enough, the relative performance among all the methods 
is similar; and, when all the parameters are fixed except for $n$, the 
relative error for each method does not change quantitatively.

\begin{figure}[h!tbp]
 \begin{center}
 \begin{tabular}{ccc}
 \subfigure[$1e5 \times 50$, $|f-f^*|/|f^*|$]{
   \includegraphics[width=0.3\textwidth] {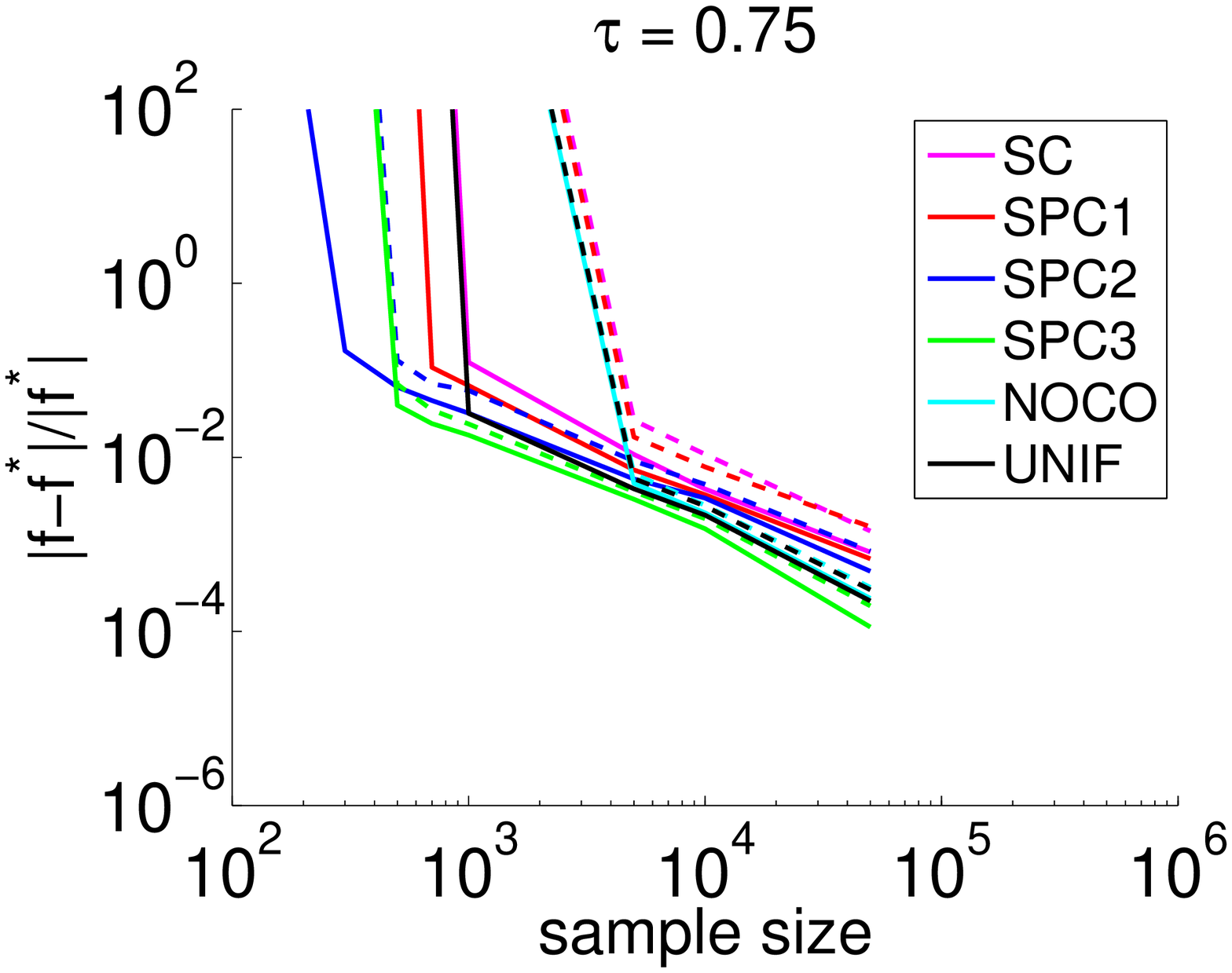}
 } 
 &
   \subfigure[$5e5 \times 50$, $|f-f^*|/|f^*|$]{
   \includegraphics[width=0.3\textwidth] {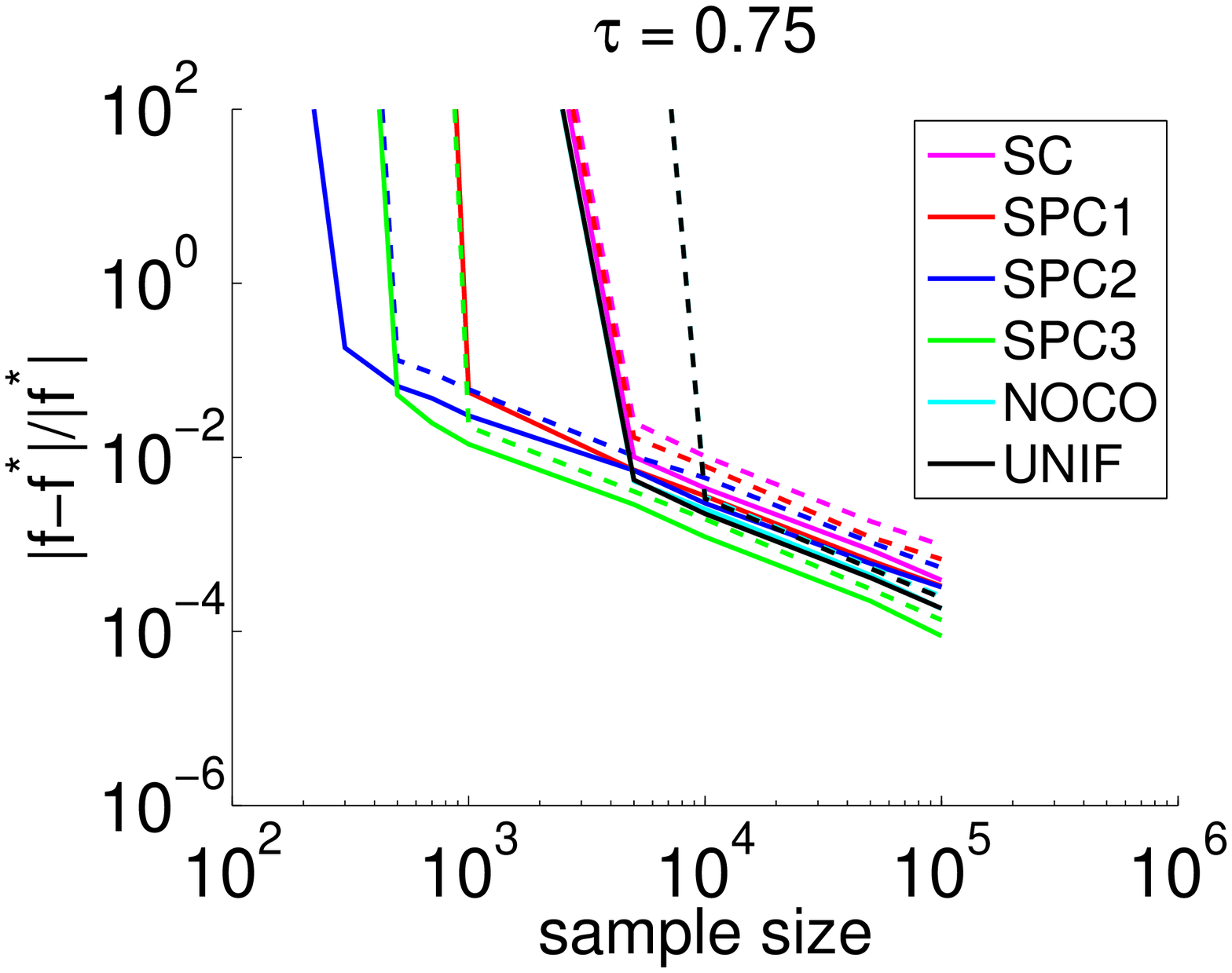}
 }
 &
  \subfigure[$1e6 \times 50$, $|f-f^*|/|f^*|$]{
   \includegraphics[width=0.3\textwidth] {FIG/err_s/quartile11_2.eps}
 }
 \\
  \subfigure[$1e5 \times 50$, $\|x-x^*\|_2/\|x^*\|_2$]{
   \includegraphics[width=0.3\textwidth] {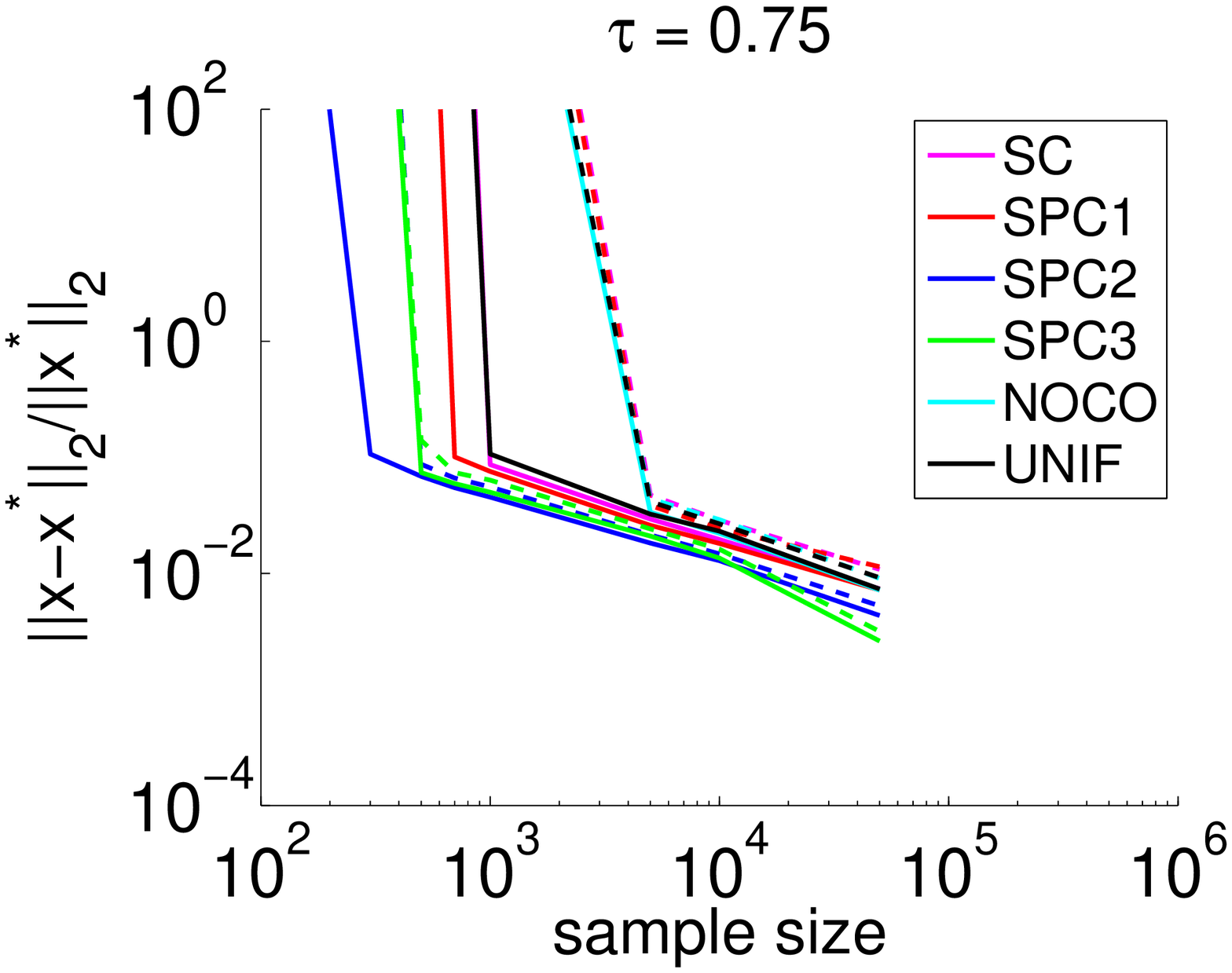}
 } 
 &
 \subfigure[$5e5 \times 50$, $\|x-x^*\|_2/\|x^*\|_2$]{
   \includegraphics[width=0.3\textwidth] {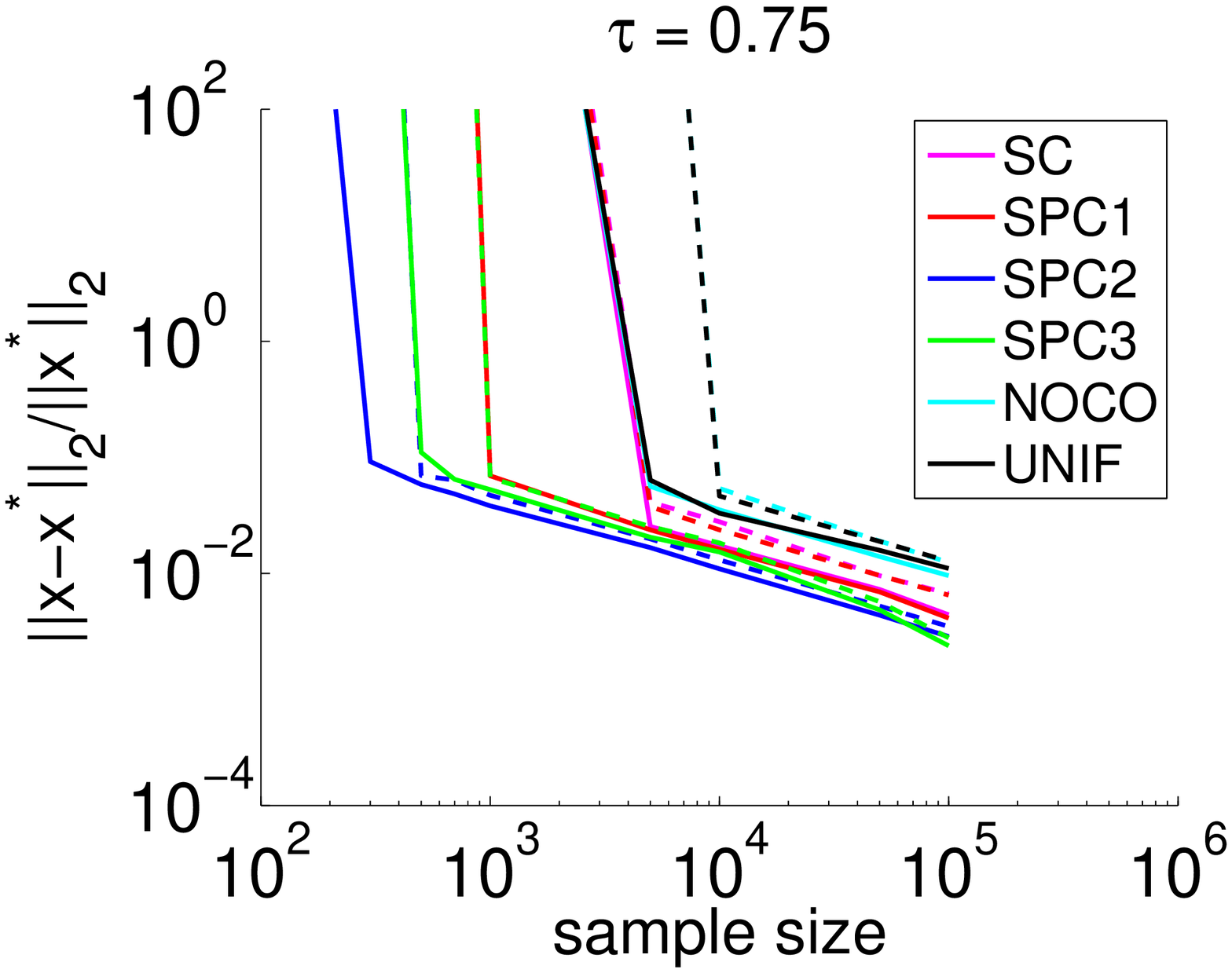}
 }
 &
   \subfigure[$1e6 \times 50$, $\|x-x^*\|_2/\|x^*\|_2$]{
   \includegraphics[width=0.3\textwidth] {FIG/err_s/quartile11_5.eps}
 }
 \end{tabular}
 \end{center}
  \caption{
   The first (solid lines) and the third (dashed lines) quartiles of the relative errors of the objective value
   (namely, $|f-f^*|/|f^*|$) and solution vector (namely, $\|x-x^*\|_2/\|x^*\|_2$),
   when the sample size $s$ changes, for different values of $n$, while $d=50$ by using 6 different methods, among 50 independent trials.
   The test is on skewed data and $\tau = 0.75$.
   The three different columns correspond to $n = 1e5, 5e5, 1e6$, respectively.
   }
  \label{comp_n}
\end{figure}

We will also let $n$ take a wider range of values. 
Figure~\ref{err_n} shows the change of relative error on the objective value 
and solution vector by using SPC3 and letting $n$ vary from $1e4$ to $1e6$ 
and $d = 50$ fixed.
Recall, from Theorem~\ref{qr_thm}, that for given a tolerance $\epsilon$, 
the required sampling complexity $s$ depends only on $d$.
That is, if we fix the sampling size $s$ and $d$, then the relative error 
should not vary much, as a function of $n$.
If we inspect Figure~\ref{err_n}, we see that the relative errors 
are almost constant as a function of increasing $n$, provided that $n$ is
much larger than $s$.
When $s$ is very close to $n$, since we are sampling roughly the same number 
of rows as in the full data, we should expect lower errors.
Also, we can see that by using SPC3, relative errors remain roughly the same in
magnitude.

\begin{figure}[h!tpb]
 \begin{center}
 \begin{tabular}{ccc}
\subfigure[$\tau = 0.5$, $|f-f^*|/|f^*|$]{
   \includegraphics[width=0.3\textwidth] {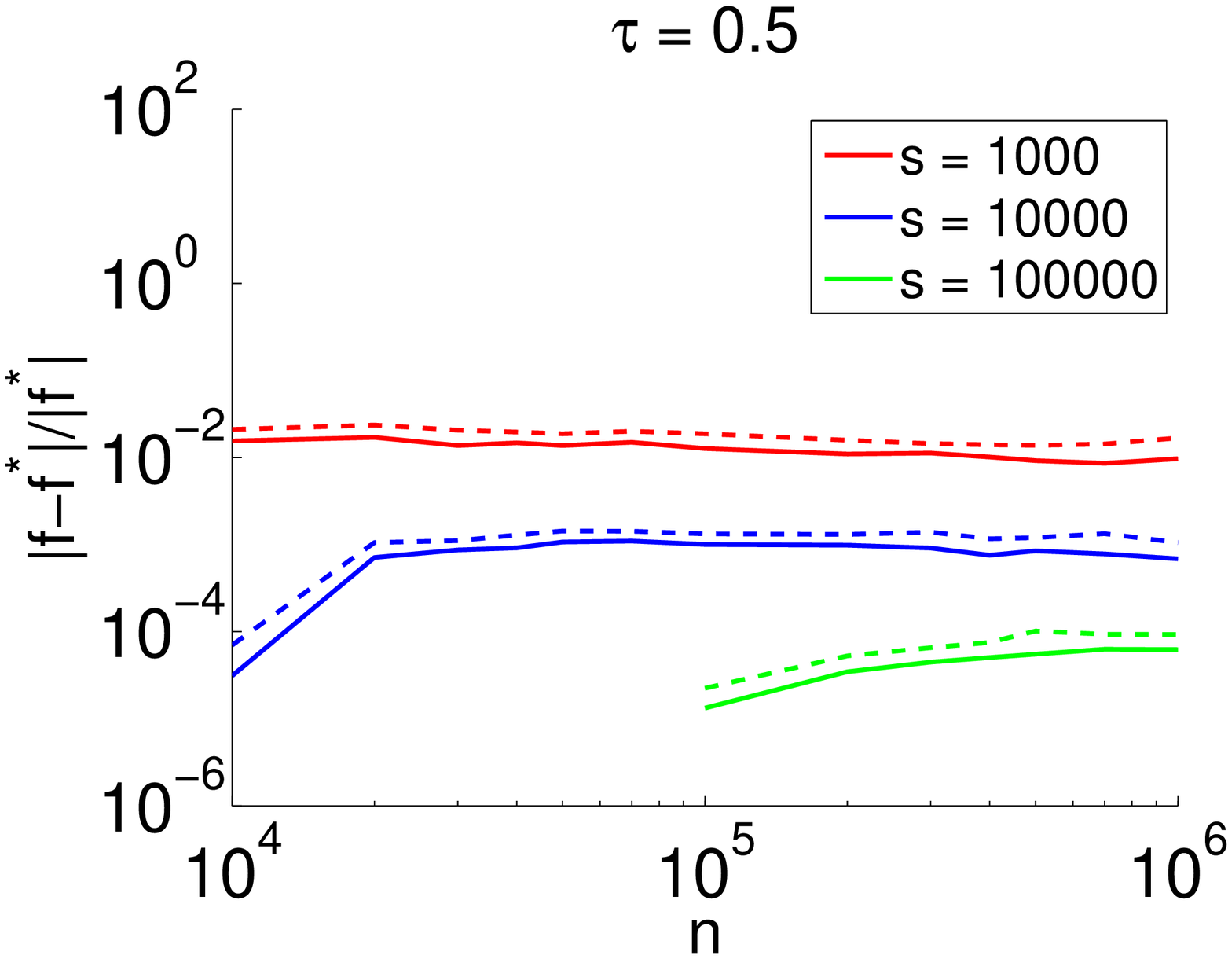}
 }
 &
 \subfigure[$\tau = 0.75$,  $|f-f^*|/|f^*|$]{
   \includegraphics[width=0.3\textwidth] {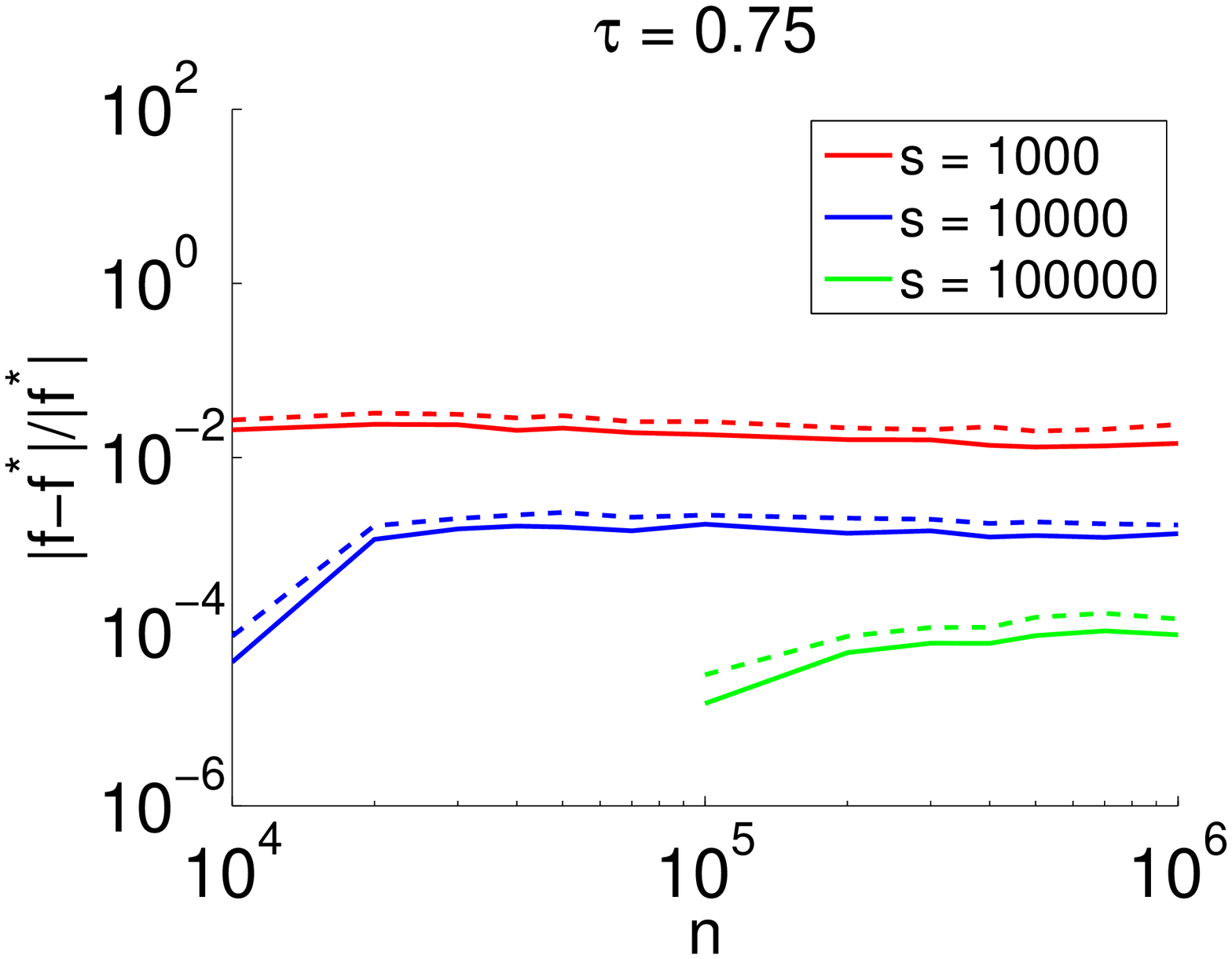}
 }
 &
\subfigure[$\tau = 0.95$,  $|f-f^*|/|f^*|$]{
   \includegraphics[width=0.3\textwidth] {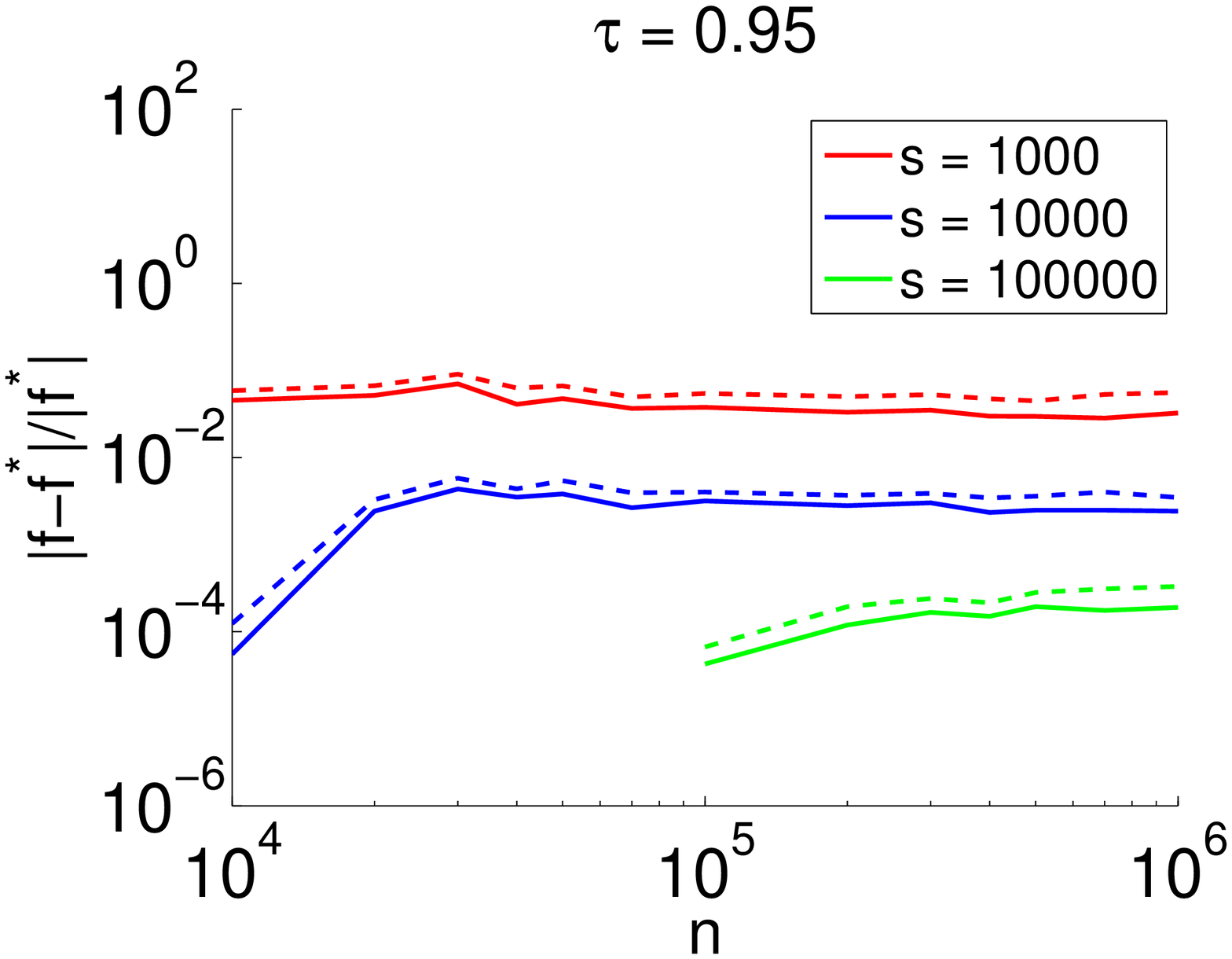}
 }
 \\
 \subfigure[$\tau = 0.5$, $\|x-x^*\|_2/\|x^*\|_2$]{
   \includegraphics[width=0.3\textwidth] {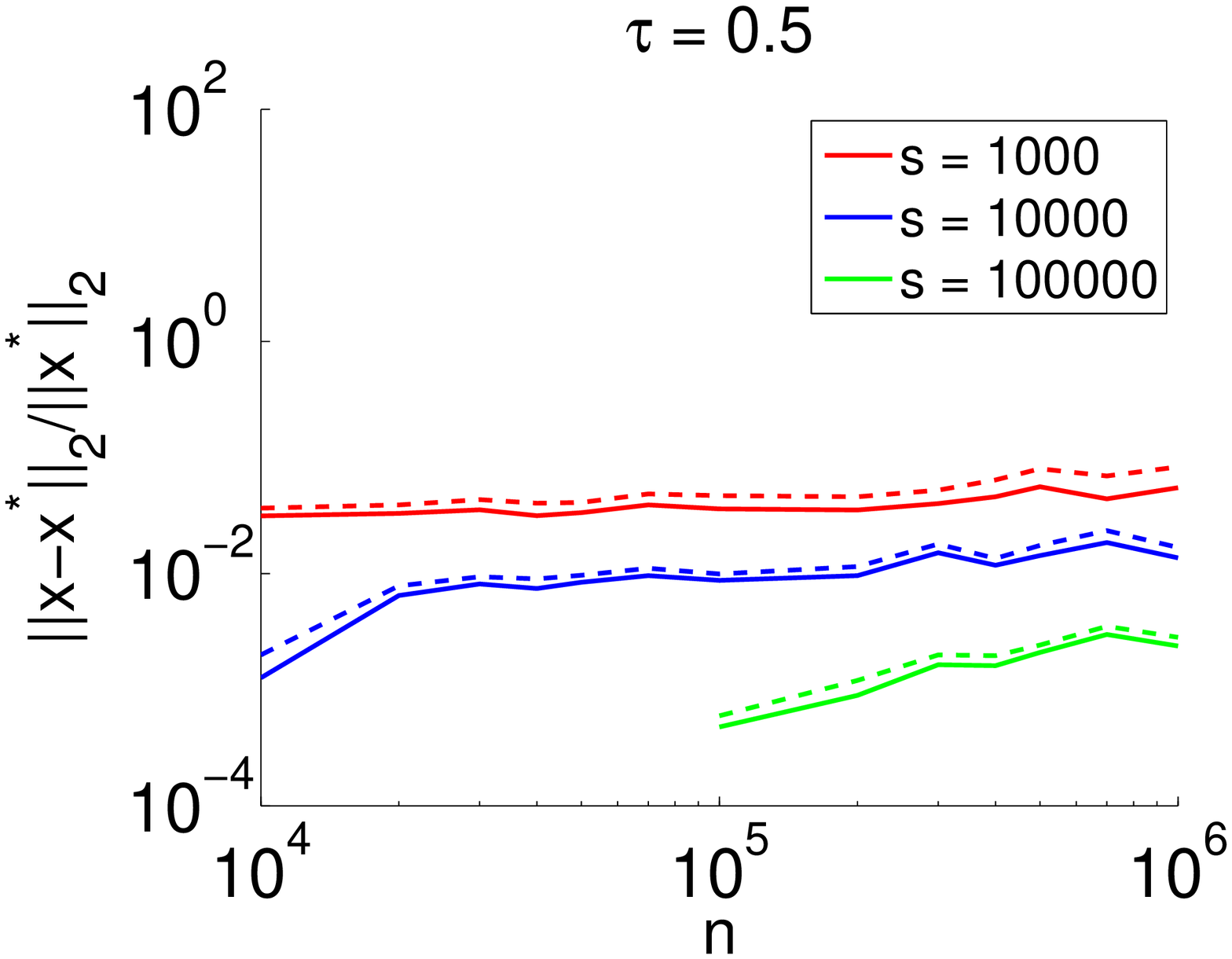}
 }
 &
 \subfigure[$\tau = 0.75$, $\|x-x^*\|_2/\|x^*\|_2$]{
   \includegraphics[width=0.3\textwidth] {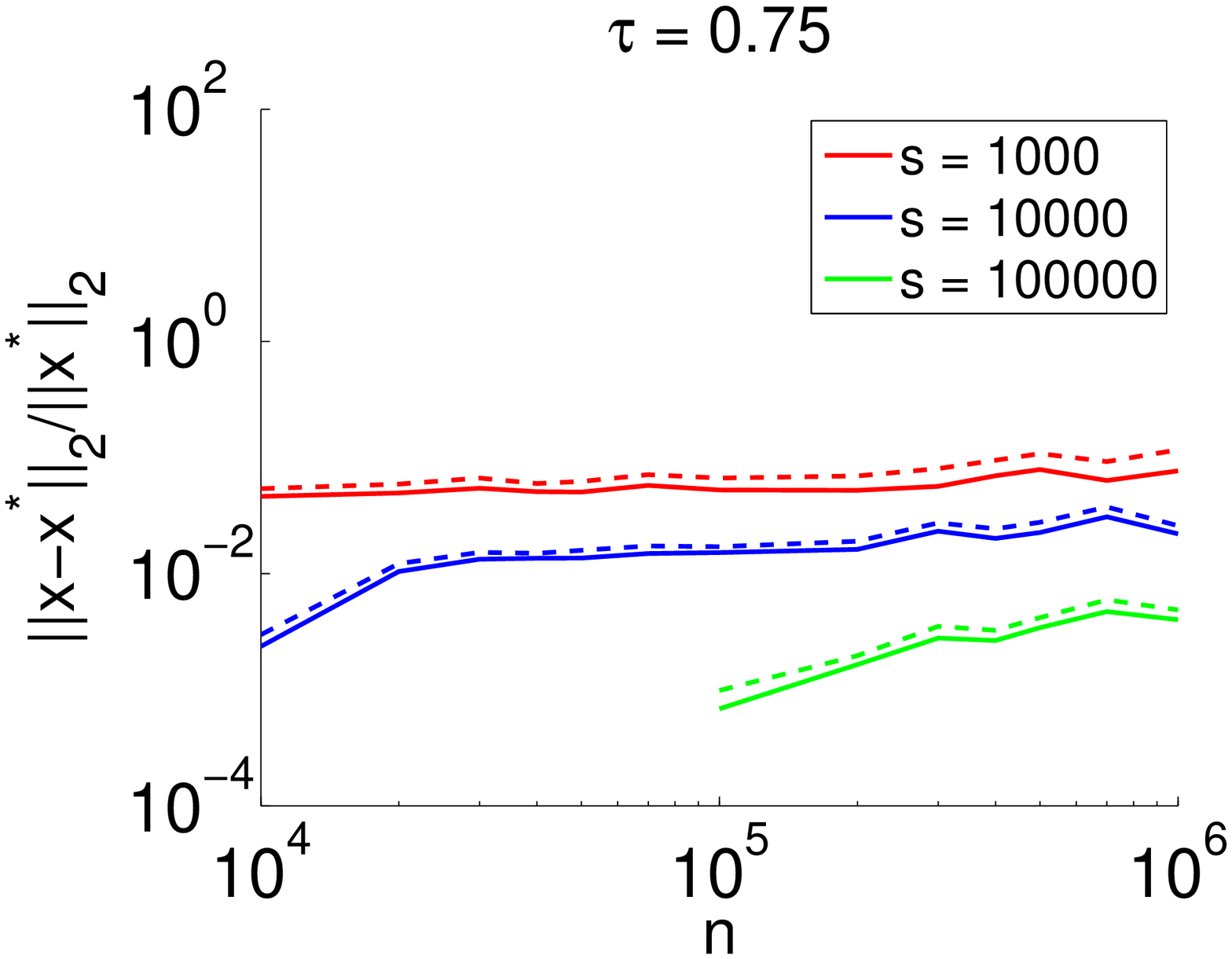}
 }
 &
 \subfigure[$\tau = 0.95$, $\|x-x^*\|_2/\|x^*\|_2$]{
   \includegraphics[width=0.3\textwidth] {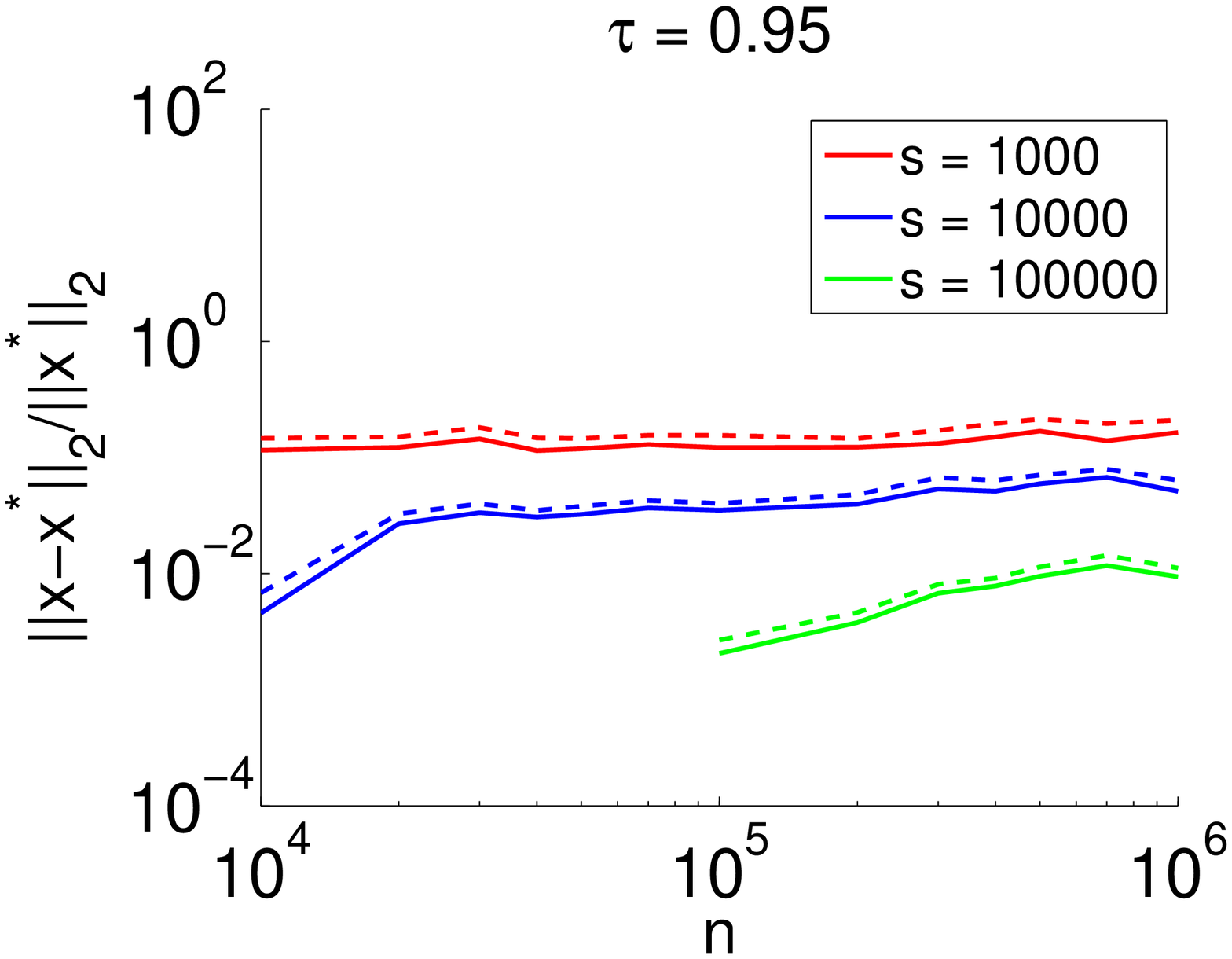}
 }
 \end{tabular}
 \end{center}
 \caption{
   The first (solid lines) and the third (dashed lines) quartiles of the relative errors of the objective value
   (namely, $|f-f^*|/|f^*|$) and solution vector (namely, $\|x-x^*\|_2/\|x^*\|_2$),
   when $n$ varying from $1e4$ to $1e6$ and $d = 50$ by using SPC3, among 50 independent trials.
   The test is on skewed data.
   The three different columns correspond to $\tau = 0.5, 0.75, 0.95$, respectively.
  }
 \label{err_n}
\end{figure}


\subsection{Quality of approximation when the lower dimension $d$ changes}
\label{relerr_d}

Next, we describe how the overall performance changes when the lower 
dimension $d$ changes.
Figures~\ref{comp_d} and~\ref{err_d} summarize our results.
These figures show the same quantities that were plotted in the previous subsection, 
except that here it is the lower dimension $d$ that is now changing, and 
the higher dimension $n = 1e6$ is fixed. 
In Figure~\ref{comp_d}, we let $d$ take values in $10, 50, 100$, 
we set $\tau = 0.75$, and we show the relative error for all 6 conditioning methods.
In Figure~\ref{err_d}, we let $d$ take more values in the range of 
$[10, 100]$, and we show the relative errors by using SPC3 for different 
sampling sizes $s$ and $\tau$~values.

\begin{figure}[h!tbp]
 \begin{center}
 \begin{tabular}{ccc}
 \subfigure[$1e6 \times 10$, $|f-f^*|/|f^*|$]{
   \includegraphics[width=0.3\textwidth] {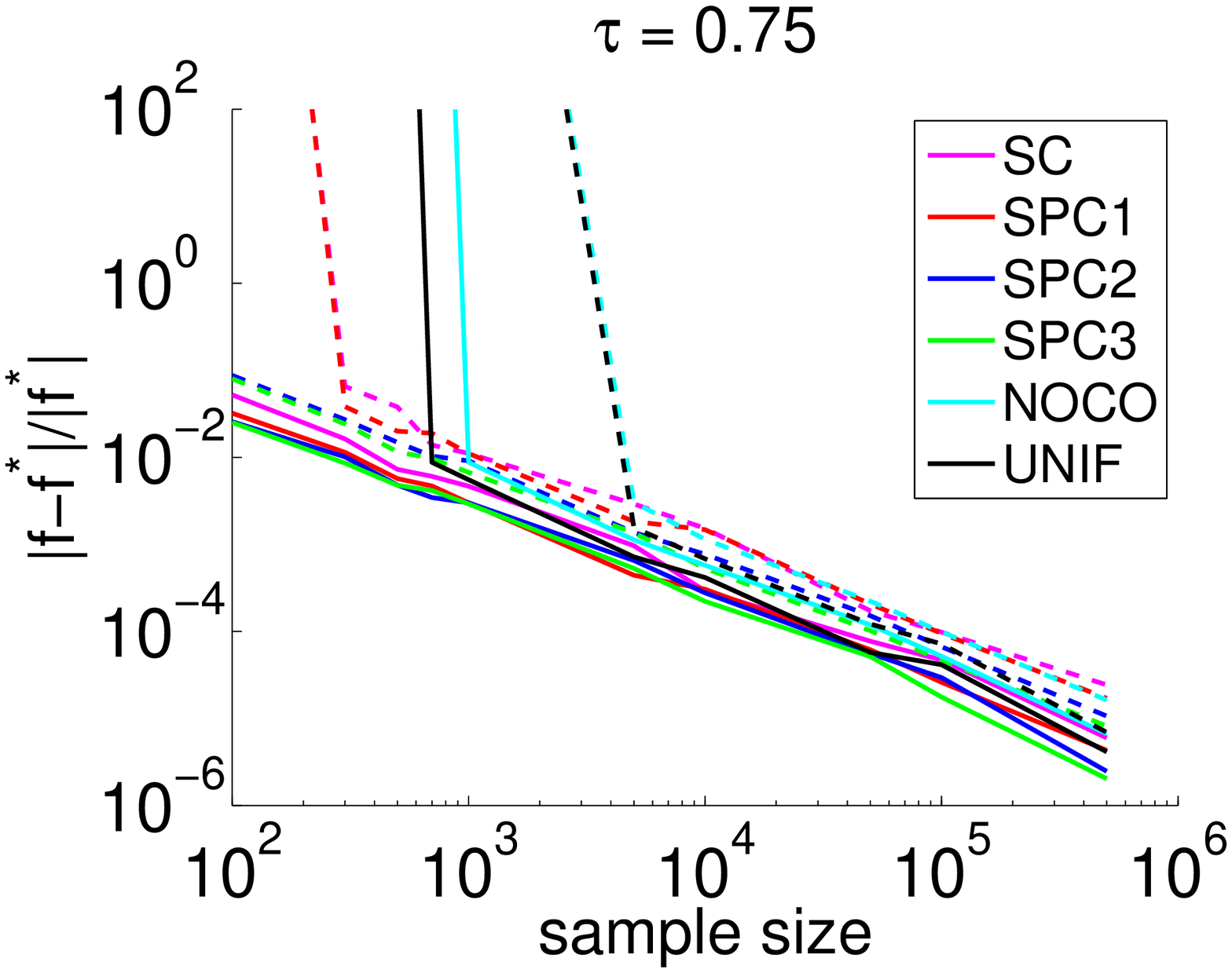}
 } 
 &
   \subfigure[$1e6 \times 50$, $|f-f^*|/|f^*|$]{
   \includegraphics[width=0.3\textwidth] {FIG/err_s/quartile11_2.eps}
 }
 &
   \subfigure[$1e6 \times 100$, $|f-f^*|/|f^*|$]{
   \includegraphics[width=0.3\textwidth] {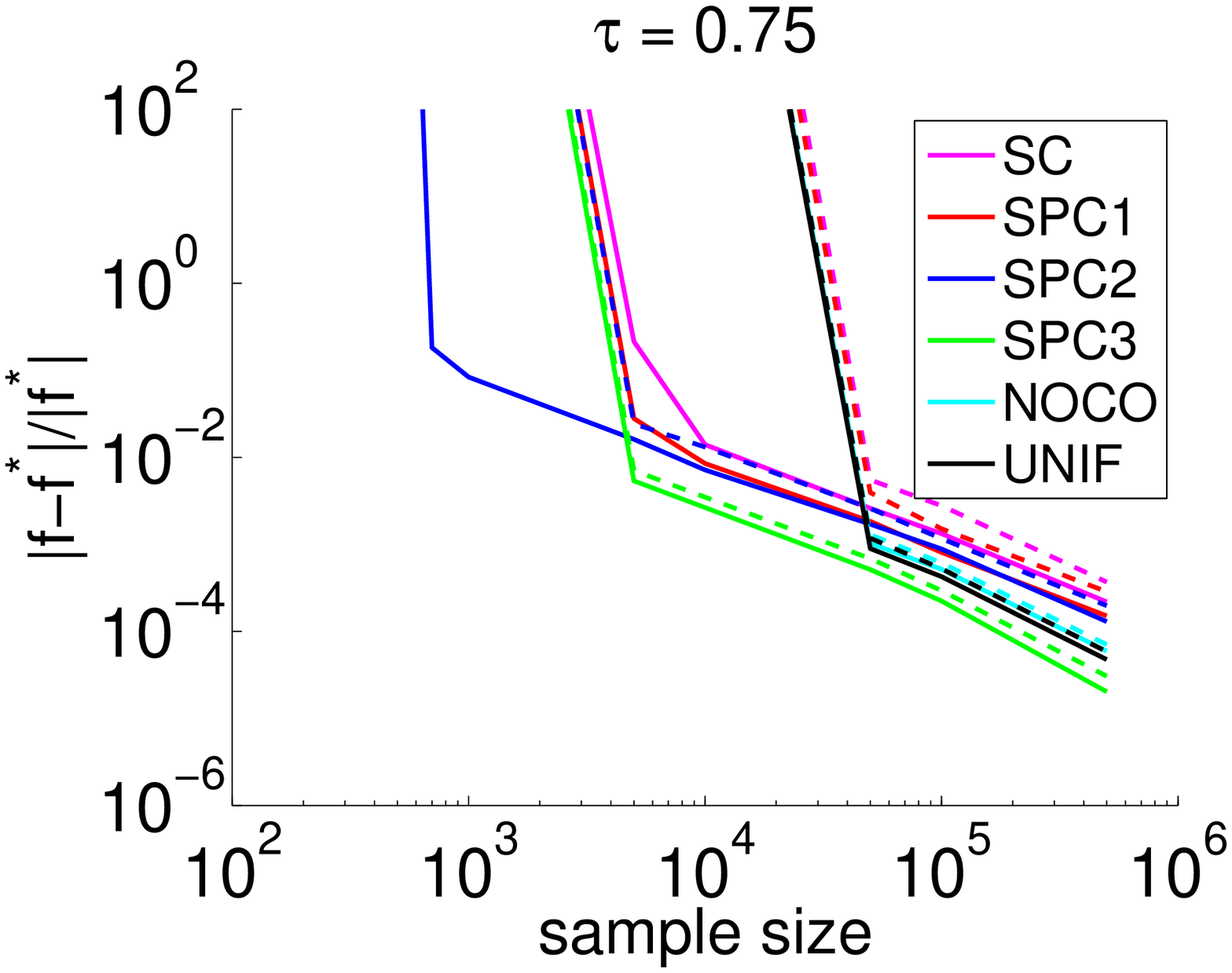}
 }
 \\
  \subfigure[$1e6 \times 10$, $\|x-x^*\|_2/\|x^*\|_2$]{
   \includegraphics[width=0.3\textwidth] {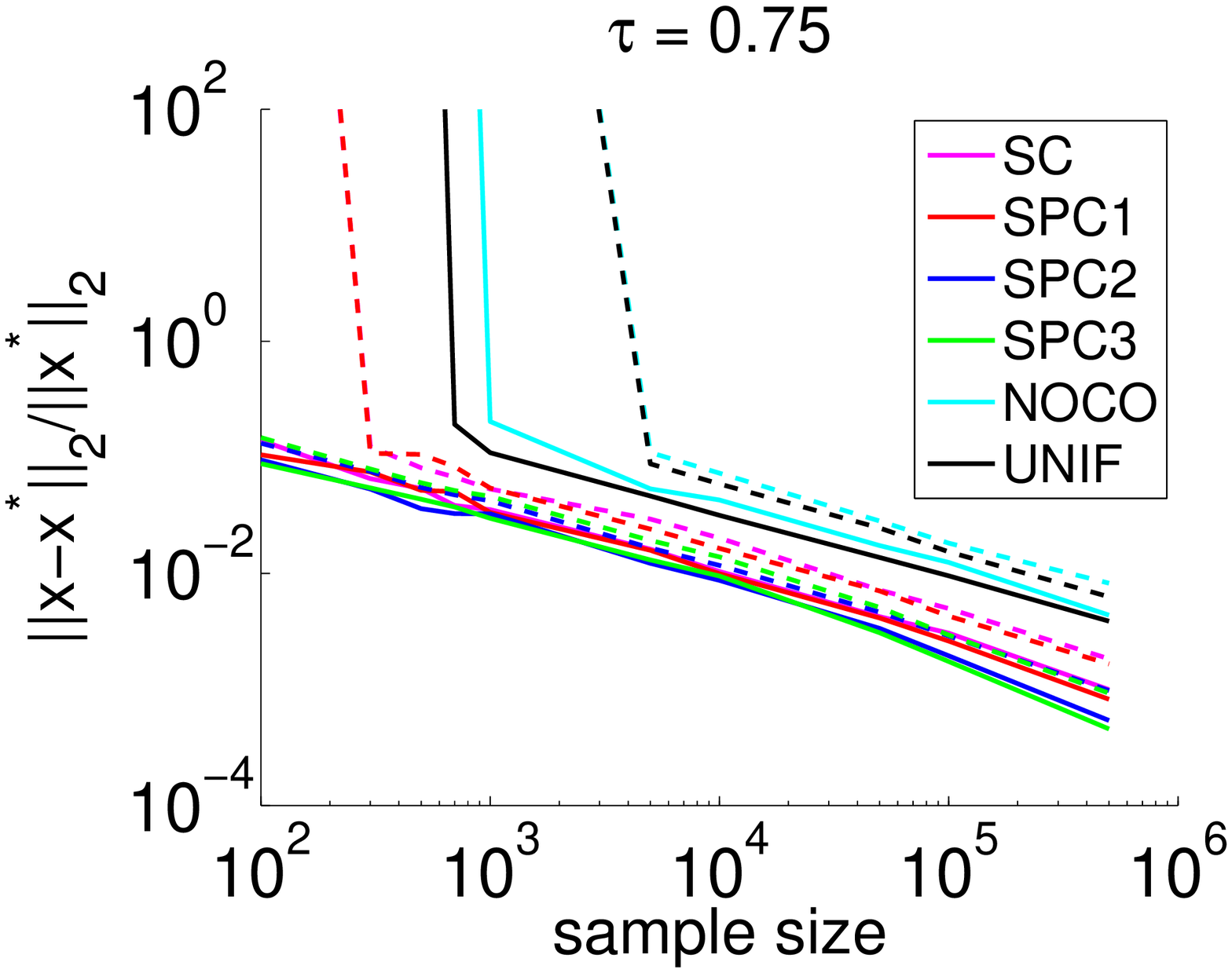}
 } 
 &
 \subfigure[$1e6 \times 50$, $\|x-x^*\|_2/\|x^*\|_2$]{
   \includegraphics[width=0.3\textwidth] {FIG/err_s/quartile11_5.eps}
 }
  &
   \subfigure[$1e6 \times 100$, $\|x-x^*\|_2/\|x^*\|_2$]{
   \includegraphics[width=0.3\textwidth] {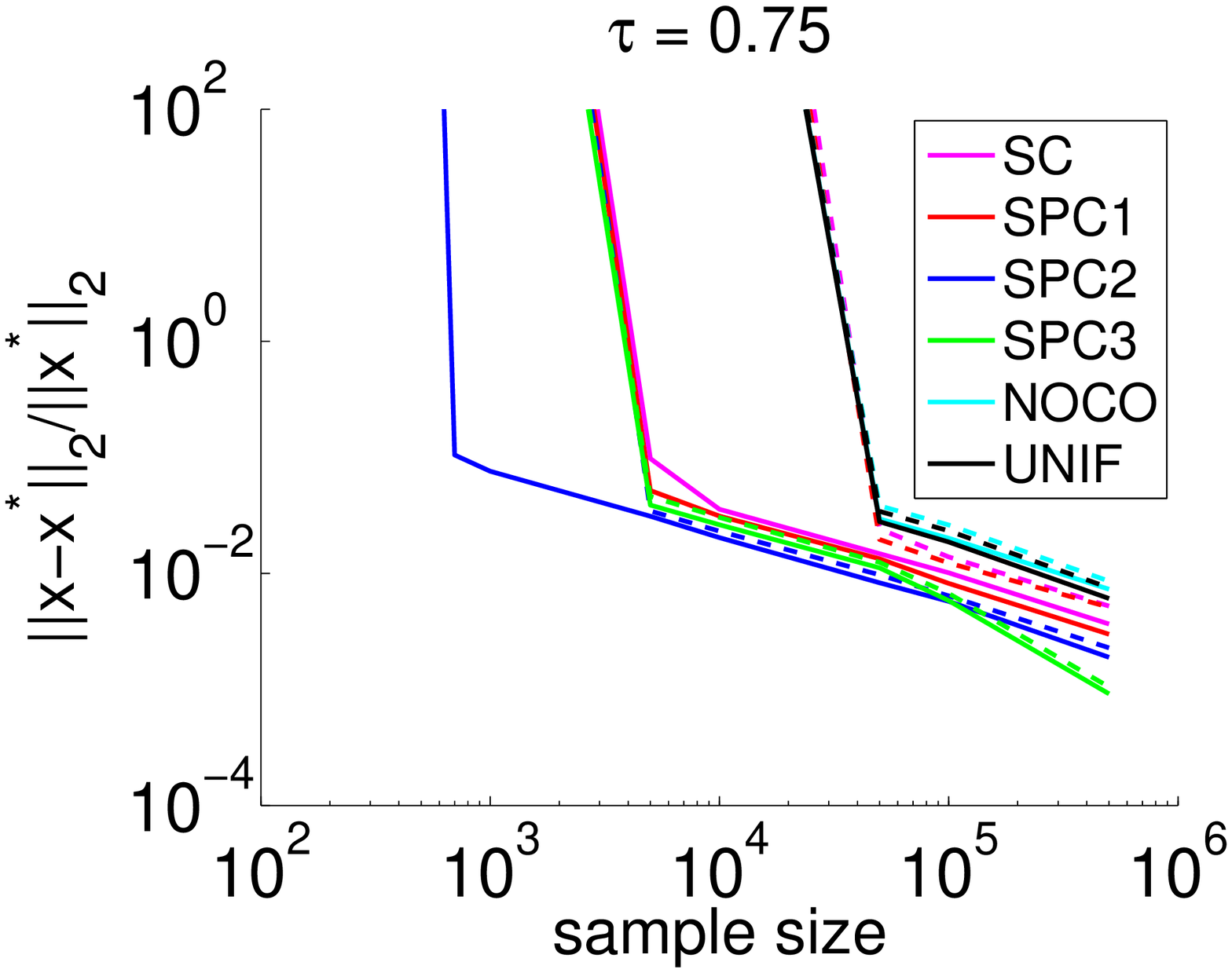}
 }
 \end{tabular}
 \end{center}
  \caption{
   The first (solid lines) and the third (dashed lines) quartiles of the relative errors of the objective value
   (namely, $|f-f^*|/|f^*|$) and solution vector (namely, $\|x-x^*\|_2/\|x^*\|_2$),
   when the sample size $s$ changes, for different values of $d$, while $n=1e6$ by using 6 different methods, among 50 independent trials.
   The test is on skewed data and $\tau = 0.75$.
   The three different columns correspond to $d = 10, 50, 100$, respectively.
   }
  \label{comp_d}
\end{figure}

For Figure~\ref{comp_d}, as $d$ gets larger, the performance of the two 
naive methods do not vary a lot.
However, this increases the difficulty for conditioning methods to yield 
2-digit accuracy.
When $d$ is quite small, most methods can yield 2-digit accuracy even when 
$s$ is not large.
When $d$ becomes large, SPC2 and SPC3 provide good estimation, even when $s<1000$. 
The relative performance among these methods remains unchanged.
For Figure~\ref{err_d}, the relative errors are monotonically increasing 
for each sampling size.
This is consistent with our theory that, to yield high accuracy, the required 
sampling size is a low-degree polynomial of $d$. 

\begin{figure}[h!tbp]
 \begin{center}
 \begin{tabular}{ccc}
\subfigure[$\tau = 0.5$, $|f-f^*|/|f^*|$]{
   \includegraphics[width=0.3\textwidth] {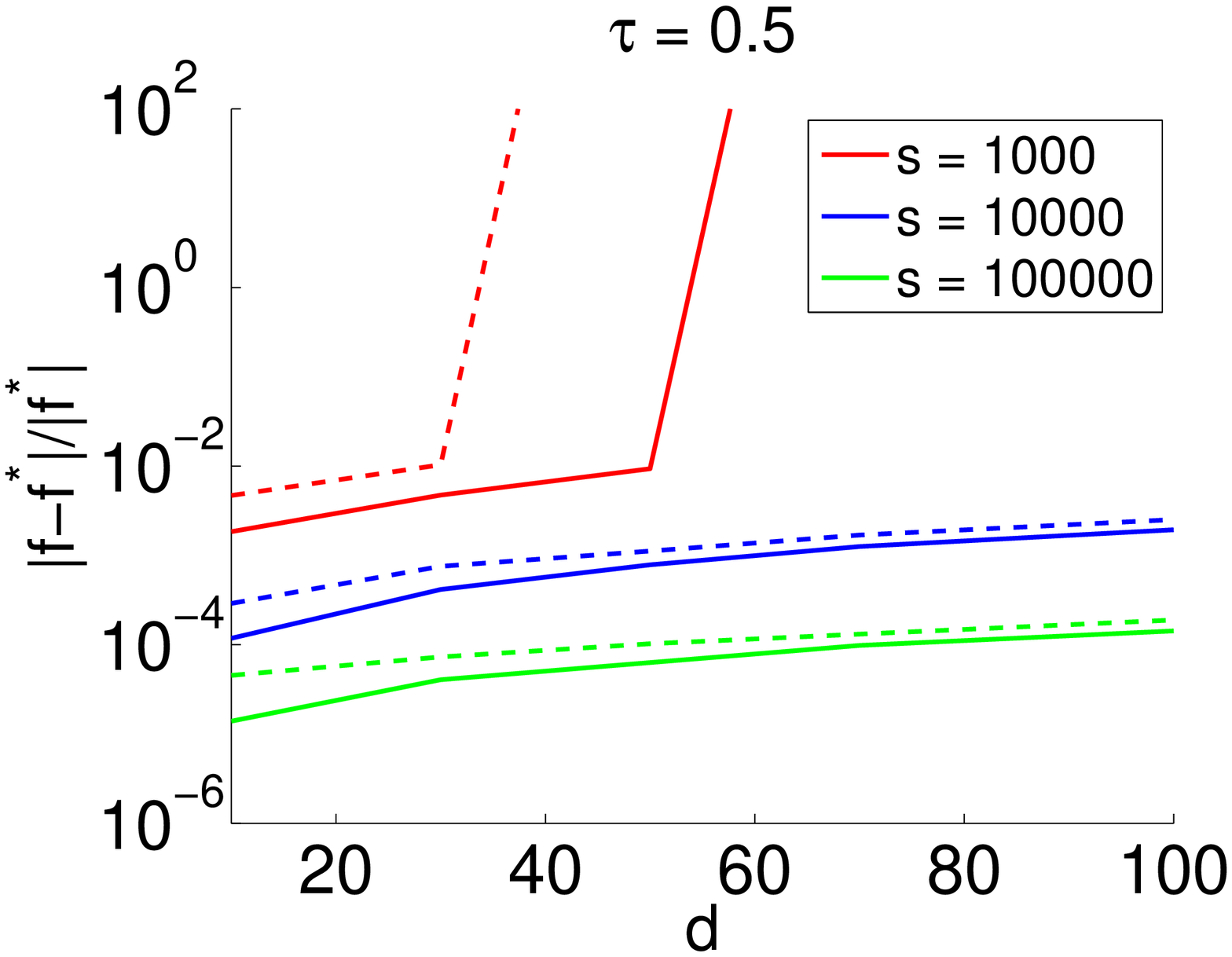}
 }
 &
 \subfigure[$\tau = 0.75$,  $|f-f^*|/|f^*|$]{
   \includegraphics[width=0.3\textwidth] {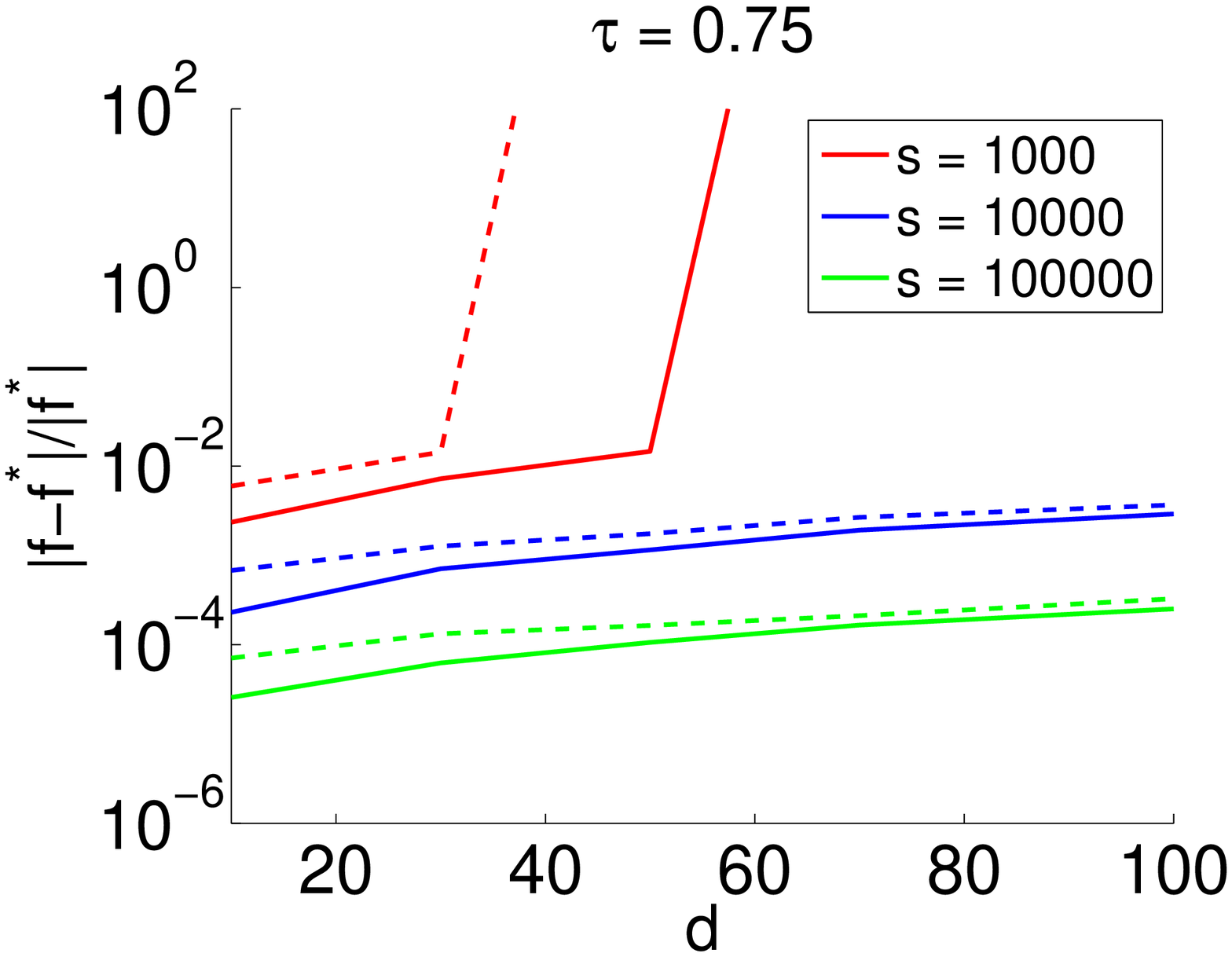}
 }
 &
\subfigure[$\tau = 0.95$,  $|f-f^*|/|f^*|$]{
   \includegraphics[width=0.3\textwidth] {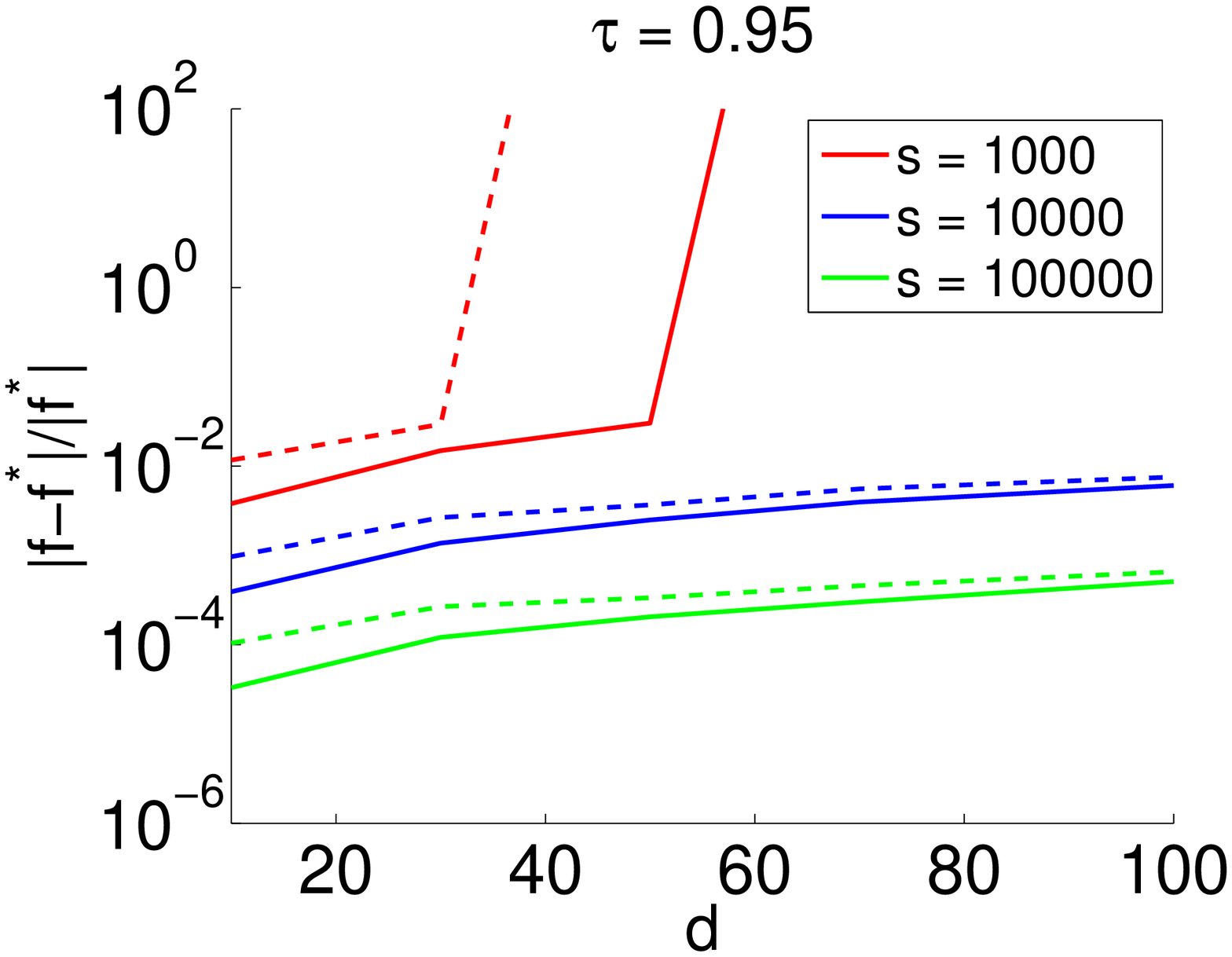}
 }
 \\
 \subfigure[$\tau = 0.5$, $\|x-x^*\|_2/\|x^*\|_2$]{
   \includegraphics[width=0.3\textwidth] {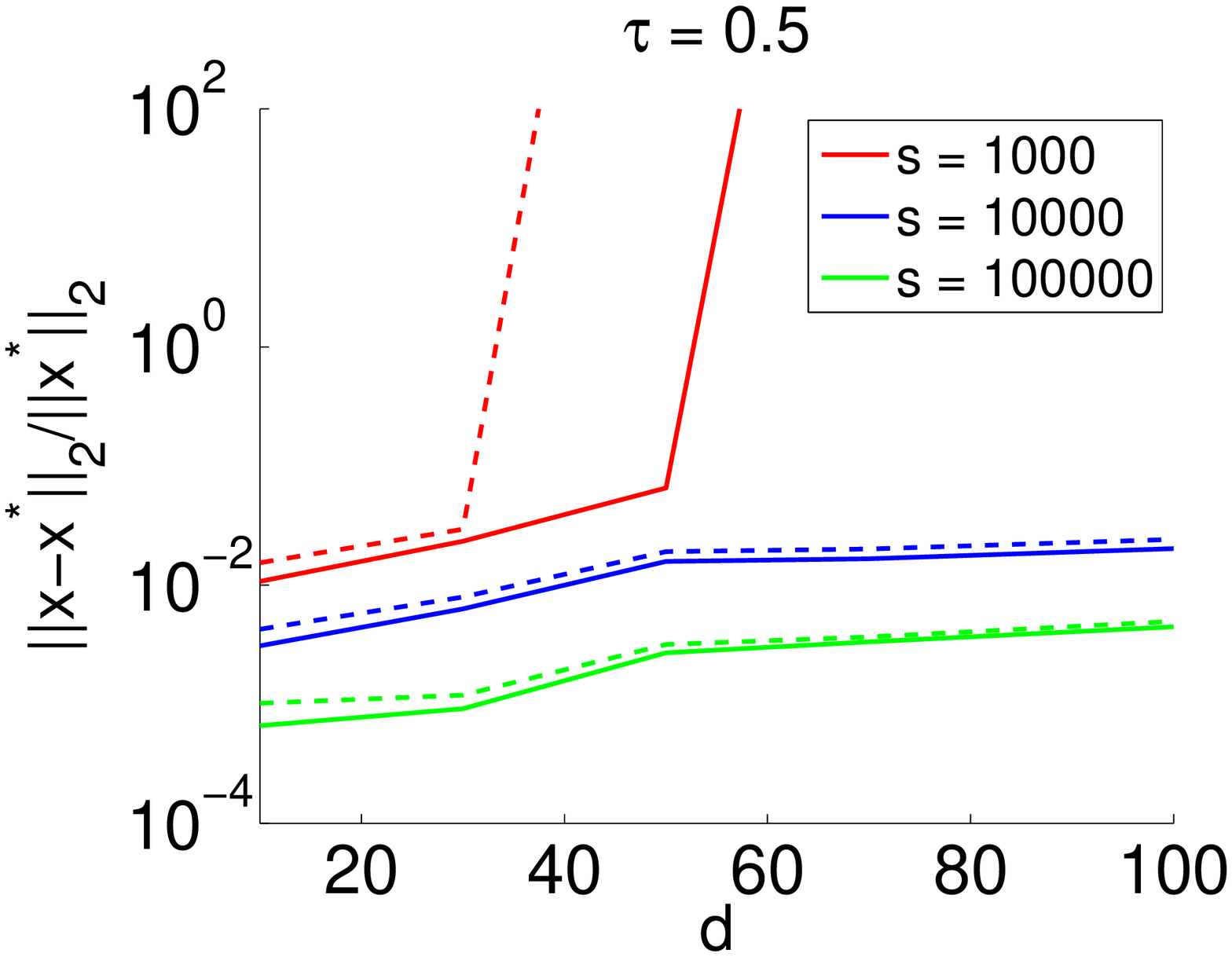}
 }
 &
 \subfigure[$\tau = 0.75$, $\|x-x^*\|_2/\|x^*\|_2$]{
   \includegraphics[width=0.3\textwidth] {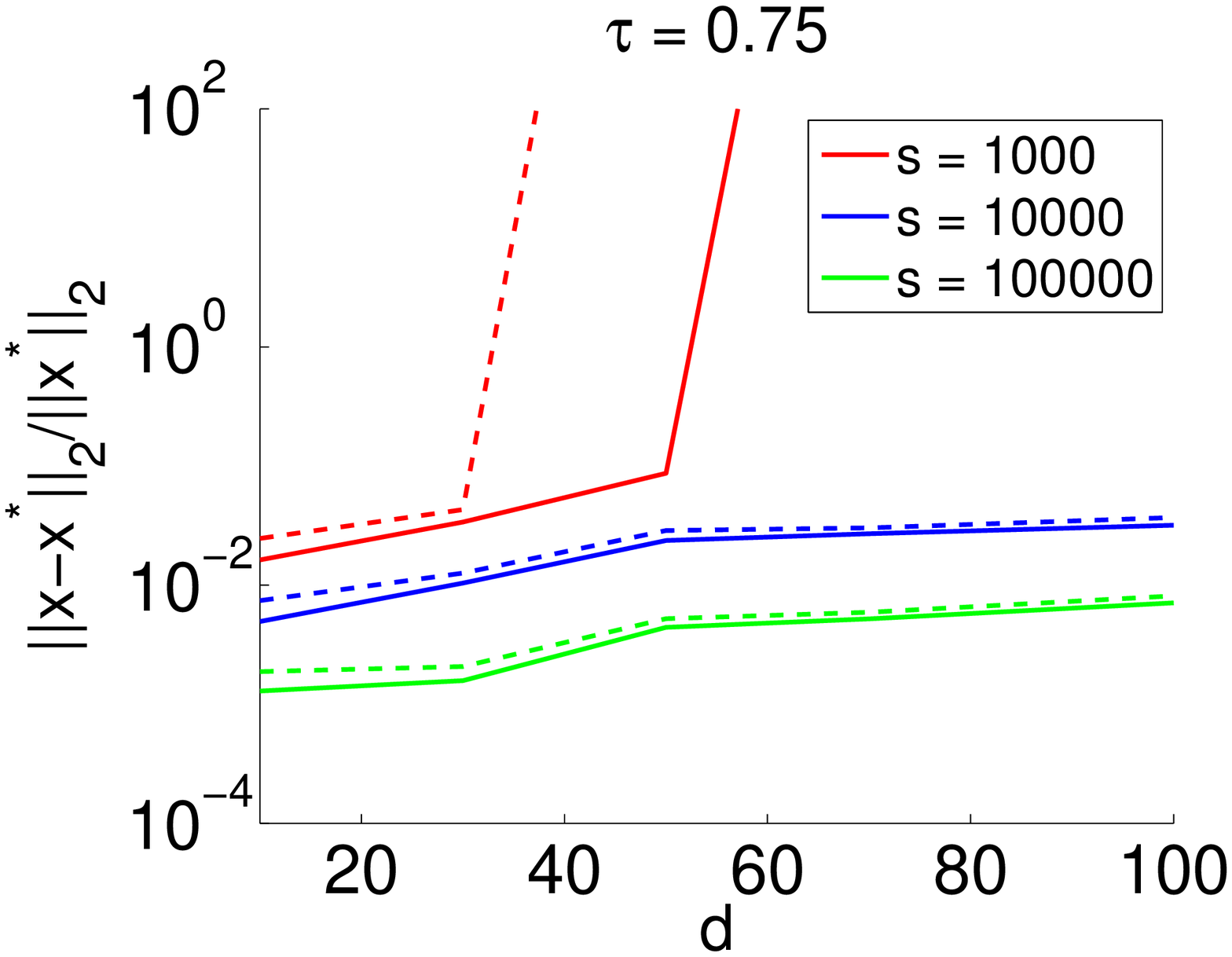}
 }
 &
 \subfigure[$\tau = 0.95$, $\|x-x^*\|_2/\|x^*\|_2$]{
   \includegraphics[width=0.3\textwidth] {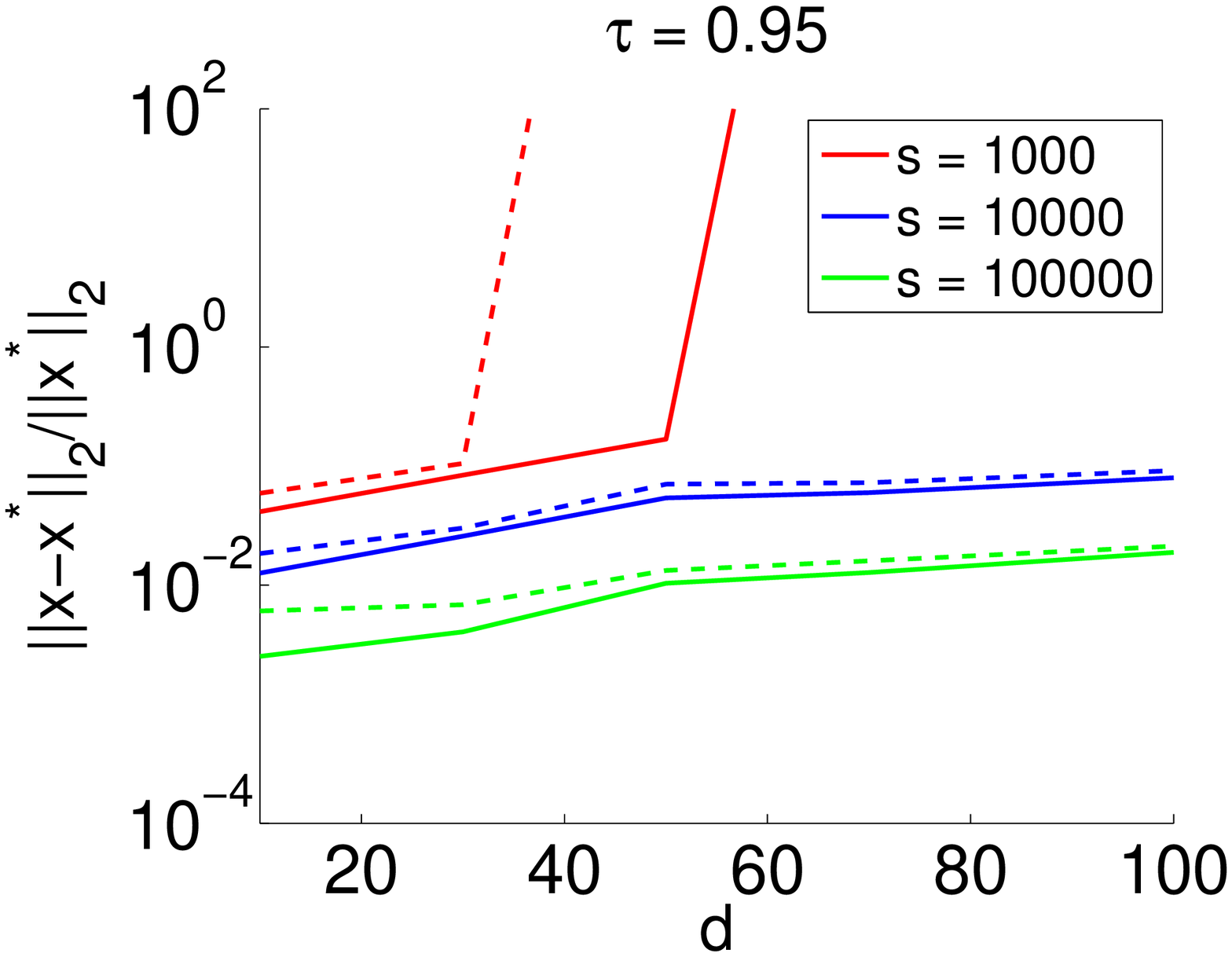}
 }
 \end{tabular}
 \end{center}
 \caption{
   The first (solid lines) and the third (dashed lines) quartiles of the relative errors of the objective value
   (namely, $|f-f^*|/|f^*|$) and solution vector (namely, $\|x-x^*\|_2/\|x^*\|_2$),
   when $d$ varying from 10 to 100 and $n = 1e6$ by using SPC3, among 50 independent trials
   The test is on skewed data.
   The three different columns correspond to $\tau = 0.5, 0.75, 0.95$, respectively.
   }
 \label{err_d}
\end{figure}


\subsection{Quality of approximation when the quantile parameter $\tau$ changes}
\label{relerr_tau}

Next, we will let $\tau$ change, for a fixed data set and fixed conditioning method, 
and we will investigate how the resulting errors behave as a function of 
$\tau$.
We will consider $\tau$ in the range of $[0.5, 0.9]$, equally spaced by 0.05,
as well as several extreme quantiles such as $0.975$ and $0.98$.
We consider skewed data with size $1e6 \times 50$; and our plots are shown in 
Figure~\ref{err_tau}.

\begin{figure}[h!tbp]
 \begin{center}
 \begin{tabular}{ccc}
\subfigure[SPC1, $|f-f^*|/|f^*|$]{
   \includegraphics[width=0.3\textwidth] {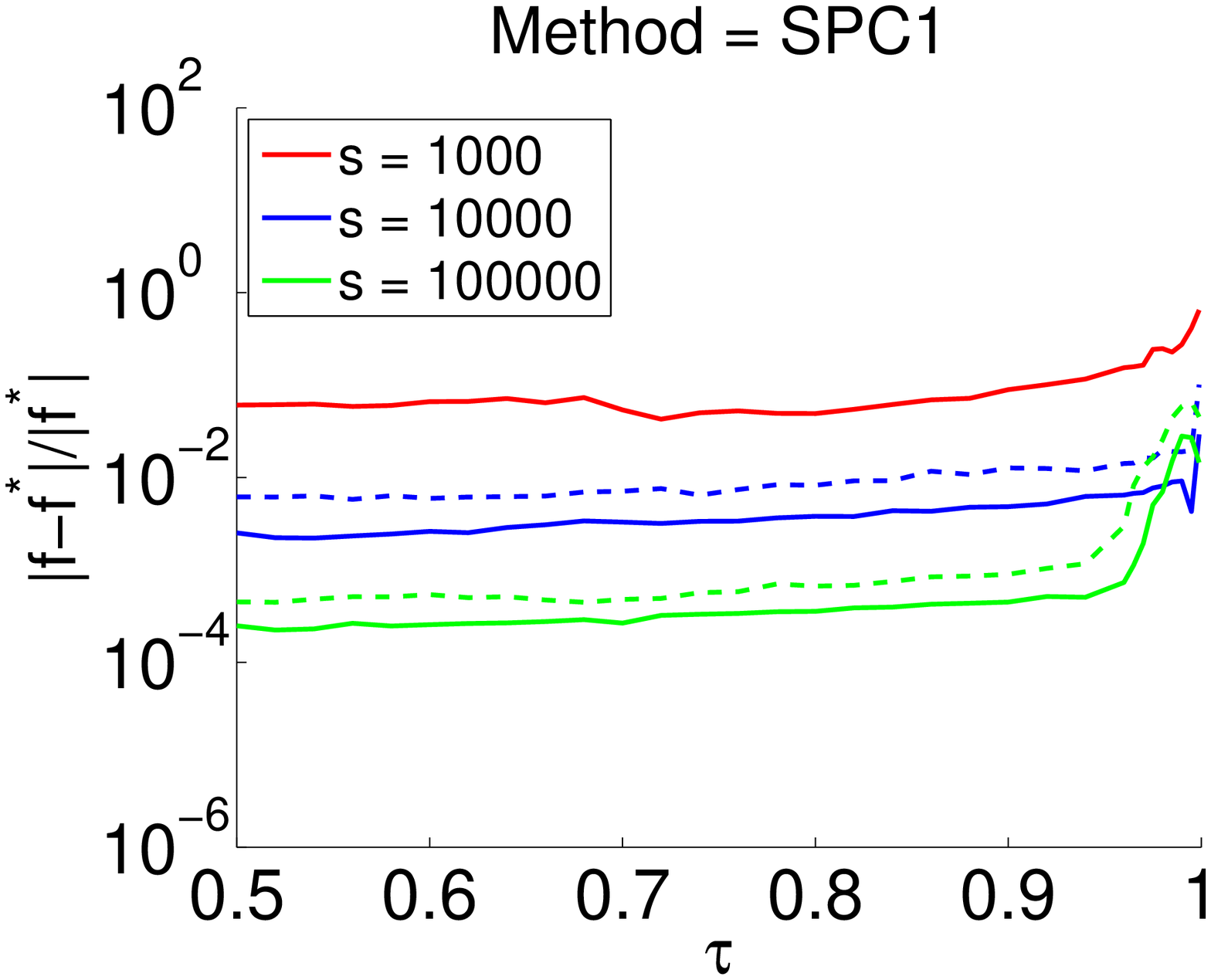}
 }
 &
 \subfigure[SPC2, $|f-f^*|/|f^*|$]{
   \includegraphics[width=0.3\textwidth] {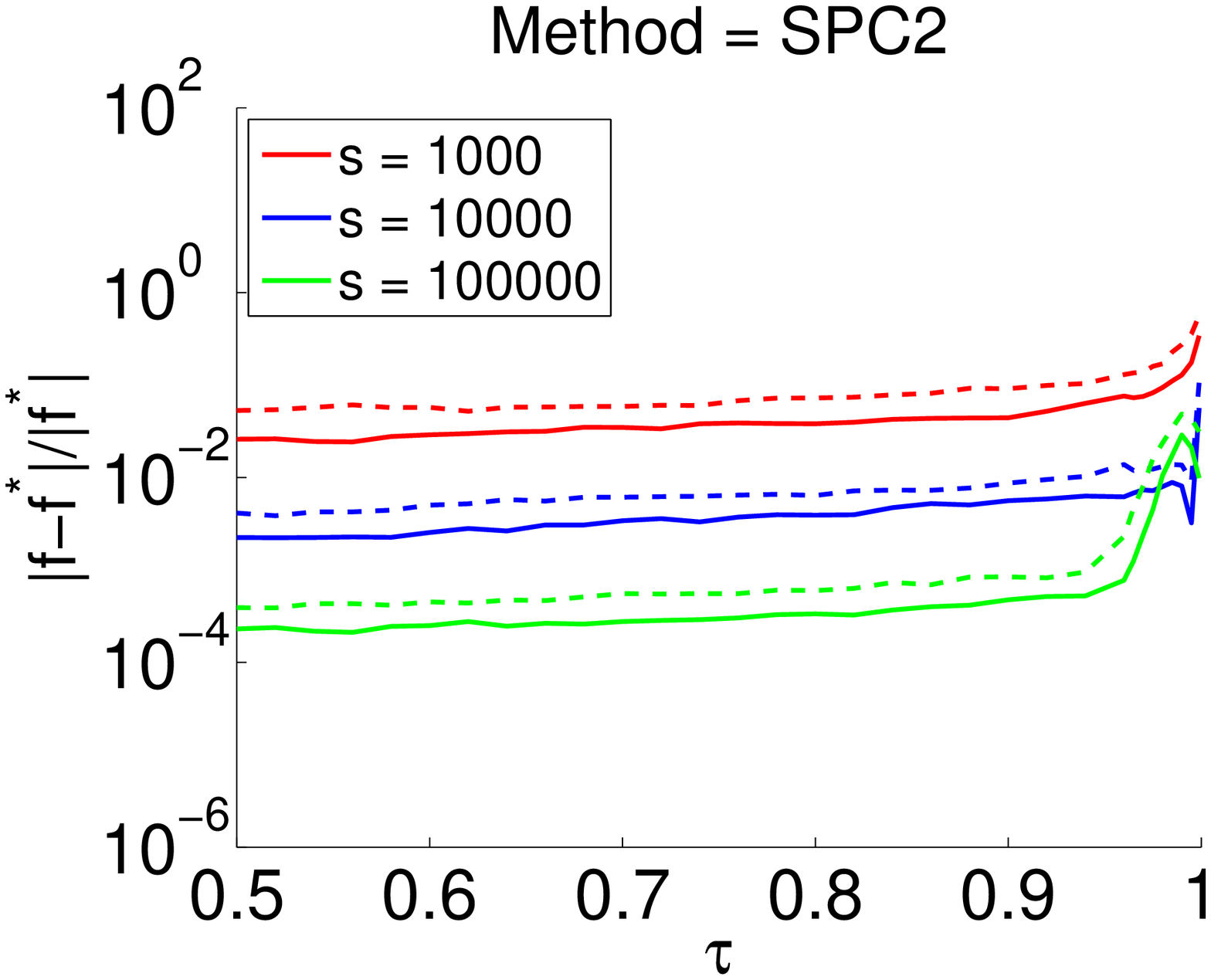}
   }
  &
  \subfigure[SPC3, $|f-f^*|/|f^*|$]{
   \includegraphics[width=0.3\textwidth] {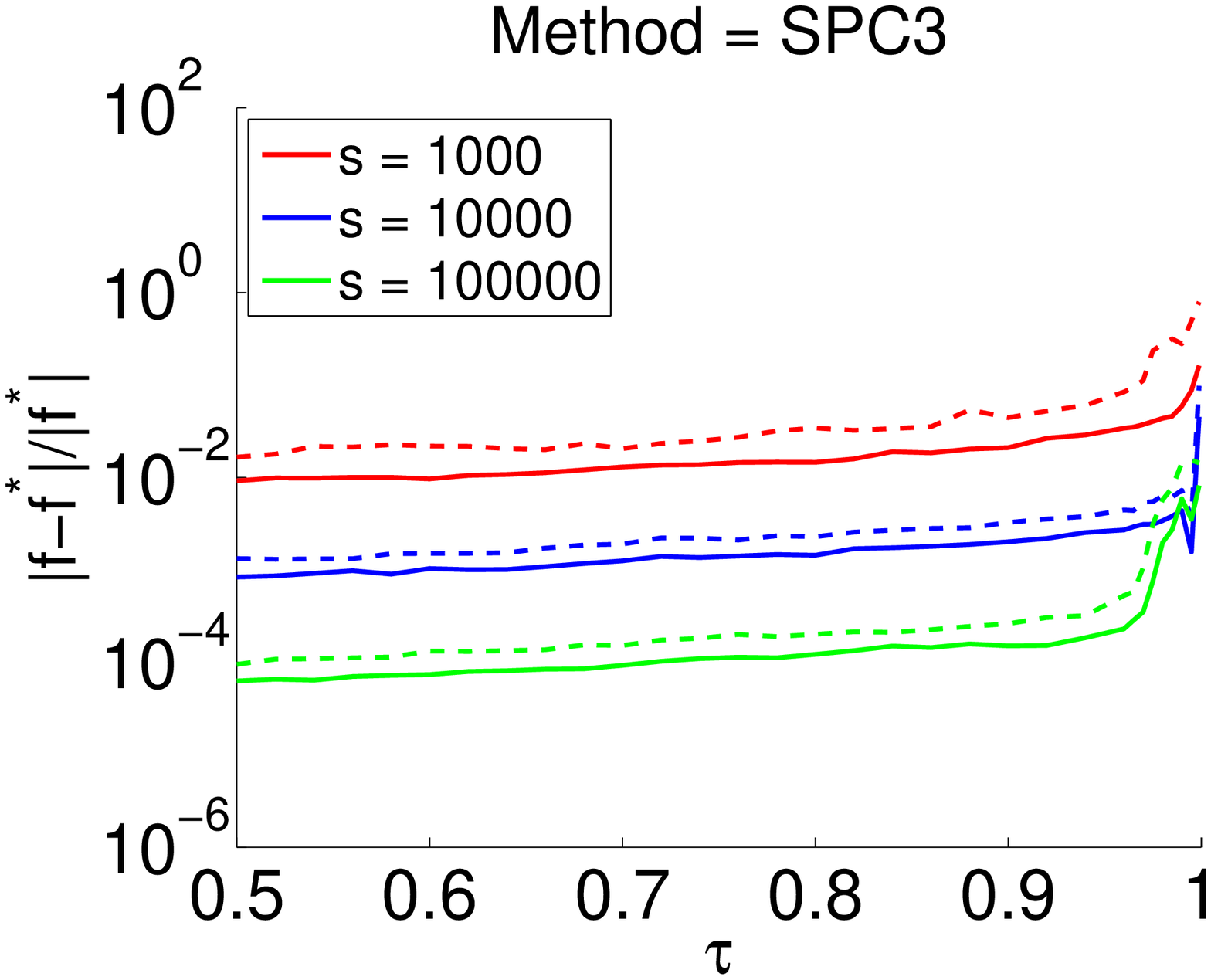}
  }
  \\
  \subfigure[SPC1, $\|x-x^*\|_2/\|x^*\|_2$]{
   \includegraphics[width=0.3\textwidth] {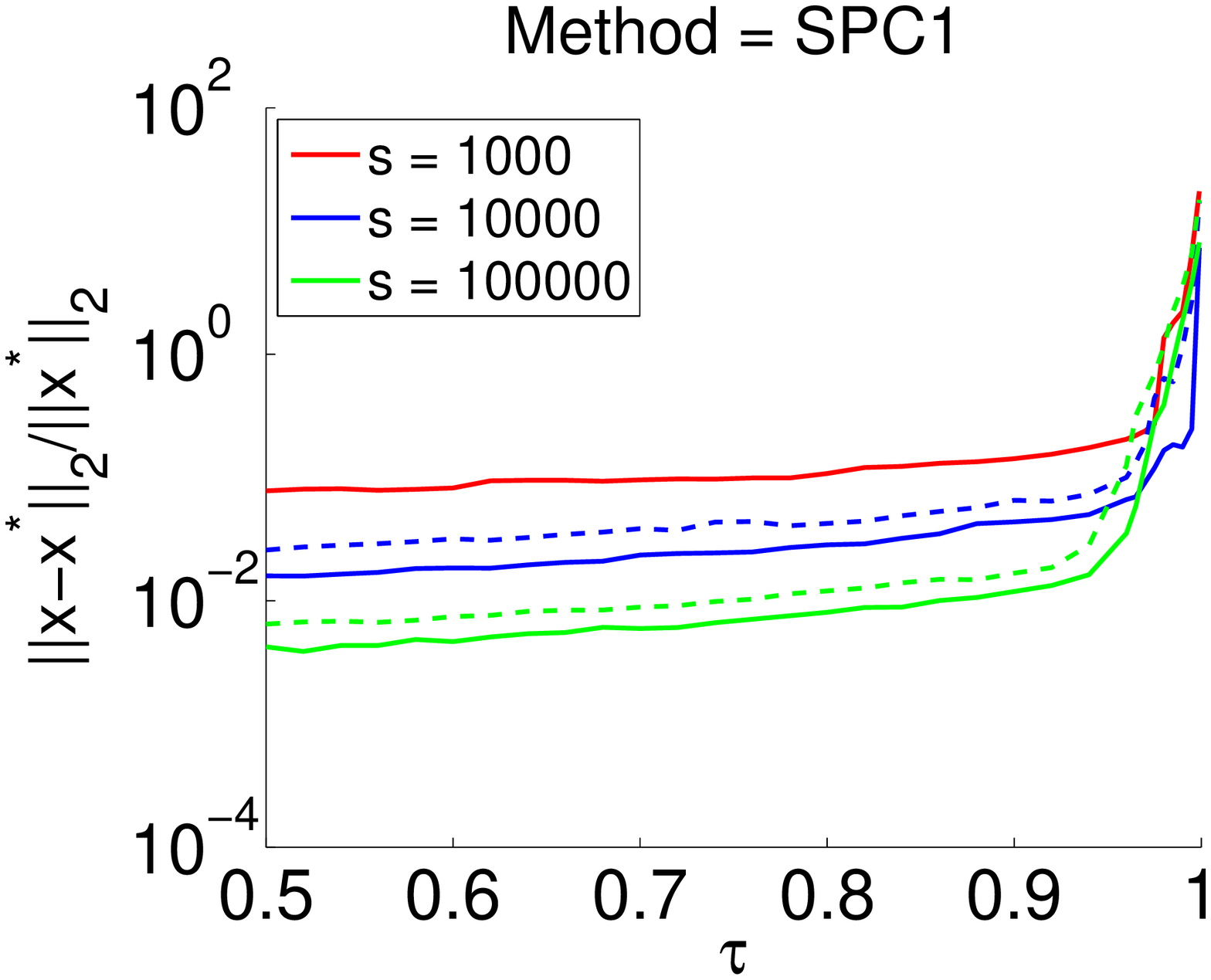}
 }
 &
 \subfigure[SPC2, $\|x-x^*\|_2/\|x^*\|_2$]{
   \includegraphics[width=0.3\textwidth] {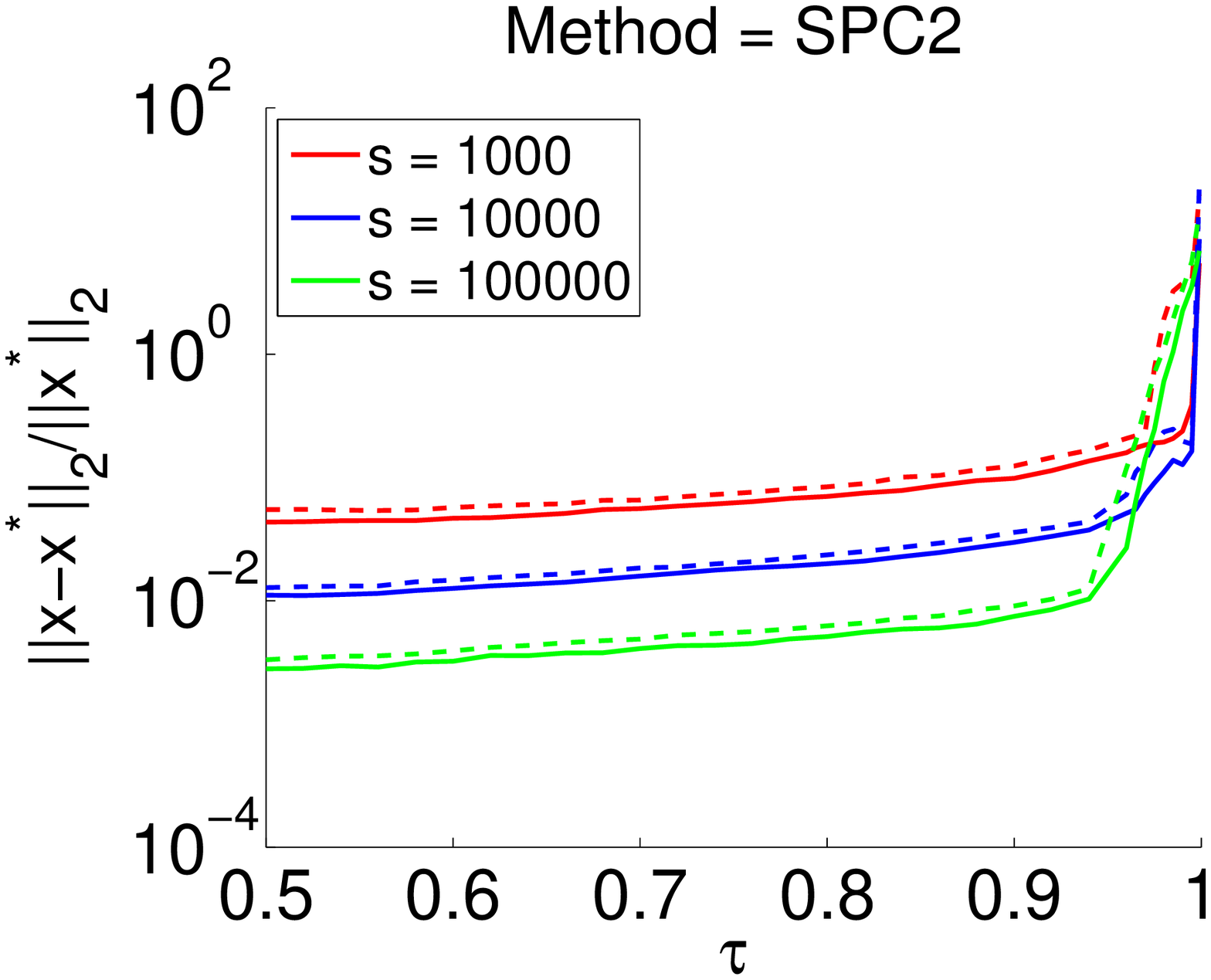}
 }
 &
\subfigure[SPC3, $\|x-x^*\|_2/\|x^*\|_2$]{
   \includegraphics[width=0.3\textwidth] {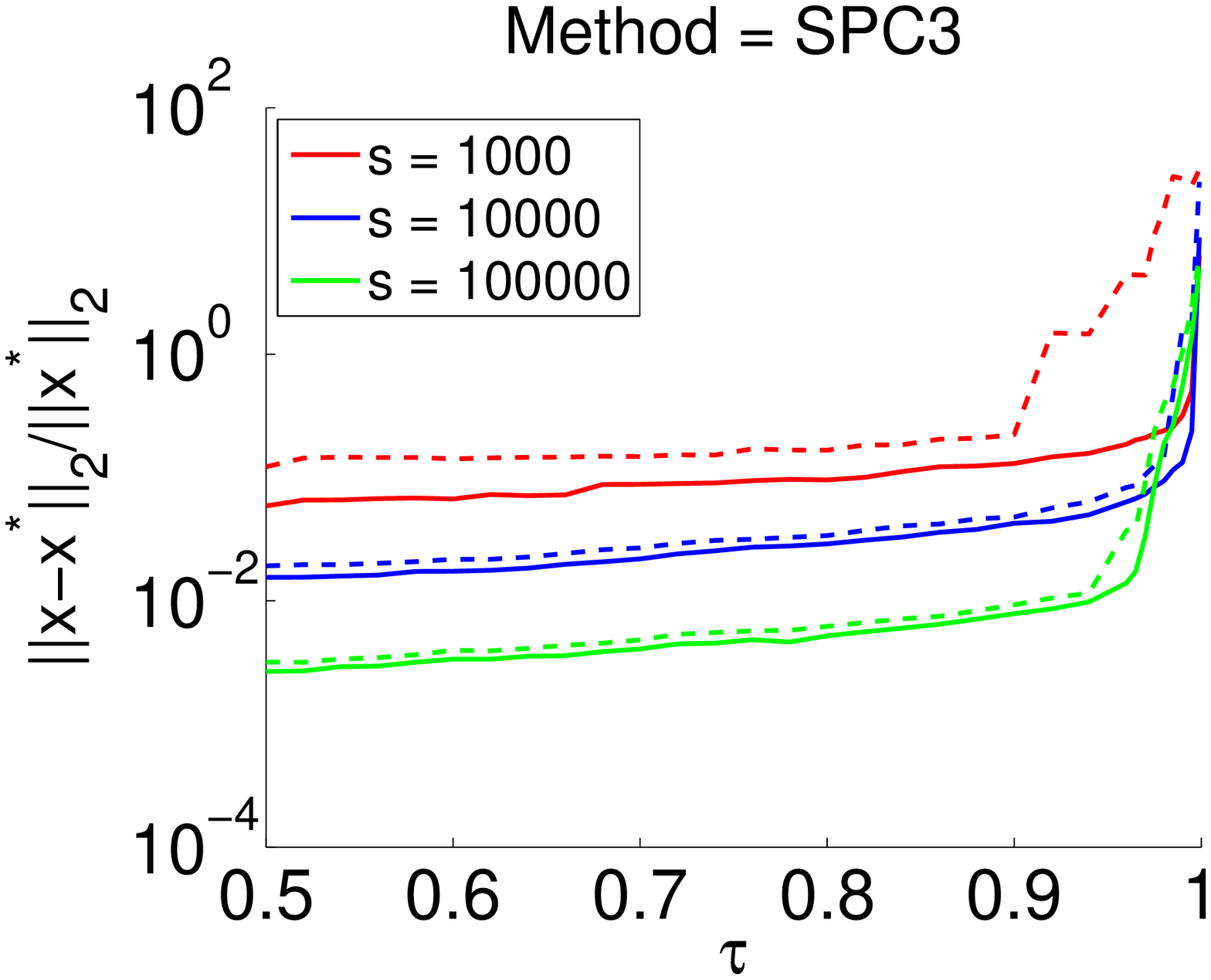}
 }
 \end{tabular}
 \end{center}
 \caption{
   The first (solid lines) and the third (dashed lines) quartiles of the relative errors of the objective value,
   (namely  $|f-f^*|/|f^*|$) and solution vector  (namely, $\|x-x^*\|_2/\|x^*\|_2$),
   when $\tau$ varying from 0.5 to 0.999 by using SPC1, SCP2, SPC3, among 50 independent trials.
   The test is on skewed data with size $1e6$ by 50.
   Within each plot, three sampling sizes are considered, namely, $1e4, 1e4, 1e5$.
     }
  \label{err_tau}
\end{figure}

The plots in Figure~\ref{err_tau} demonstrate that, given the same method 
and sampling size $s$, the relative errors are monotonically increasing 
but only very gradually, \emph{i.e.}, they do not change very substantially 
in the range of $[0.5, 0.95]$.
On the other hand, all the methods generate high relative errors when $\tau$ 
takes extreme values very near $1$ (or $0$).
Overall, SPC2 and SPC3 performs better than SPC1.
Although for some quantiles SPC3 can yield slightly lower errors than SPC2, 
it too yields worst results when $\tau$ takes on extreme values.


\subsection{Evaluation on running time performance}
\label{running_time}
 
In this subsection, we will describe running time issues, with an emphasis 
on how the running time behaves as a function of $s$, $d$ and $n$.

\paragraph{When the sampling size $s$ changes\\}

To start, Figure~\ref{time_s} shows the running time for computing three 
subproblems associated with three different $\tau$ values by using six 
methods (namely, SC, SPC1, SPC2, SPC3, NOCO, UNIF) when the sampling size 
$s$ changes.
(This is simply the running time comparison for all the six methods used to 
generate Figure~\ref{err_s}.)
As expected, the two naive methods (NOCO and UNIF) run faster than other methods in most 
cases---since they don't perform the additional step of conditioning.
For $s < 10^4$, among the conditioning-based methods, SPC1 runs fastest, 
followed by SPC3 and then SPC2.
As $s$ increases, however, the faster methods, including NOCO and UNIF, become 
relatively more expensive; and when $s \approx 5e5$, all of the curves, 
except for SPC1, reach almost the same point.

To understand what is happening here, recall that we accept the sampling 
size $s$ as an input in our algorithm; and we then construct our sampling 
probabilities by $\hat p_i = \min\{1, s \cdot \lambda_i / \sum \lambda_i\}$,
where $\lambda_i$ is the estimation of the $\ell_1$ norm of the $i$-th row 
of a well-conditioned basis.
(See Step 4 in Algorithm~\ref{embed_alg}.)
Hence, the $s$ is not the exact sampling size.
Indeed, upon examination, in this regime when $s$ is large, the actual 
sampling size is often much less than the input $s$.
As a result, almost all the conditioning-based algorithms are solving a 
subproblem with size, say, $s/2 \times d$, while the two naive methods are 
are solving subproblem with size about $s \times d$.
The difference of running time for solving problems with these sizes can be
quite large when $s$ is large.
For conditioning-based algorithms, the running time mainly comes from the 
time for conditioning and solving the subproblem.
Thus, since SPC1 needs the least time for conditioning, it should be clear 
why SPC1 needs much less time when $s$ is very large.

\begin{figure}[h!tbp]
\begin{center}
   \includegraphics[scale = 0.5] {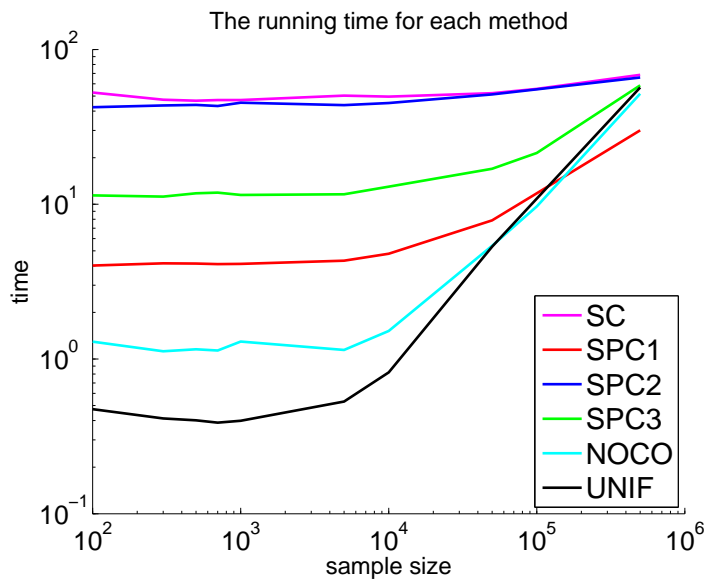}
 \end{center}
 \caption{The running time for solving the three problems associated with three different $\tau$ values
  by using six methods, namely, SC, SPC1, SPC2, SPC3, NOCO, UNIF, when the sampling size $s$ changes.
     }
 \label{time_s}
\end{figure}

\paragraph{When the higher dimension $n$ changes\\}

Next, we compare the running time of our method with some competing methods 
when data size increases.
The competing methods are the primal-dual method, referred to as 
\texttt{ipm}, and that with preprocessing, referred to as \texttt{prqfn}; 
see \cite{PK97} for more details on these two methods.

We let the large dimension $n$ increase from $1e5$ to $1e8$, and we fix 
$s = 5e4$.
For completeness, in addition to the skewed data, we will consider two 
additional data sets.
First, we also consider a design matrix with entries generated from i.i.d. 
Gaussian distribution, where the response vector is generated in the same 
manner as the skewed data.
Also, we will replicate the census data 20 times to obtain a data set with 
size $1e8$ by $11$.
For each $n$, we extract the leading $n \times d$ submatrix of the 
replicated matrix, and we record the corresponding running time.  
The results of running time on all three data sets are shown in 
Figure~\ref{time_n}.

\begin{figure}[h!tbp]
 \begin{center}
  \begin{tabular}{ccc}  
  \subfigure[$\tau  = 0.5$]{
   \includegraphics[width=0.3\textwidth] {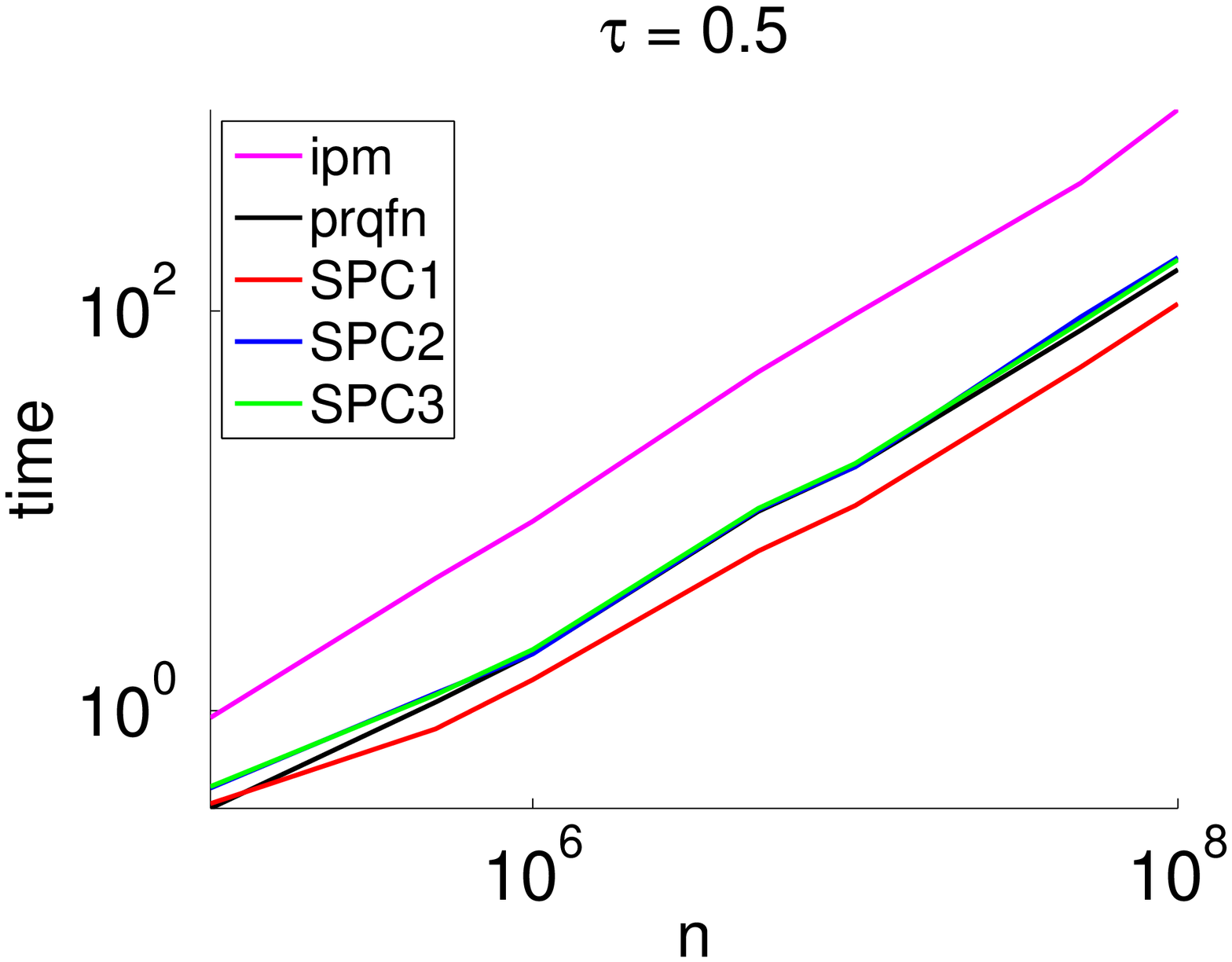}
 }
 &
\subfigure[$\tau  = 0.75$]{
   \includegraphics[width=0.3\textwidth] {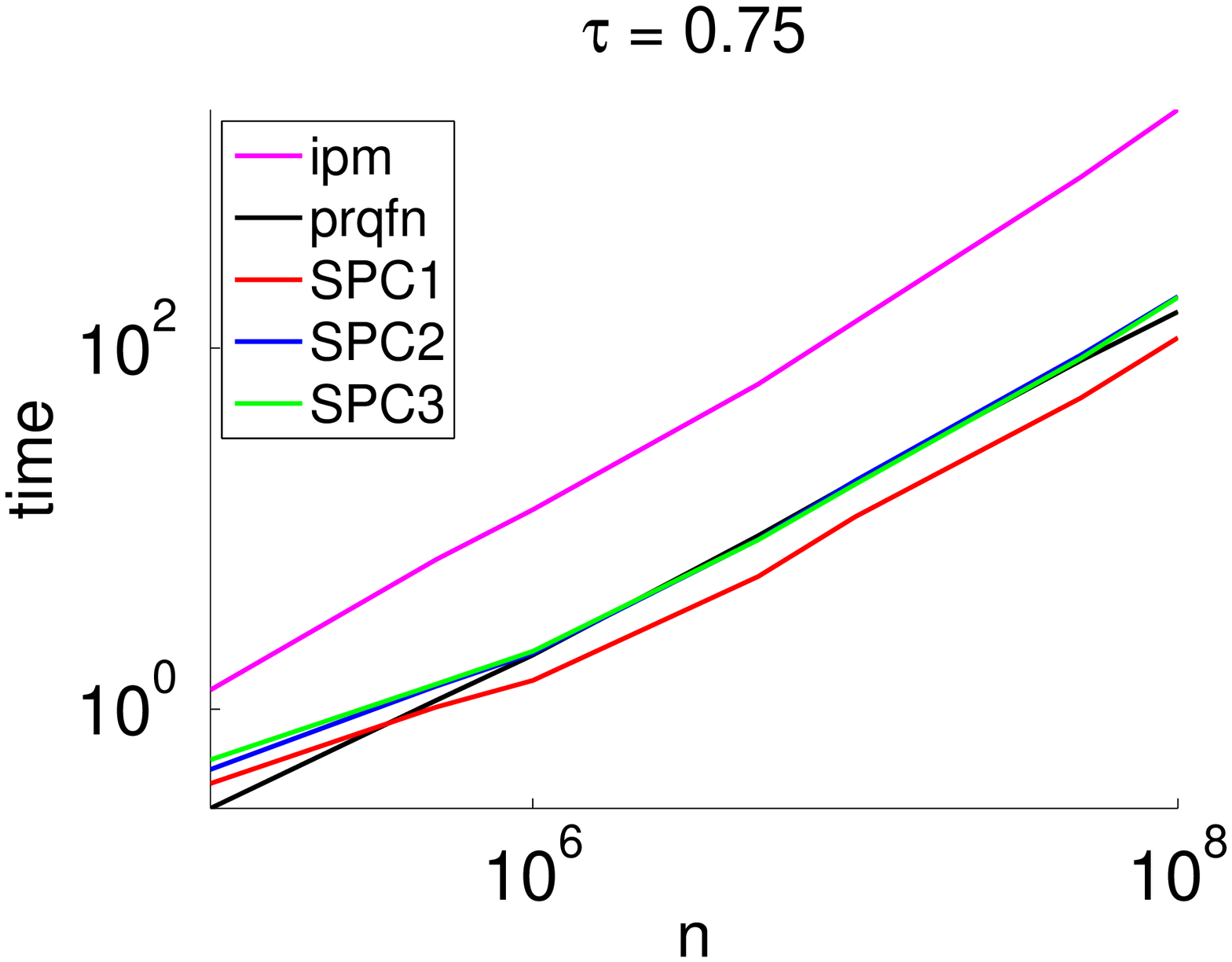}
 }
 &
 \subfigure[$\tau  = 0.95$]{
   \includegraphics[width=0.3\textwidth] {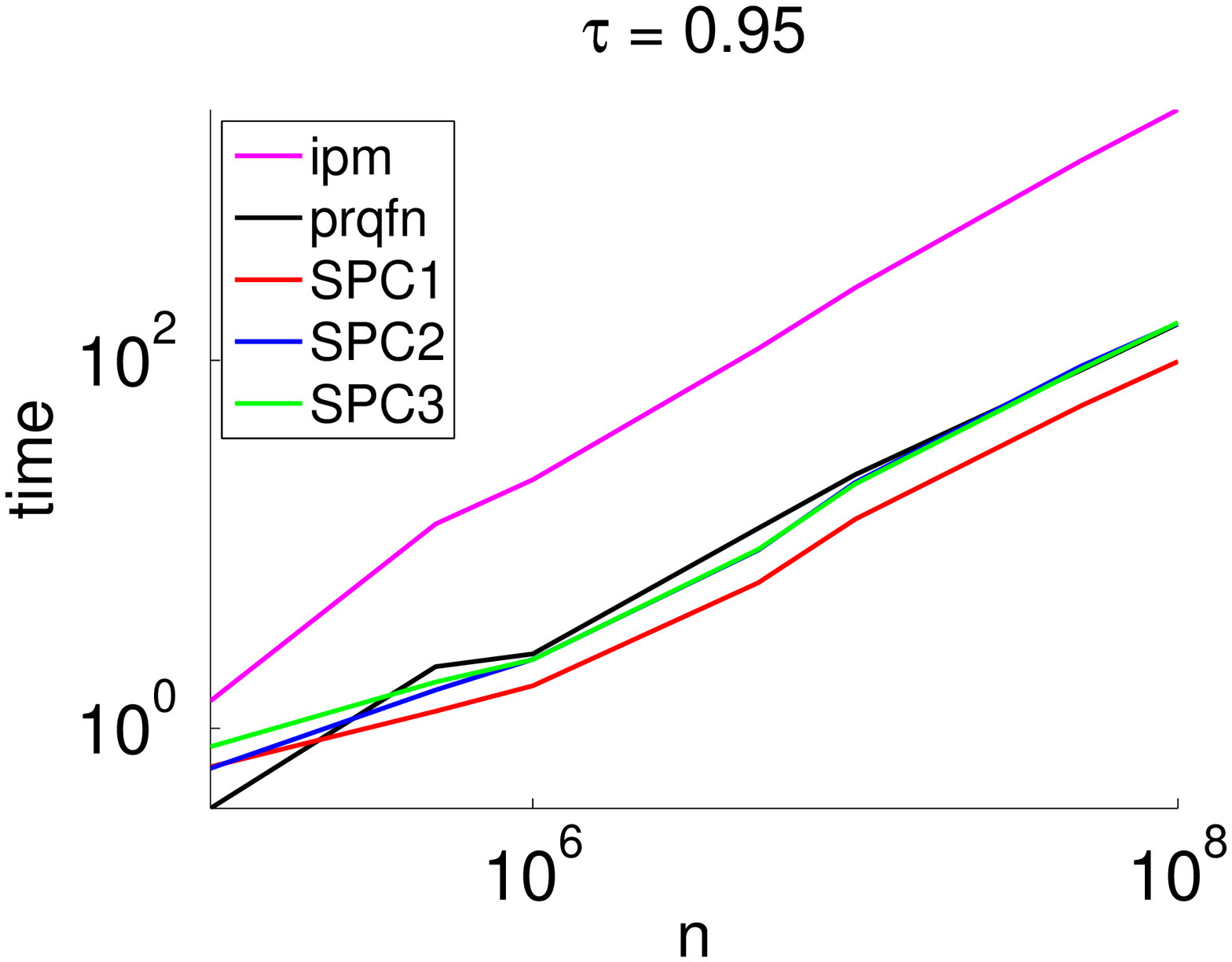}
 }
 \\
 \multicolumn{3}{c}{\bf Skewed data with $d = 11$}
 \\
   \subfigure[$\tau  = 0.5$]{
   \includegraphics[width=0.3\textwidth] {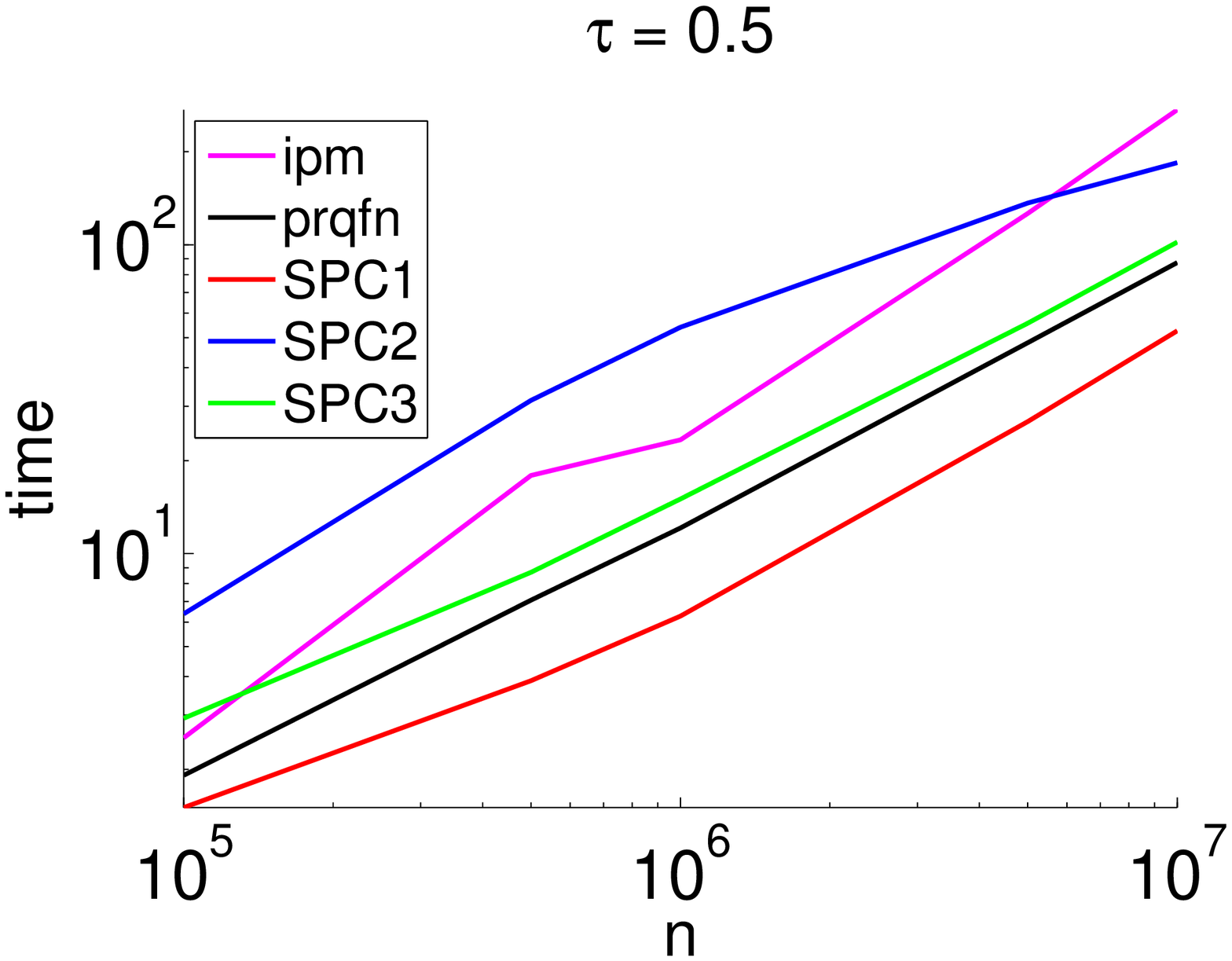}
 }
 &
\subfigure[$\tau  = 0.75$]{
   \includegraphics[width=0.3\textwidth] {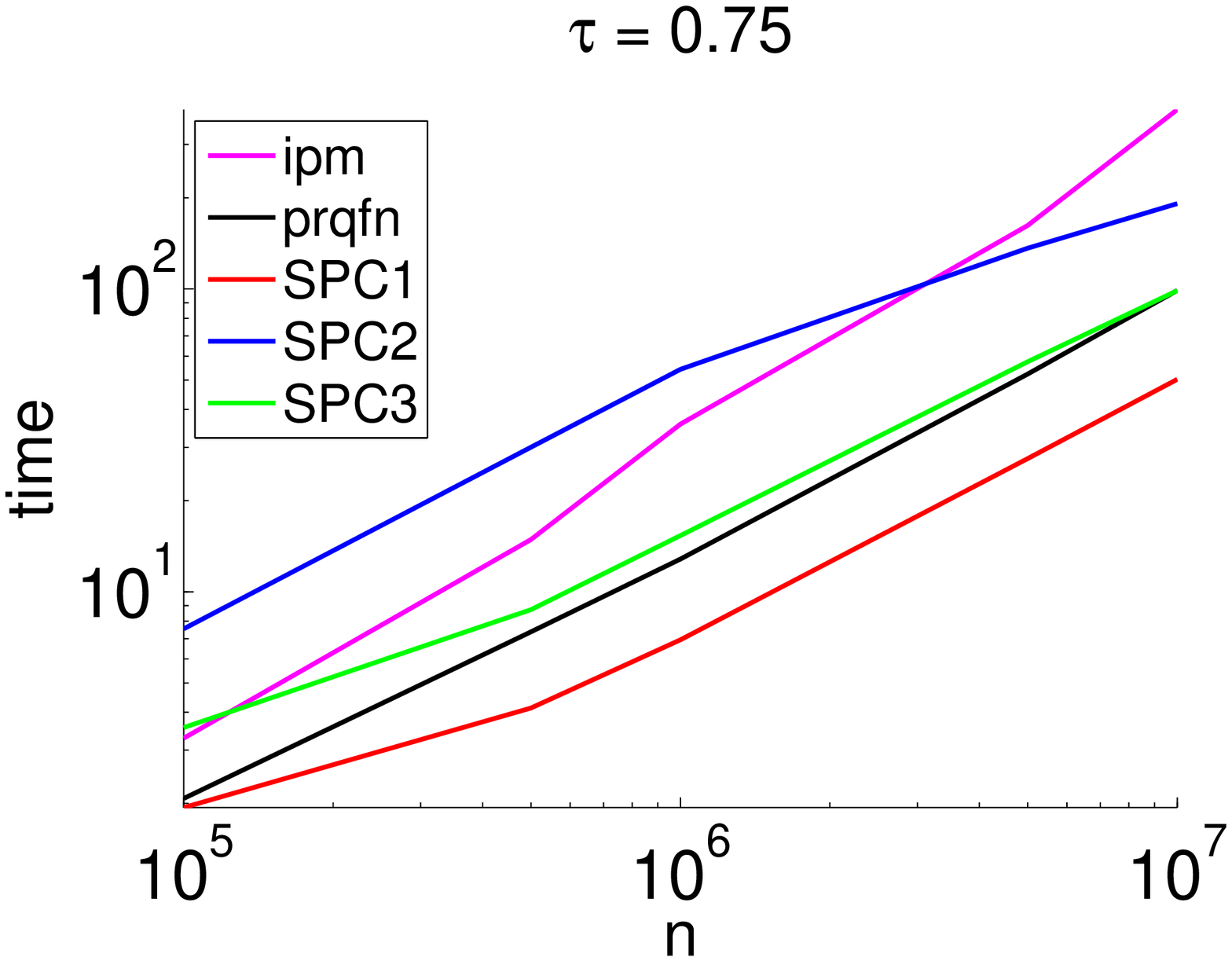}
 }
 &
 \subfigure[$\tau  = 0.95$]{
   \includegraphics[width=0.3\textwidth] {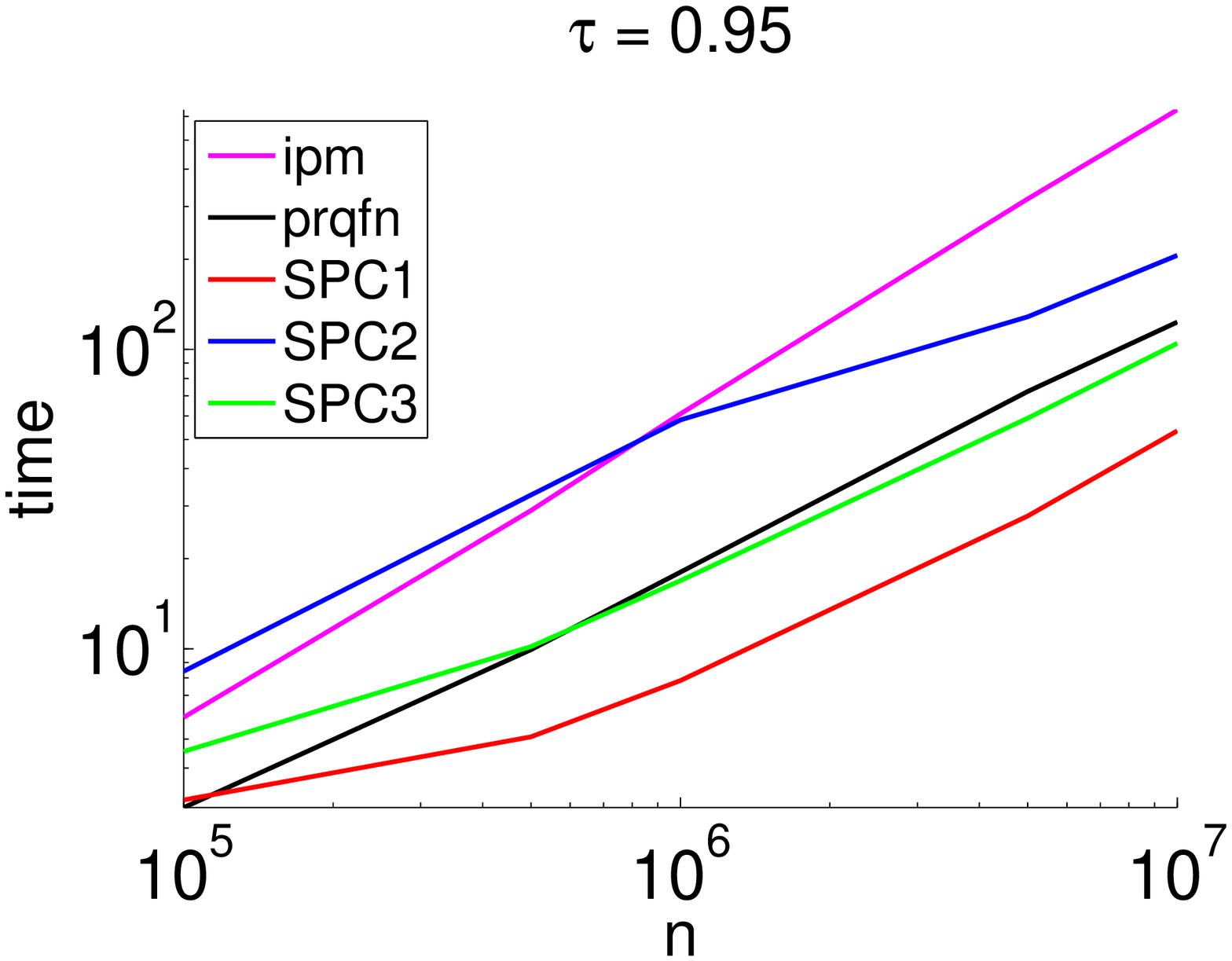}
 }
 \\
 \multicolumn{3}{c}{\bf Skewed data with $d = 50$} 
 \\
  \subfigure[$\tau  = 0.5$]{
  \includegraphics[width = 0.3\textwidth] {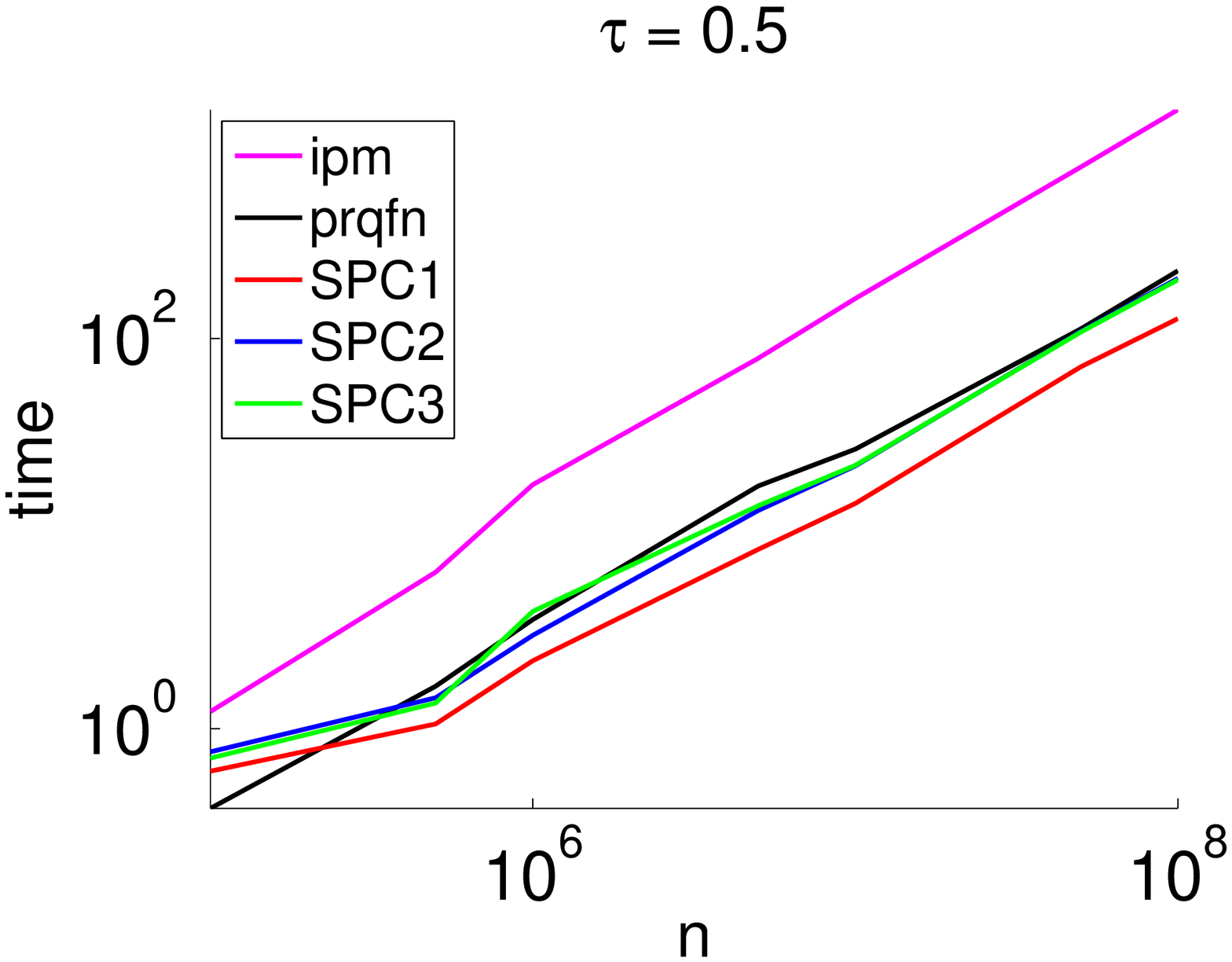}
 }
 &
\subfigure[$\tau  = 0.75$]{
   \includegraphics[width=0.3\textwidth] {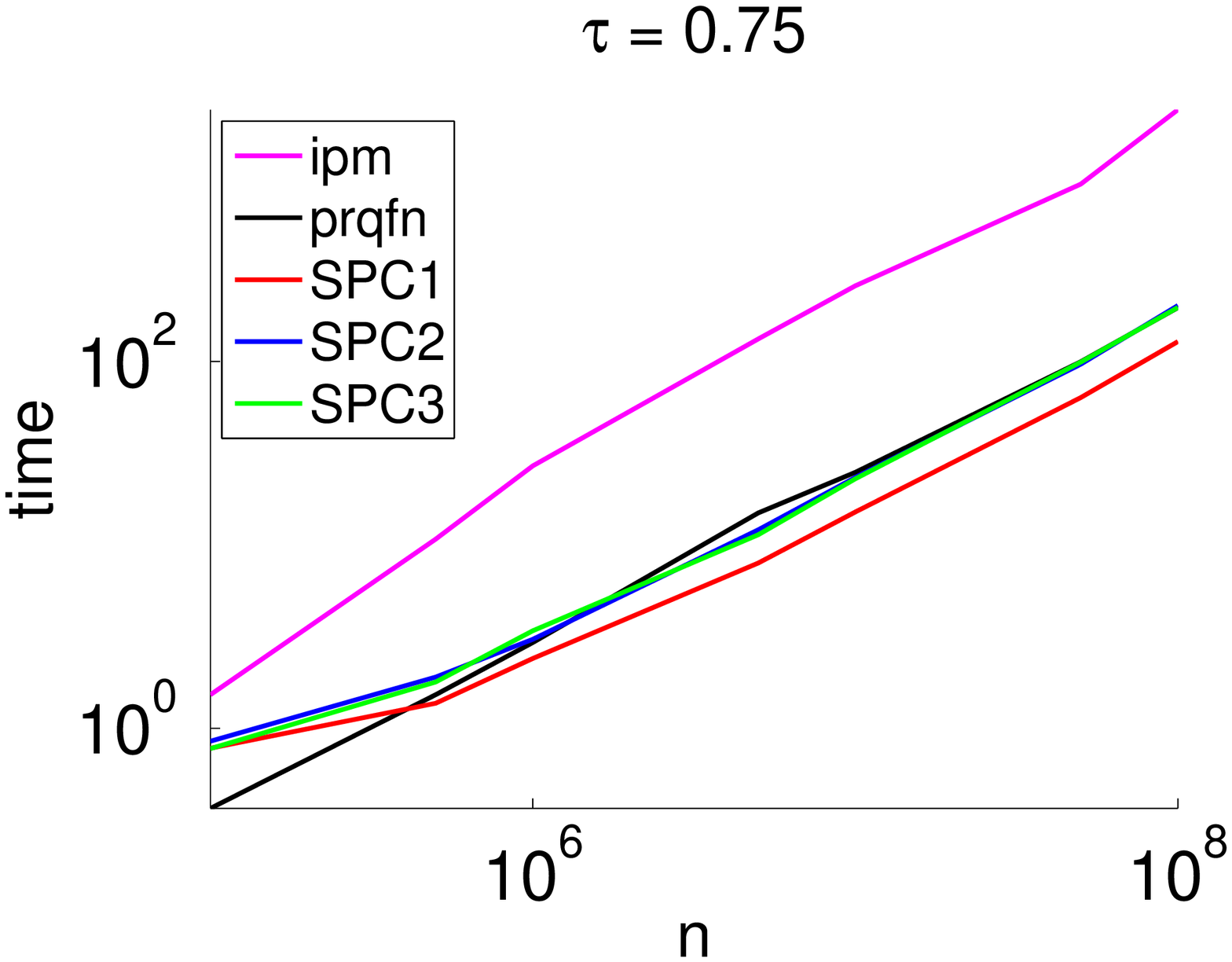}
 }
 &
 \subfigure[$\tau  = 0.95$]{
   \includegraphics[width=0.3\textwidth] {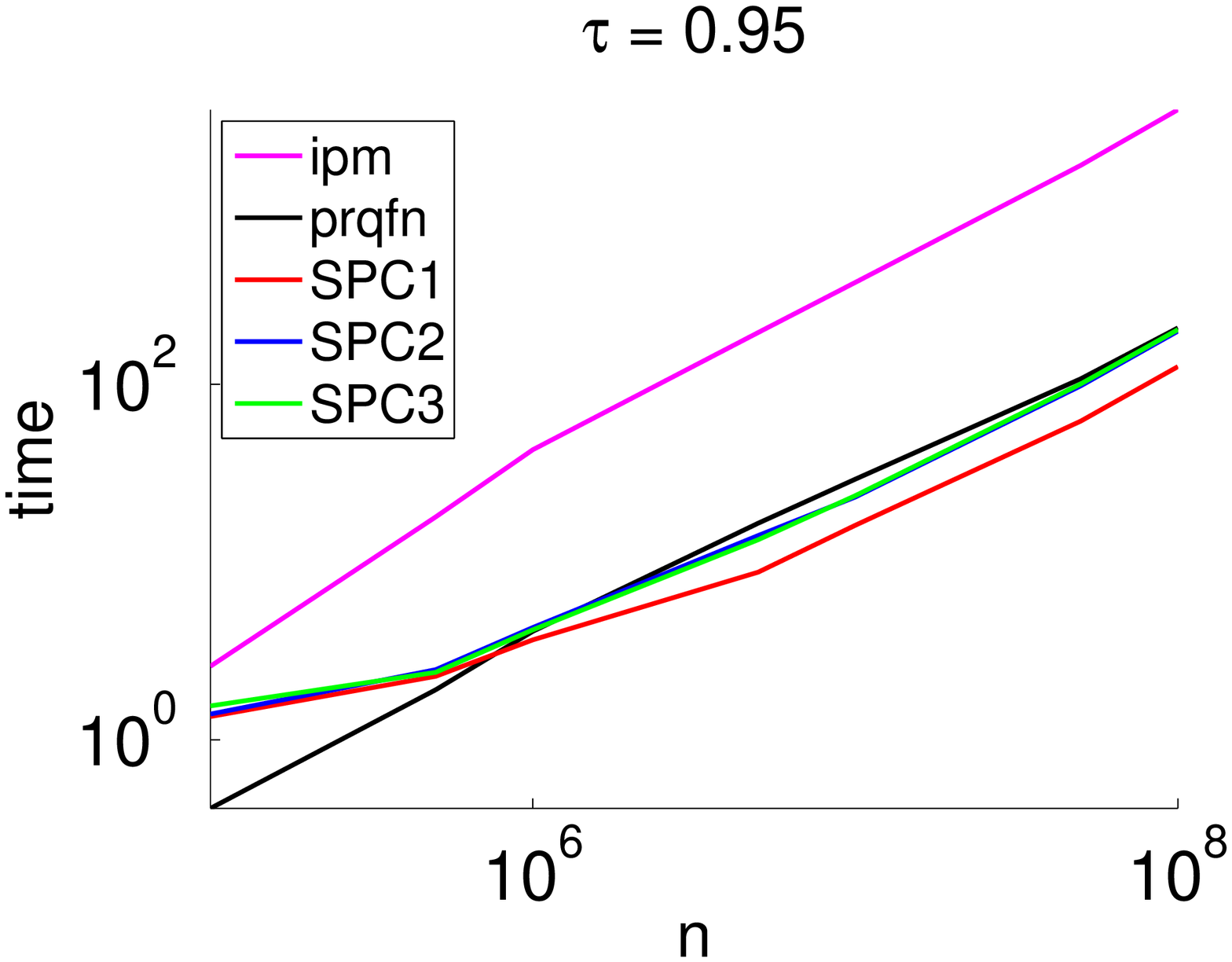}
 }
 \\
  \multicolumn{3}{c}{\bf Gaussian data with $d = 11$}
 \\
  \subfigure[$\tau  = 0.5$]{
   \includegraphics[width=0.3\textwidth] {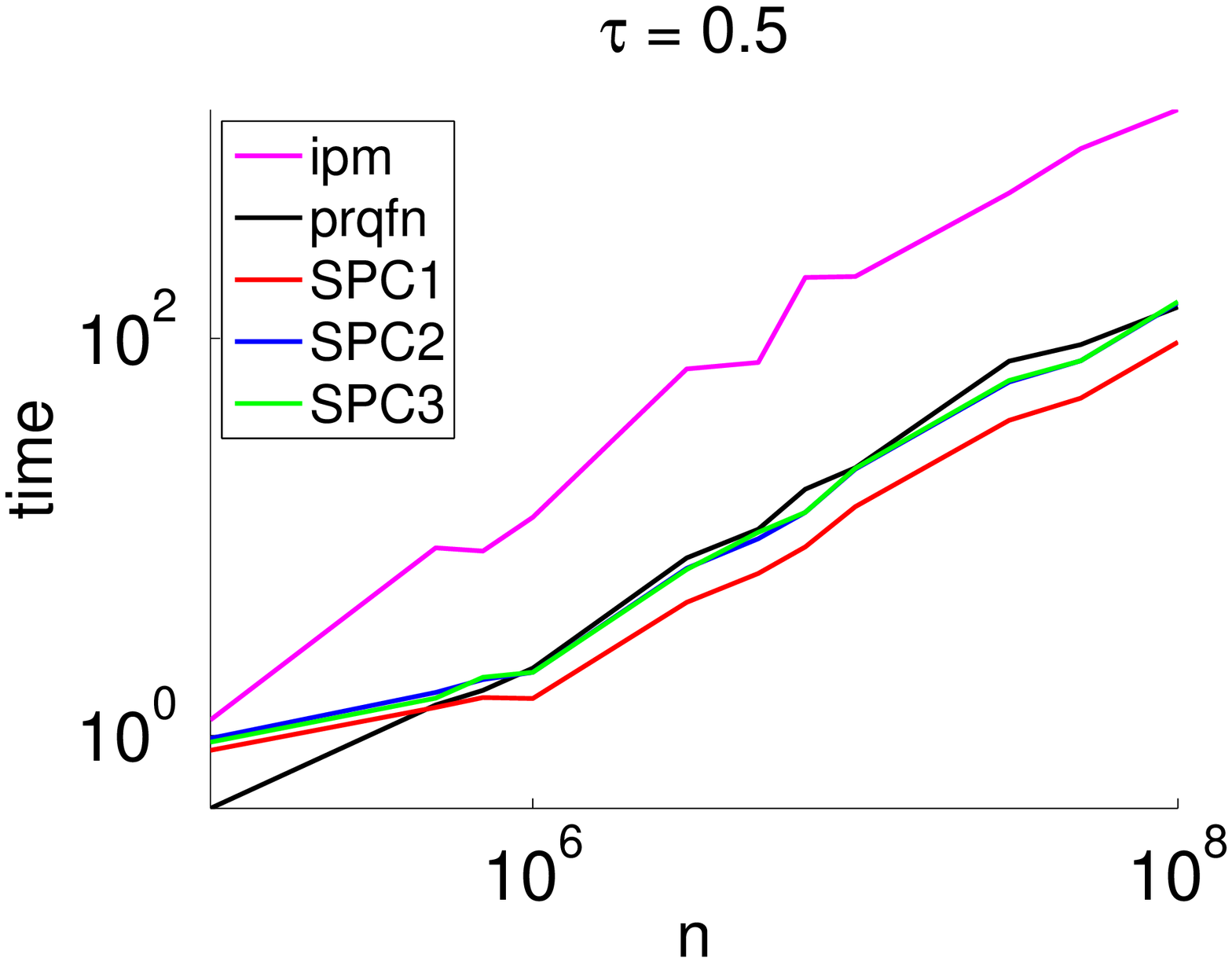}
 }
 &
\subfigure[$\tau  = 0.75$]{
   \includegraphics[width=0.3\textwidth] {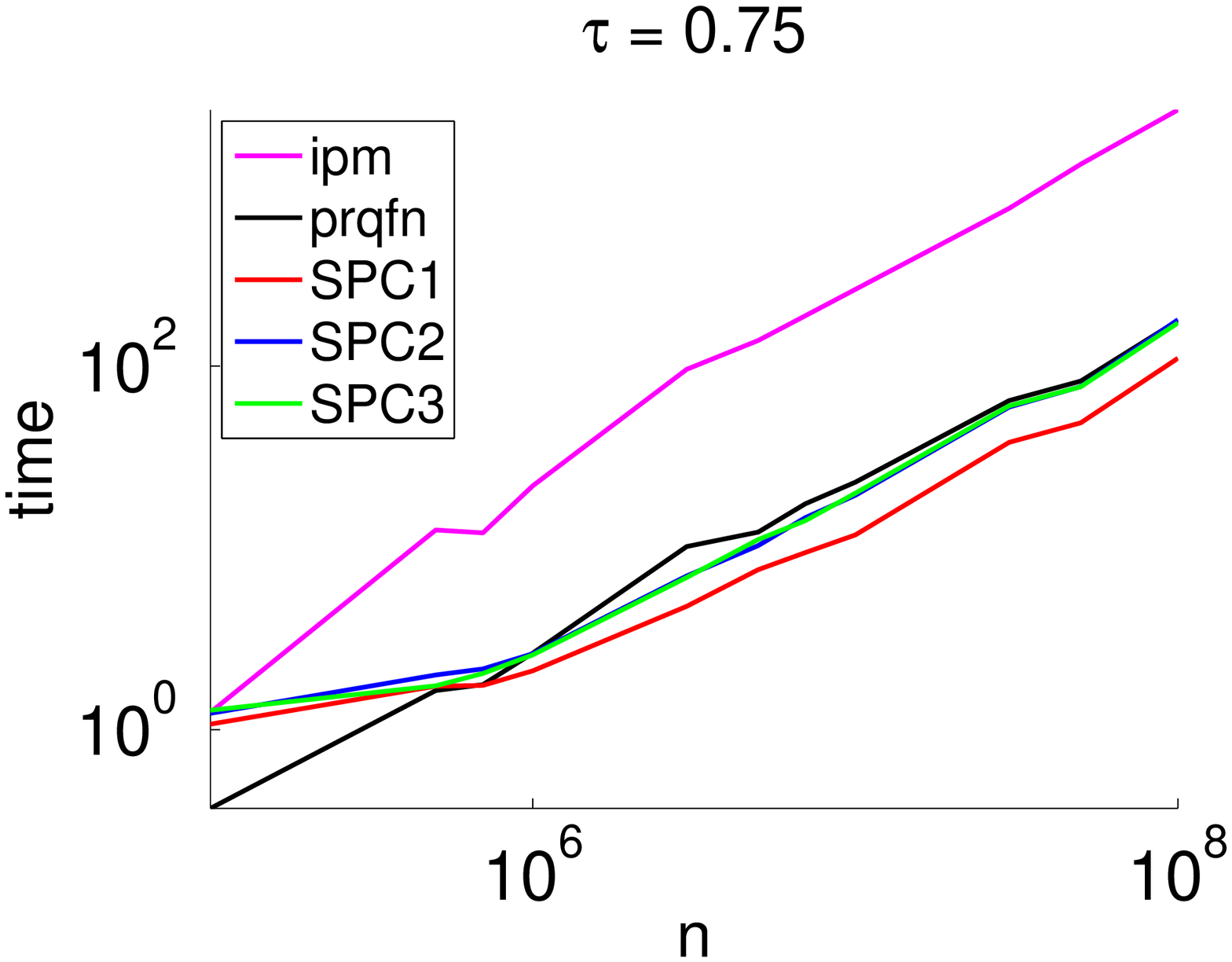}
 }
 &
 \subfigure[$\tau  = 0.95$]{
   \includegraphics[width=0.3\textwidth] {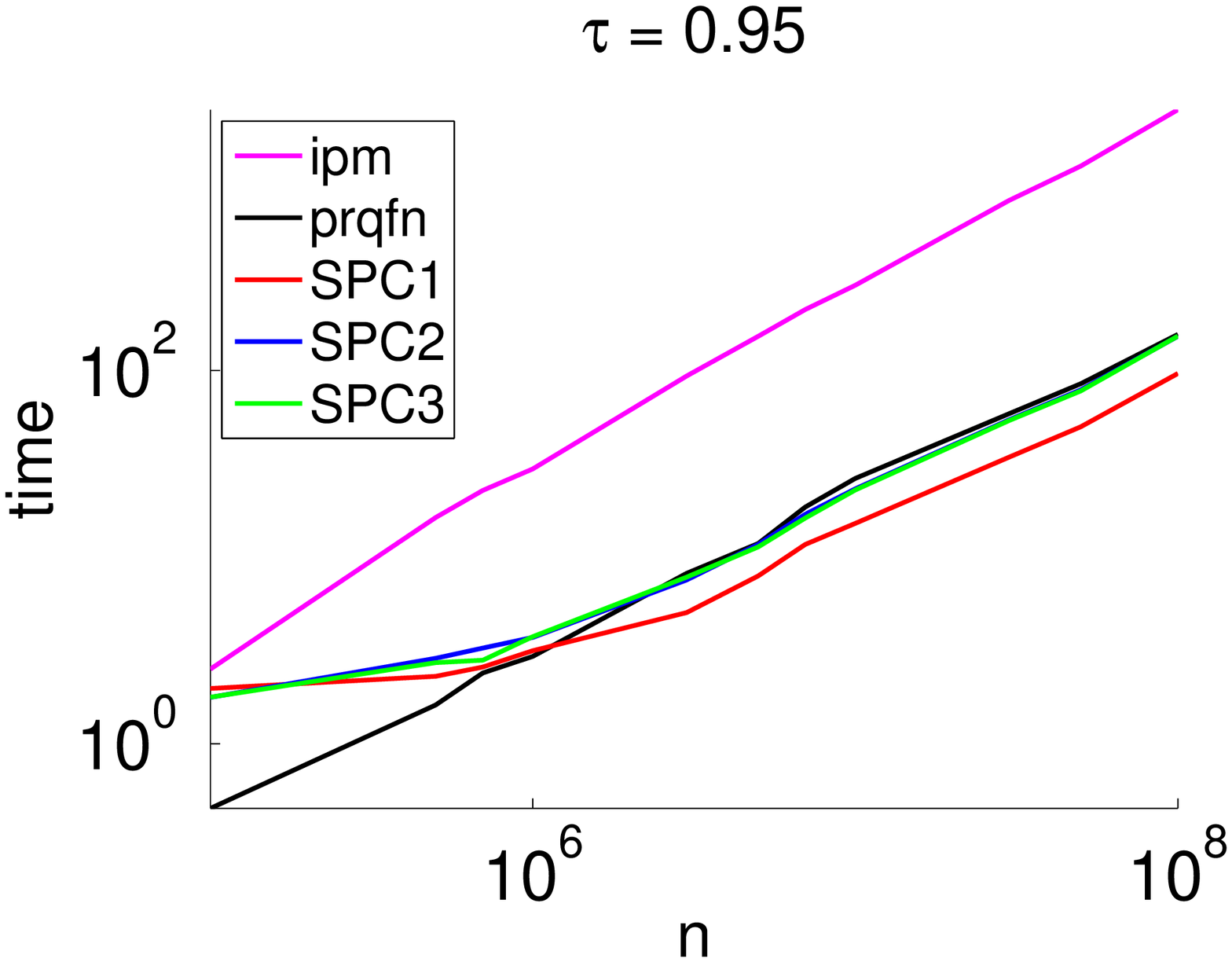}
 }
 \\
 \multicolumn{3}{c}{\bf Replicated census data with $d = 11$} 
  \\
  \end{tabular}
 \end{center}
 \caption{The running time for five methods
 (\texttt{ipm}, \texttt{prqfn}, SPC1, SPC2 and SPC3) on the same data set,
 with $d$ fixed and $n$ changing.
 The sampling size $s = 5e4$, and the three columns correspond to 
 $\tau=0.5, 0.75, 0.95$, respectively.
}
  \label{time_n}
\end{figure}

From the plots in Figure~\ref{time_n} we see, SPC1 runs faster than any 
other methods across all the data sets, in some cases significantly so. 
SPC2, SPC3 and \texttt{prqfn} perform similarly in most cases, and they 
appear to have a linear rate of increase.
Also, the relative performance between each method does not vary a lot as
the data type changes.

Notice that for the skewed data, when $d = 50$, SPC2 runs much slower than 
when $d = 10$.
The reason for this is that, for conditioning-based methods, the running 
time is composed of two parts, namely, the time for conditioning and the 
time for solving the subproblem.
For SPC2, an ellipsoid rounding needs to be applied on a smaller data set
whose larger dimension is a polynomial of $d$.
When the sampling size $s$ is small, \emph{i.e.}, the size of the subproblem 
is not too large, the dominant running time for SPC2 will be the time for ellipsoid 
rounding, and as $d$ increase (by, say, a factor of $5$) we expect a worse 
running time.
Notice also that, for all the methods, the running time does not vary a lot 
when $\tau$ changes.
Finally, notice that all the conditioning-based methods run faster on 
skewed data, especially when $d$ is small.
The reason is that the running time for these three methods is of the order 
of input-sparsity time, and the skewed data are very sparse.

\paragraph{When the lower dimension $d$ changes\\}

Finally, we will describe the scaling of the running time as the lower 
dimension $d$ changes.
To do so, we fixed $n = 1e6$ and the sampling size $s = 1e4$.
We let all five methods run on the data set with $d$ varying from $5$ up to 
$180$.
When $d \approx 200$, the scaling was such that all the methods 
except for SPC1 and SPC3 became too expensive.
Thus, we let only SPC1 and SPC3 run on additional data sets with $d$ up to 
$270$.
The plots are shown in Figure~\ref{time_d}.

 \begin{figure}[h!tbp]
  \begin{center}
   \begin{tabular}{ccc}
  \subfigure[$\tau  = 0.5$]{
   \includegraphics[width=0.3\textwidth] {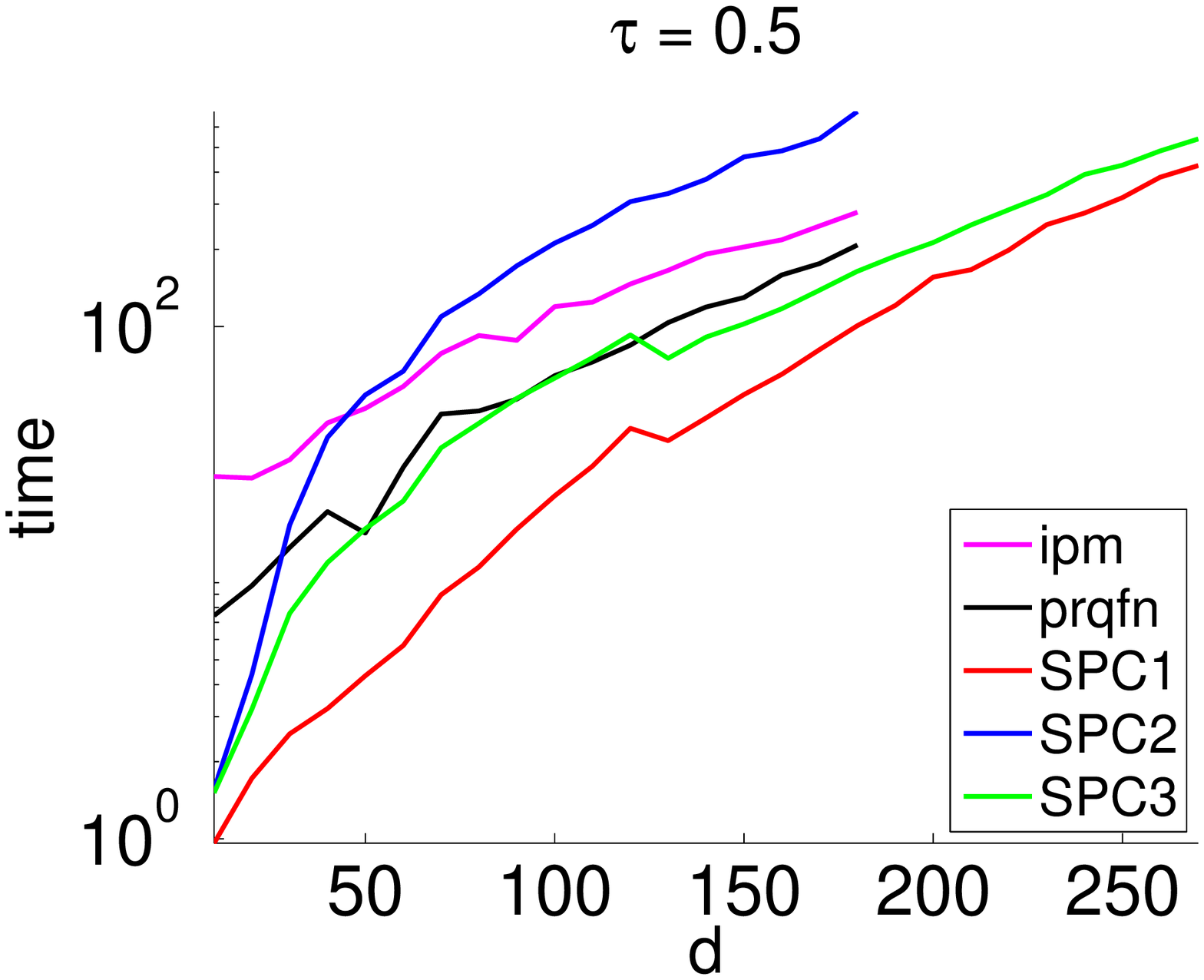}
 }
 &
\subfigure[$\tau  = 0.75$]{
   \includegraphics[width=0.3\textwidth] {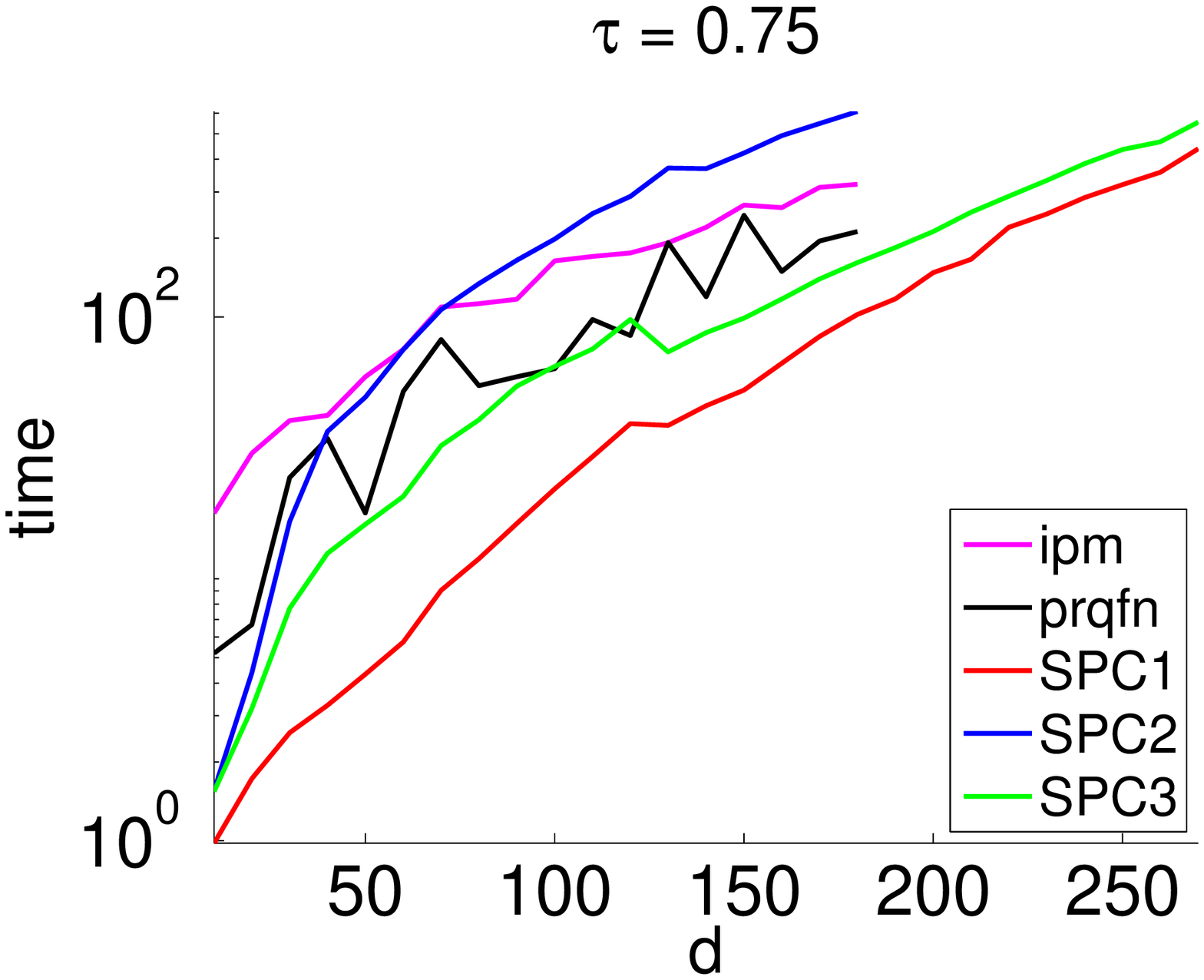}
 }
 &
 \subfigure[$\tau  = 0.95$]{
   \includegraphics[width=0.3\textwidth] {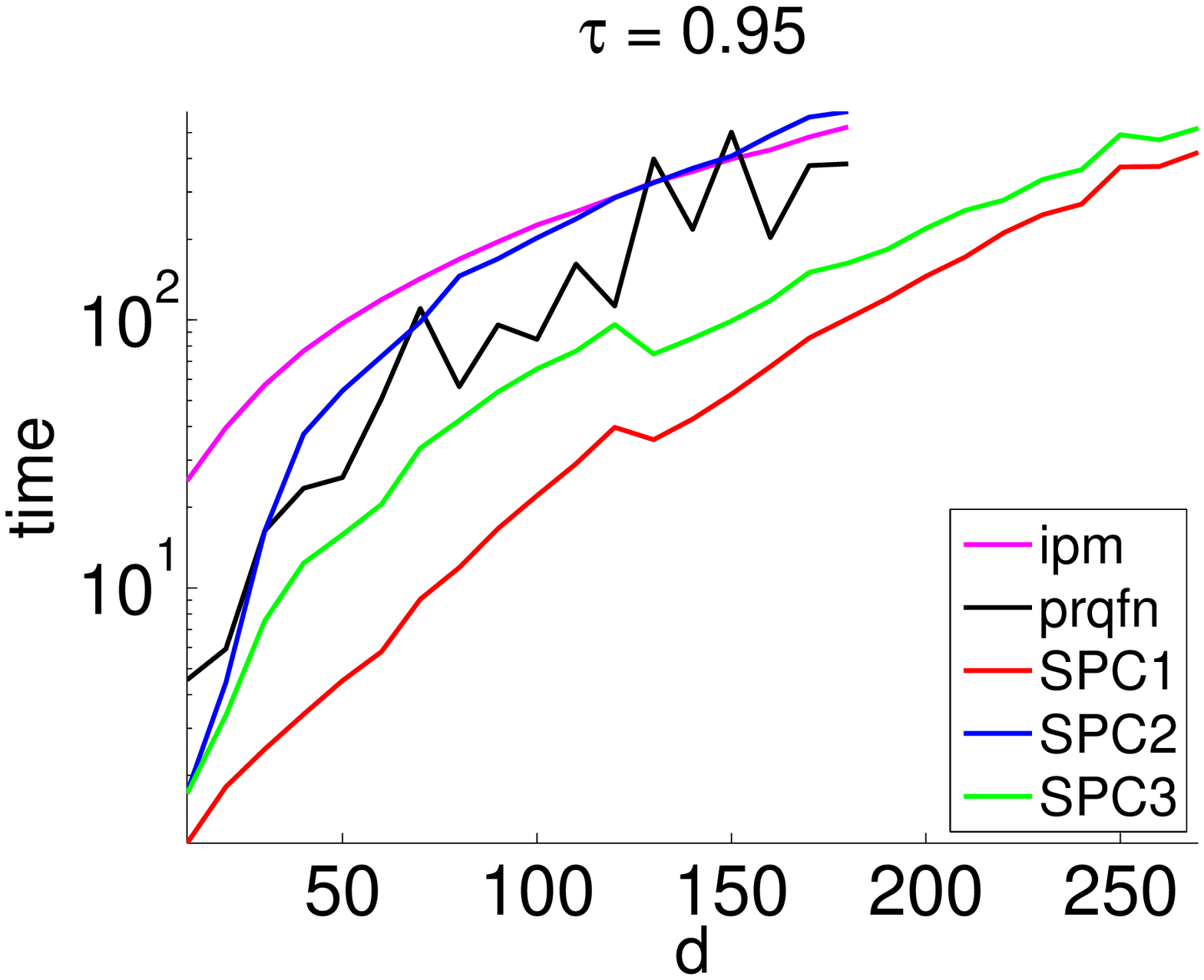}
 }
  \end{tabular}
  \end{center}
  \caption{The running time for five methods
    (\texttt{ipm}, \texttt{prqfn}, SPC1, SPC2, and SPC3) for solving skewed data, with $n = 1e6$, $s = 1e4$,
     when $d$ varies. SPC1 and SPC3 show better scaling than other methods when $d<180$.
      For this reason, we keep running the experiments
     for SPC1 and SPC3 until $d = 270$.
     When $d<100$, the three conditioning-based methods can yield 2-digit accuracy.
     When for $d \in [100, 180]$, they can yield 1-digit accuracy.
     }
  \label{time_d}
 \end{figure}

From the plots in Figure~\ref{time_d}, we can see that when $d < 180$, SPC1 
runs significantly faster than any other method, followed by SPC3 and 
\texttt{prqfn}.
The performance of \texttt{prqfn} is quite variable.
The reason for this is that there is a step in \texttt{prqfn} that involves 
uniform sampling, and the number of subproblems to be solved in each time 
might vary a lot.
The scalings of SPC2 and \texttt{ipm} are similar, and when $d$ gets much 
larger, say $d>200$, they may not be favorable due to the running time.
When $d<180$, all the conditioning methods can yield at least 1-digit 
accuracy.
Although one can only get an approximation to the true solution by using 
SPC1 and SPC3, they will be a good choice when $d$ gets even larger, say up 
to several hundred, as we shown in Figure~\ref{time_d}.
We note that we could let $d$ get even larger for SPC1 and SPC3, 
demonstrating that SPC1 and SPC3 is able to run with a much larger lower 
dimension than the other methods.

\noindent
\textbf{Remark.}
One may notice a slight but sudden change in the running time for SPC1 and 
SPC3 at $d \approx 130$.
After we traced down the reason, we found out that the difference come from the time in the conditioning step
(since the subproblems they are solving have similar size), especially the time for performing the QR factorization.
At this size, it will be normal to take more time to factorize a slightly smaller matrix due to the structure of cache line, and it is for this reason that we see that minor decrease in running time with increasing $d$.
We point out that the running time of our conditioning-based algorithm is mainly affected by the time for the conditioning step.
That is also the reason why it does not vary a lot when $\tau$ changes.


\subsection{Evaluation on solution of Census data}
\label{census_data}

Here, we will describe more about
the accuracy on the census data when SPC3 is applied to it.
The size of the census data is roughly $5e6 \times 11$.

We will generate plots that are similar to those appeared in \cite{KH01}.
For each coefficient, we will compute a few quantities of it, as a function of $\tau$,
when $\tau$ varies from 0.05 to 0.95.
We compute a point-wise 90 percent confidence interval
for each $\tau$ by bootstrapping.
These are shown as the shaded area in each subfigure.
Also, we compute the quartiles of the approximated solutions by using SPC3 from 200 independent trials
with sampling size $s = 5e4$
to show how close we can get to the confidence interval.
In addition, we also show the solution to Least Square regression (LS)
and Least Absolute Deviations regression (LAD) on the same problem. 
The plots are shown in Figure~\ref{ci}.

\begin{figure}[h!tbp]
 \begin{center}
  \begin{tabular}{ccc}
\subfigure[Intercept]{
   \includegraphics[width=0.3\textwidth] {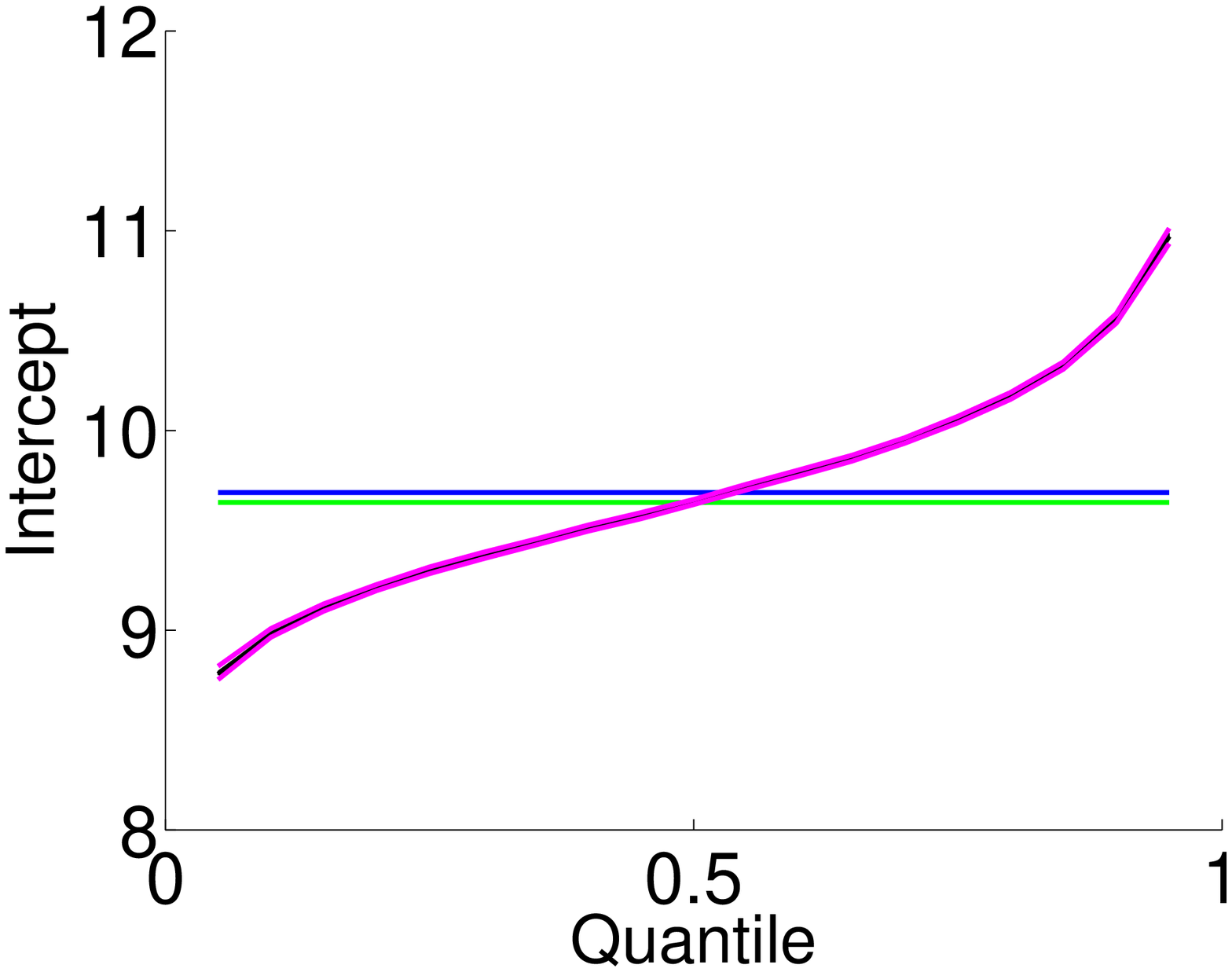}
 }
 &
\subfigure[Sex]{
   \includegraphics[width=0.3\textwidth] {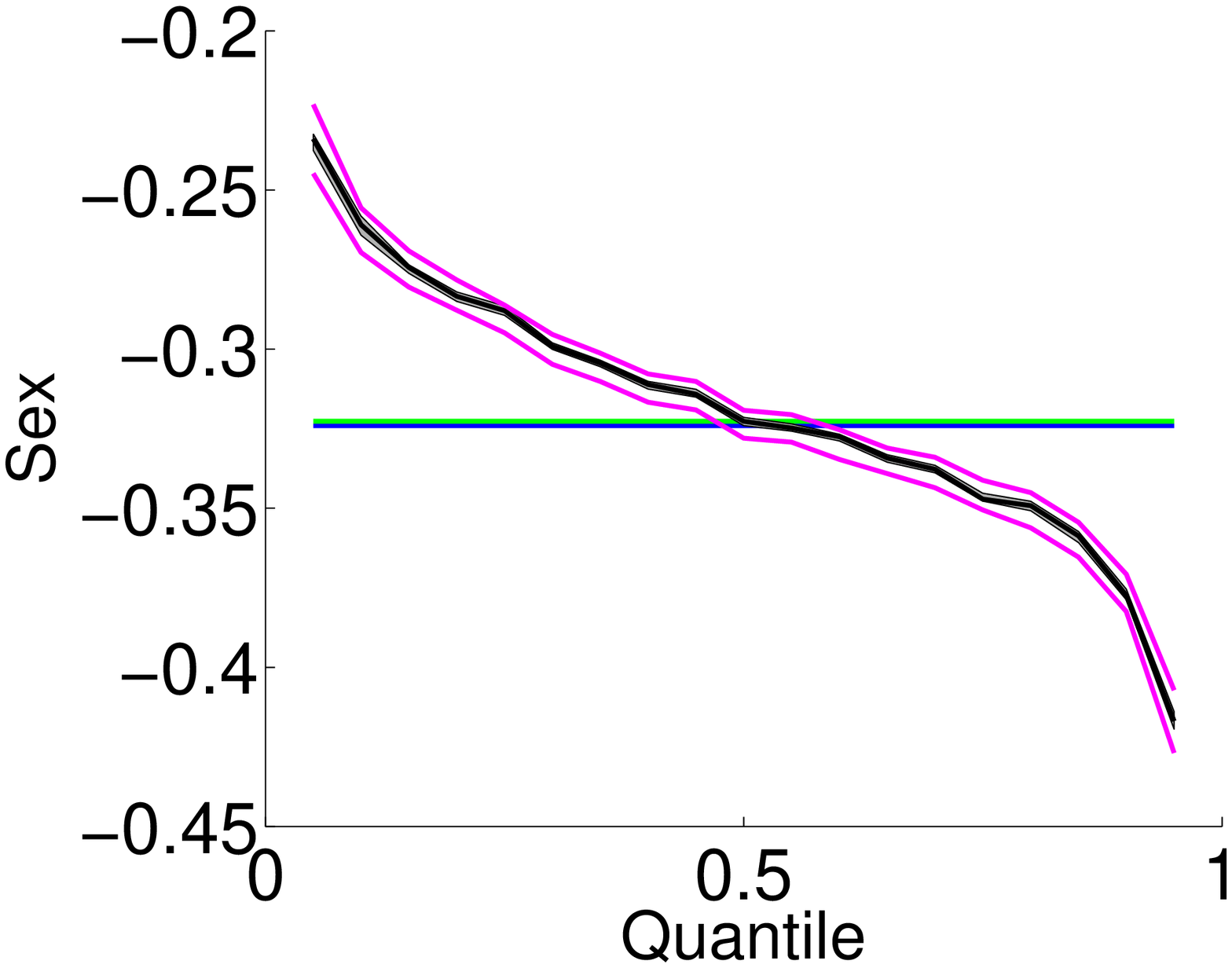}
 }
 &
\subfigure[Age $\in [30,40)$]{
   \includegraphics[width=0.3\textwidth] {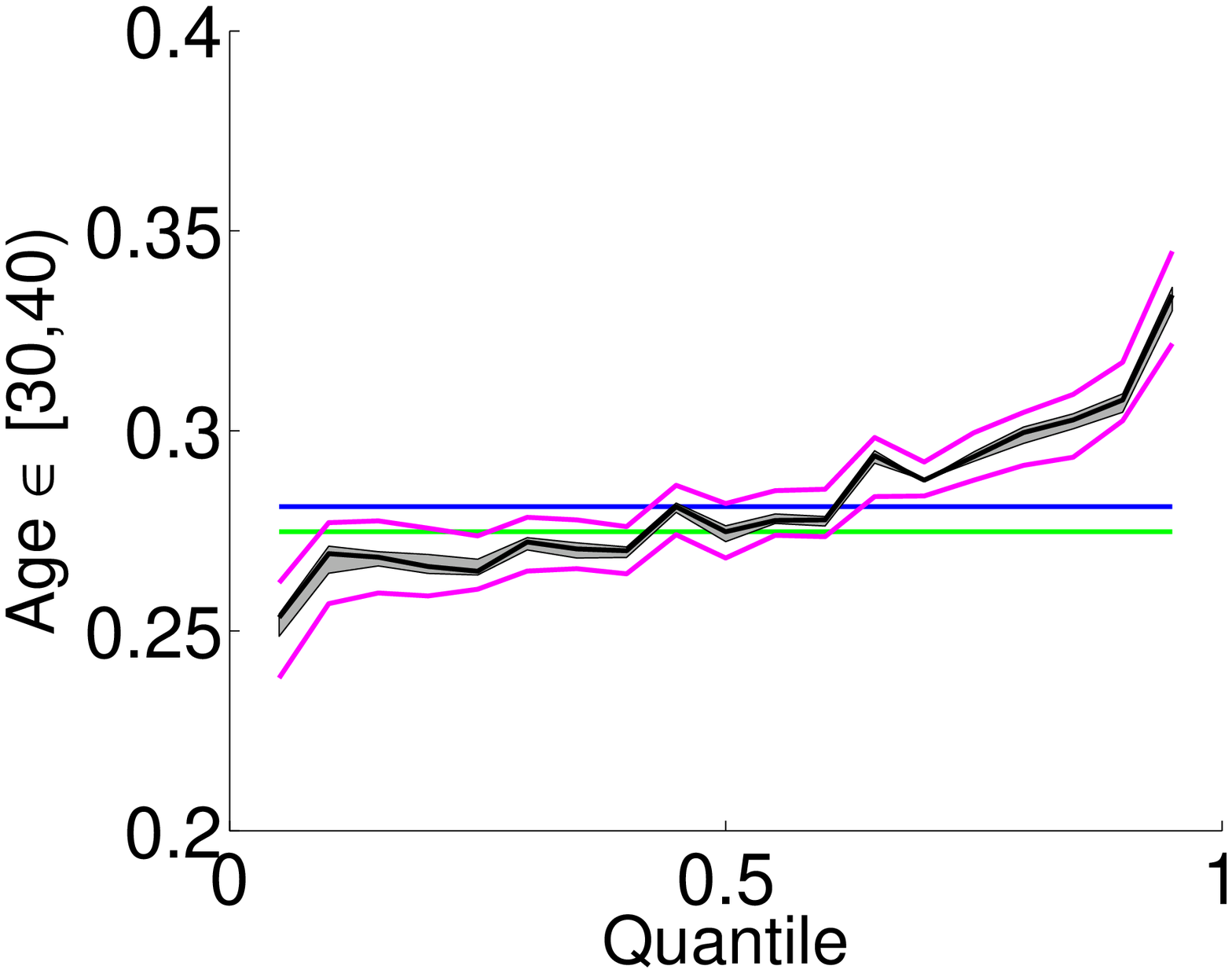}
 }
 \\
\subfigure[Age $\in [40,50)$]{
   \includegraphics[width=0.3\textwidth] {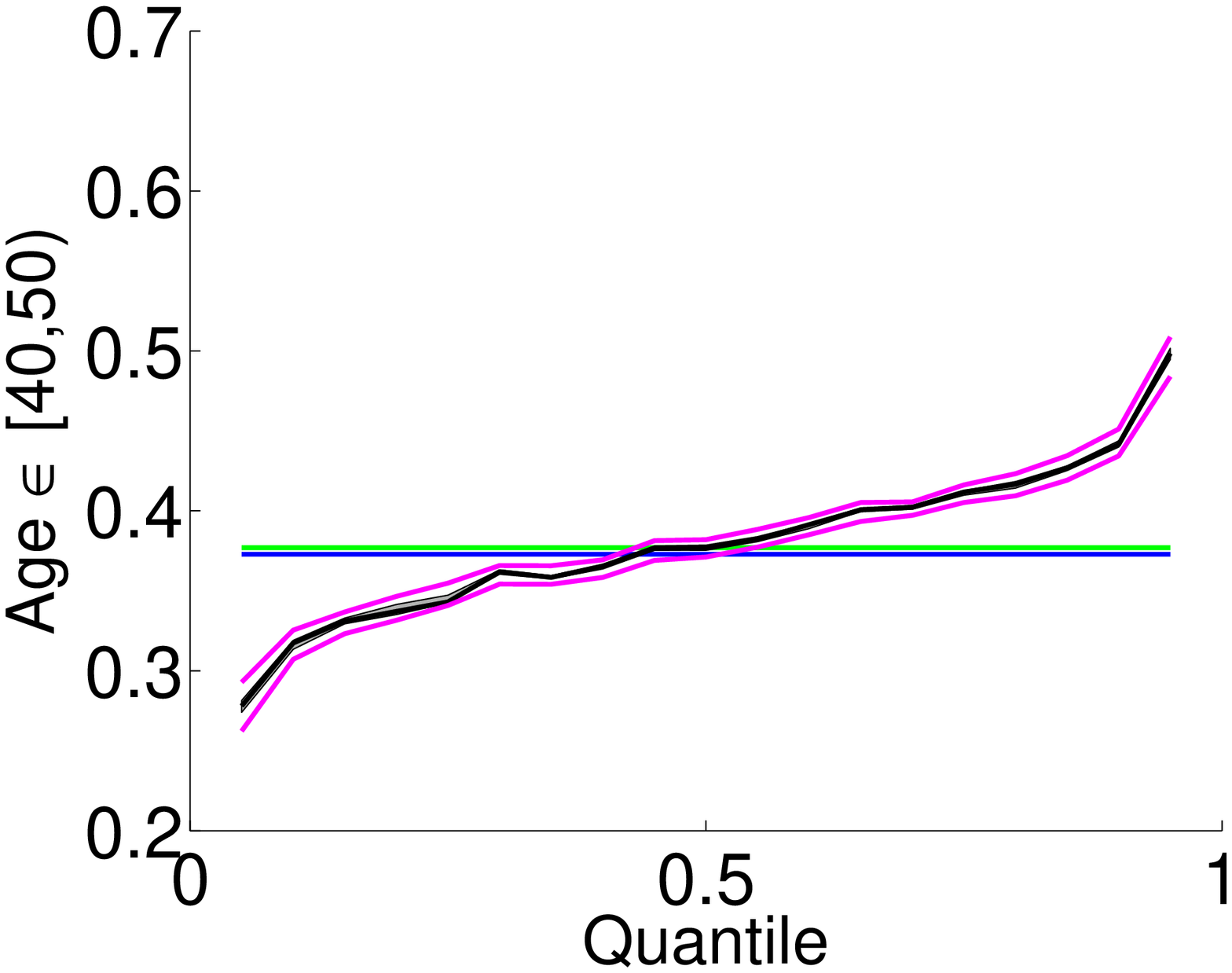}
 }
 &
\subfigure[Age $\in [50,60)$]{
   \includegraphics[width=0.3\textwidth] {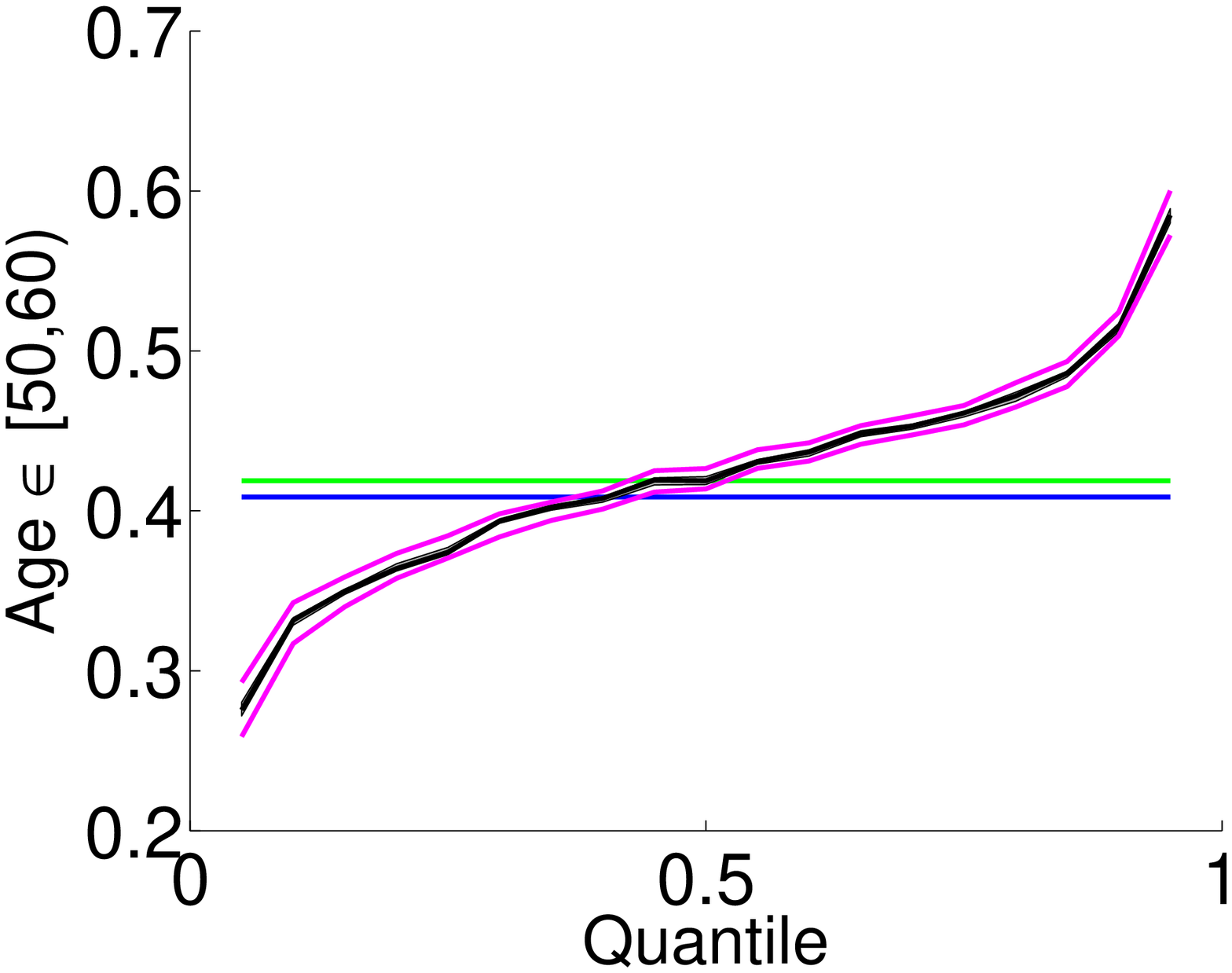}
 }
 &
\subfigure[Age $\in [60,70)$]{
   \includegraphics[width=0.3\textwidth] {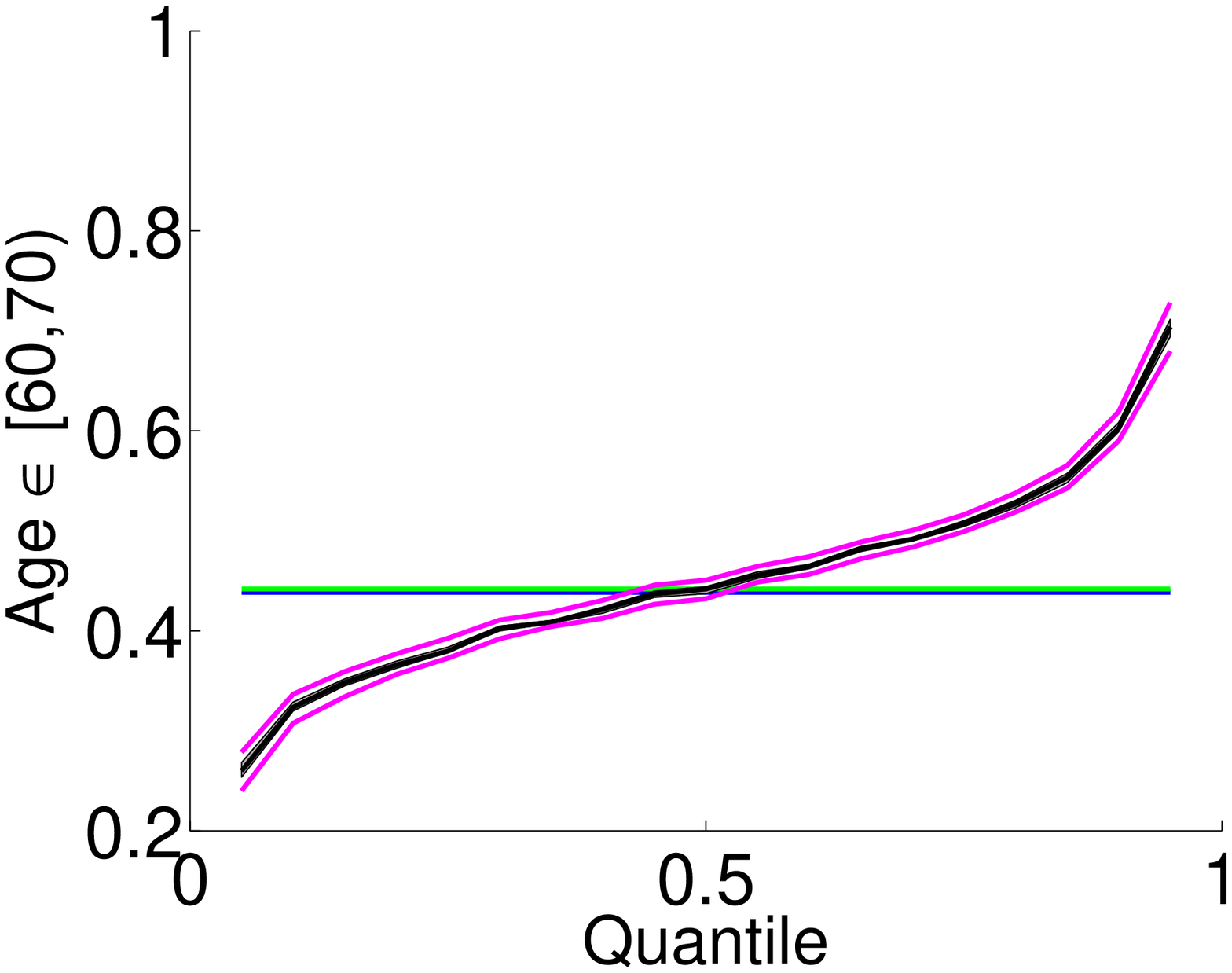}
 }
 \\
 \subfigure[Age $\geq 70$]{
   \includegraphics[width=0.3\textwidth] {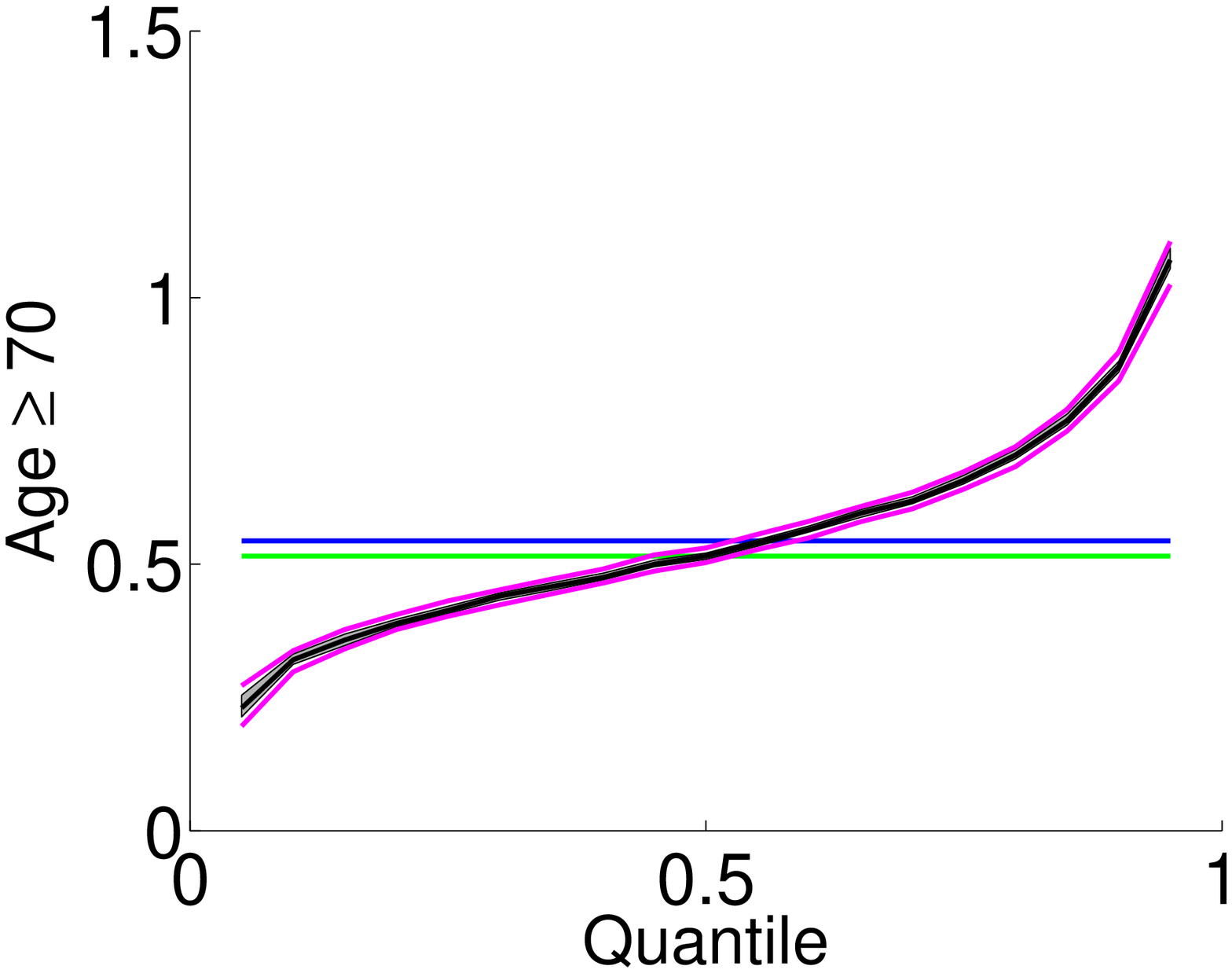}
 }
 &
\subfigure[Non\_white]{
   \includegraphics[width=0.3\textwidth] {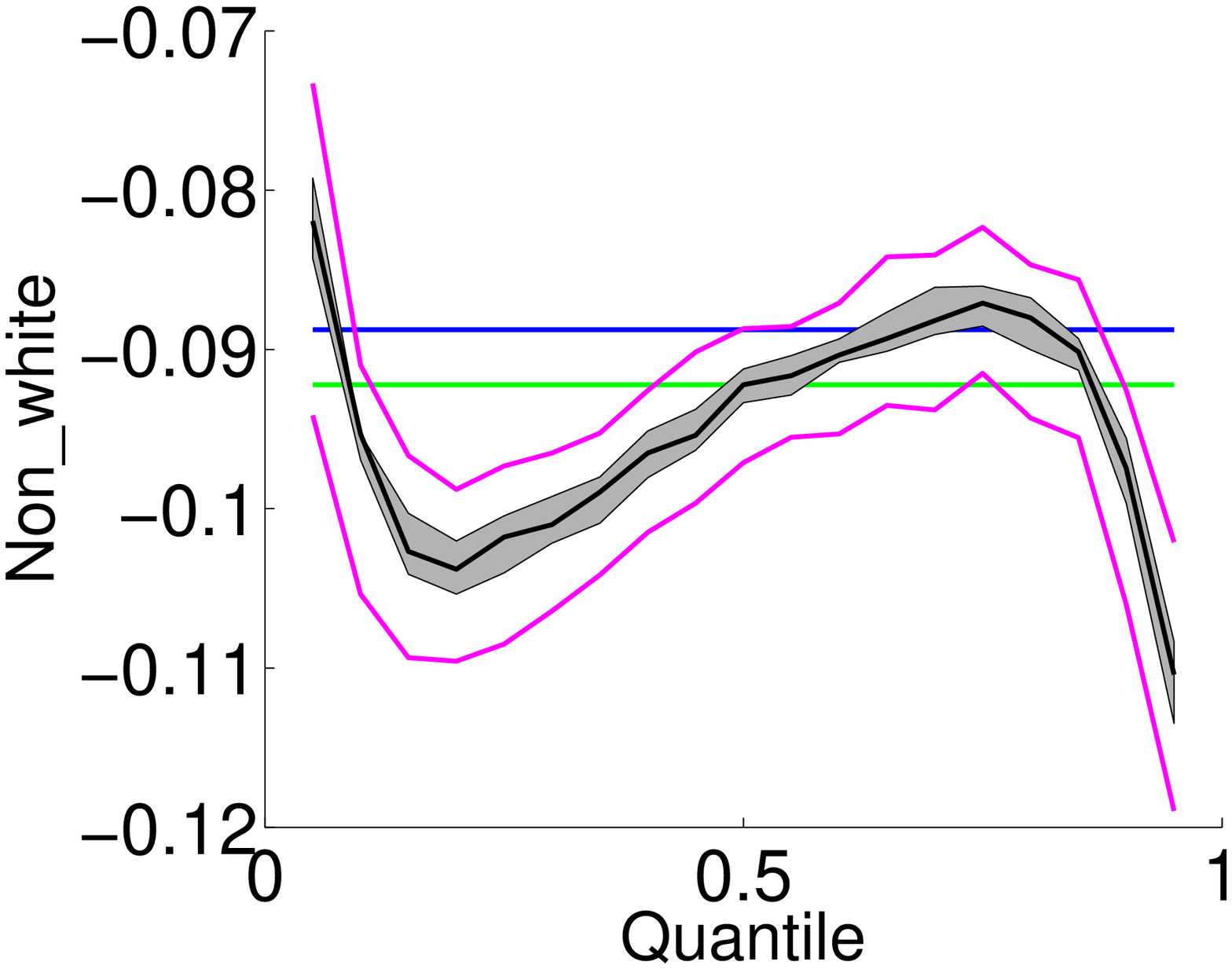}
 }
 &
\subfigure[Unmarried]{
   \includegraphics[width=0.3\textwidth] {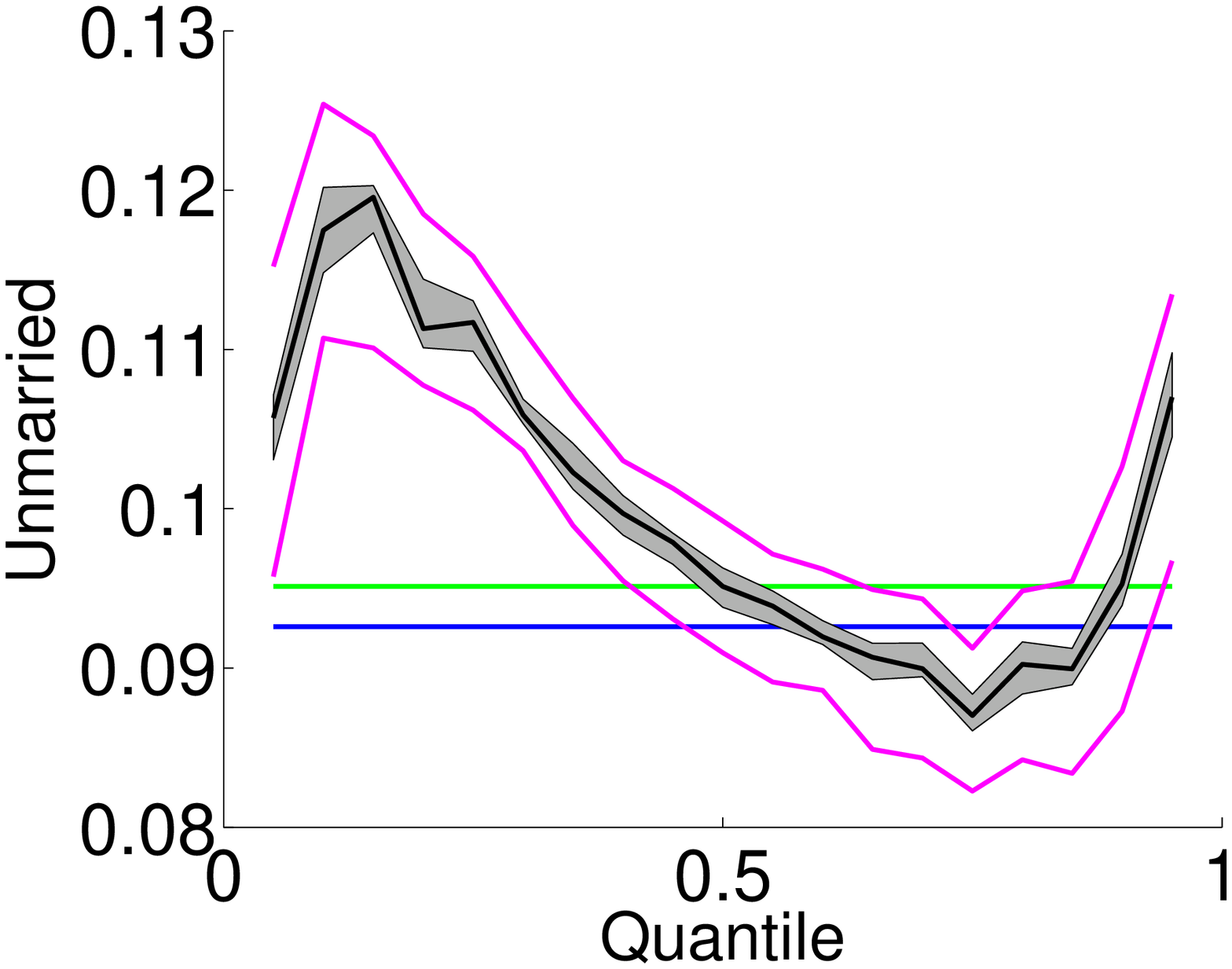}
 }
 \\
\subfigure[Education]{
   \includegraphics[width=0.3\textwidth] {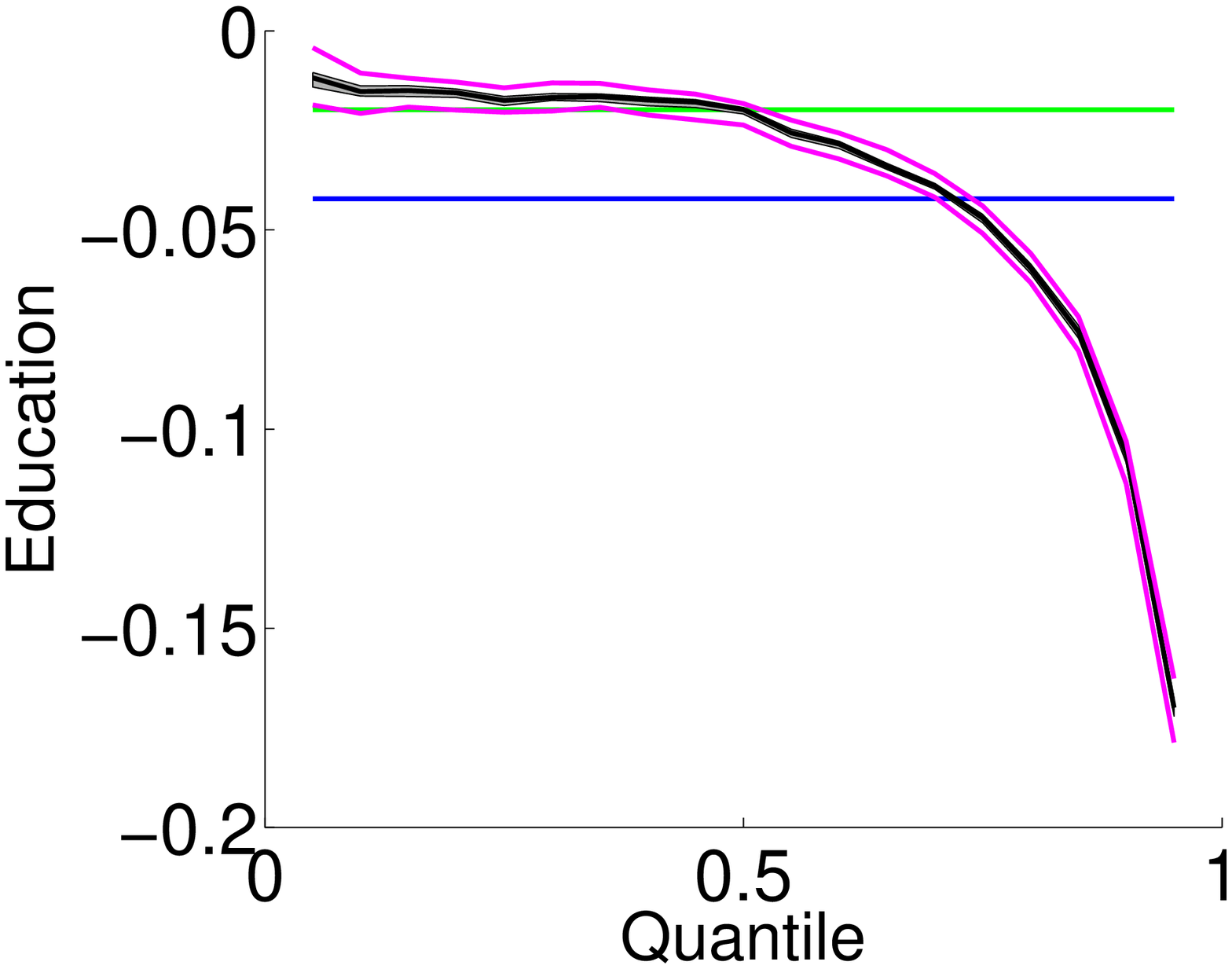}
 }   
 &
 \subfigure[Education$^2$]{
   \includegraphics[width=0.3\textwidth] {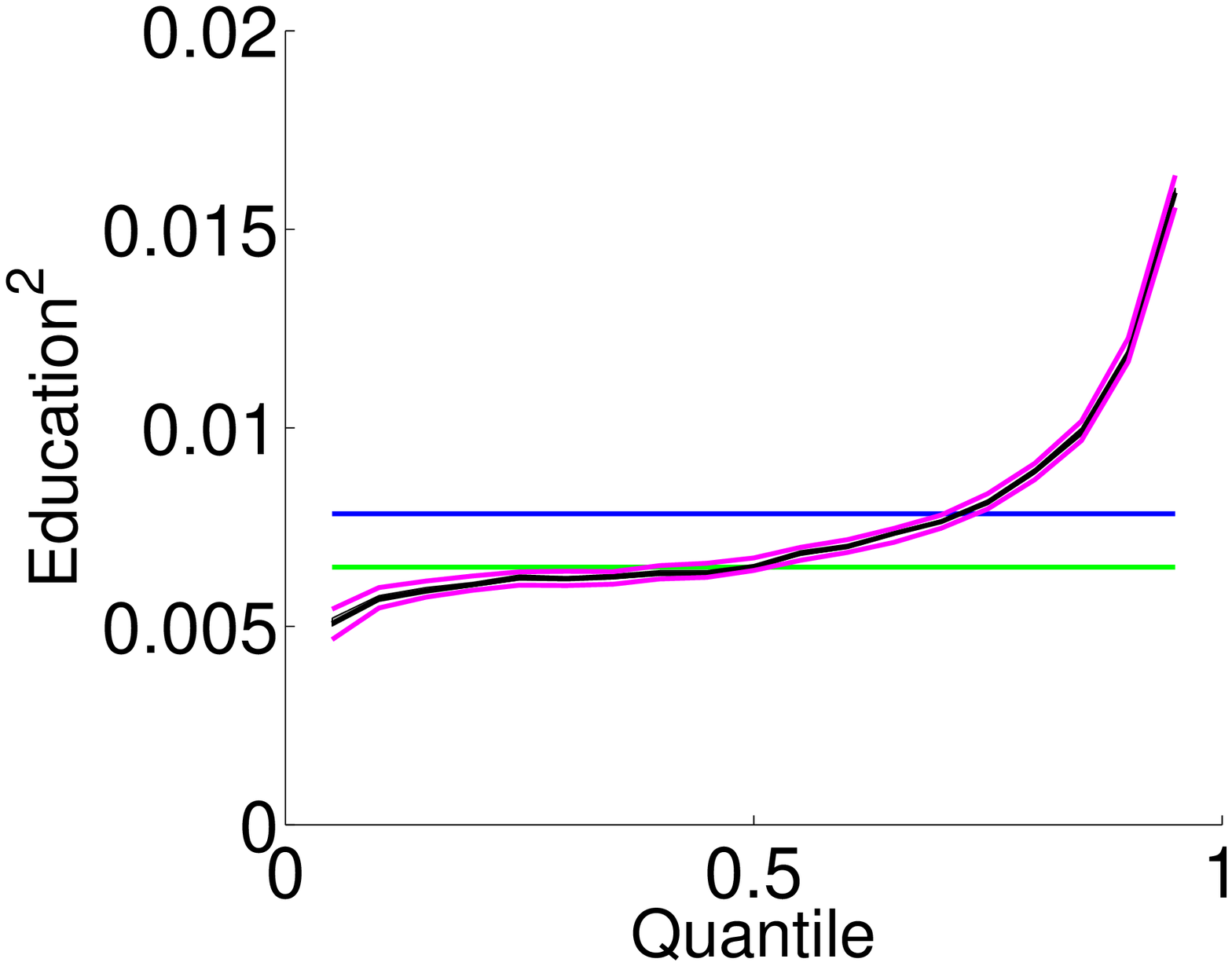}
 }  
 &
 \subfigure[Legend]{
   \includegraphics[width=0.3\textwidth] {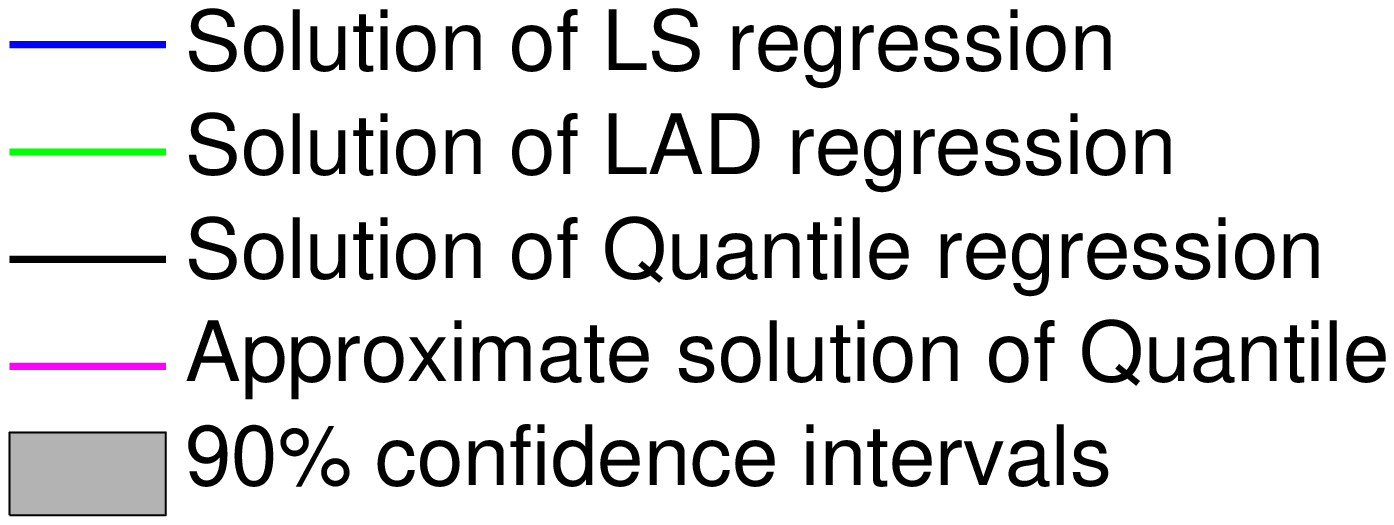}
 }  

  \end{tabular}
 \end{center}
 \caption{Each subfigure is associated with a coefficient in the census data. 
   The shaded area shows a point-wise 90\% confidence interval.
    The black curve inside is the true solution when $\tau$ changes from 0.05 to 0.95.
     The blue and green lines correspond to the $\ell_2$ and $\ell_1$ solution, respectively.
     The two magenta curves show the first and third quartiles of solutions obtained 
     by using SPC3, among 200 independent trials with sampling size $s = 5e4$ (about 1\% of the original data).
      }
   \label{ci}
 \end{figure}

From these plots we can see that, although the two quartiles are not inside
the confidence interval, they are quite close, even for this value of $s$.
The sampling size in each trial is only $5e4$ which is about 1 percent of the original data; while
for bootstrapping, we are resampling the same number of rows as in the original matrix with replacement.
In addition, the median of these 50 solutions is in the shaded area and close to the true solution.
Indeed, for most of the coefficients, SPC3 can generate 2-digit accuracy.
Note that we also computed the exact values of the quartiles; we don't 
present them here since 
they are very similar to those in Table~\ref{census_result} below in terms of accuracy.
See Table~\ref{census_result} in Section~\ref{large_emp} for more details.
All in all, SPC3 performs quite well on this real data.


\section{Empirical Evaluation on Large-scale Quantile Regression}
\label{large_emp}

In this section, we continue our empirical evaluation with an evaluation 
of our main algorithm applied to terabyte-scale problems.
Here, the data sets are generated by ``stacking'' the medium-scale data a 
few thousand times.
Although this leads to ``redundant'' data, which may favor sampling methods, this has the advantage that it
leads terabyte-sized problems whose optimal solution at different 
quantiles are known.
At this terabyte scale, \texttt{ipm} has two major issues: memory 
requirement and running time.
Although shared memory machines with more than a terabyte RAM exist, they 
are rare in practice (now in 2013).
Instead, the MapReduce framework is the \emph{de facto} standard parallel 
environment for large data analysis.
Apache Hadoop\footnote{Apache Hadoop, \url{http://hadoop.apache.org/}}, an 
open source implementation of MapReduce, is widely-used in practice. 
Since our sampling algorithm only needs several passes through the data and it 
is embarrassingly parallel, it is straightforward to implement it on Hadoop.

For a skewed data with size $1e6 \times 50$, we stack it vertically 2500 times.
This leads to a data with size $2.5e9 \times 50$.
In order to show the evaluations similar to Figure~\ref{err_s},
we still implement SC, SPC1, SPC2, SPC3, NOCO and UNIF.
Figure~\ref{err_s_large} shows the relative errors on the replicated skewed data set by using 
the six methods.
We only show the results for $\tau = 0.5$ and $0.75$ since the conditioning methods tend to
generate abnormal results when $\tau = 0.95$.
These plots correspond with and should be compared to the four subfigures in the first two rows and columns of 
Figure~\ref{err_s}. 

\begin{figure}[h!tbp]
 \begin{center}
 \begin{tabular}{ccc}
\subfigure[$\tau = 0.5$, $|f-f^*|/|f^*|$]{
   \includegraphics[width=0.45\textwidth] {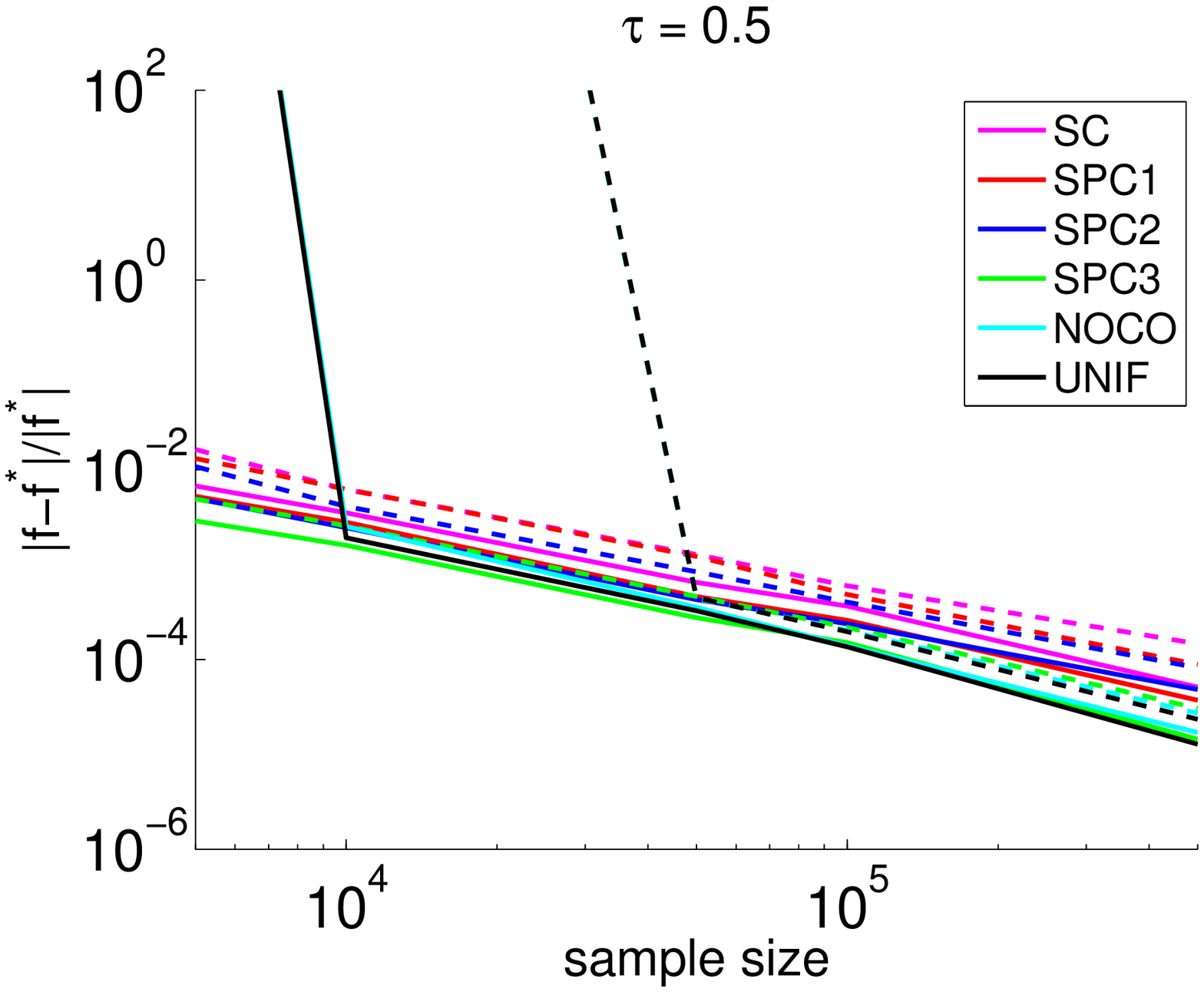}
 }
 &
 \subfigure[$\tau = 0.75$,  $|f-f^*|/|f^*|$]{
   \includegraphics[width=0.45\textwidth] {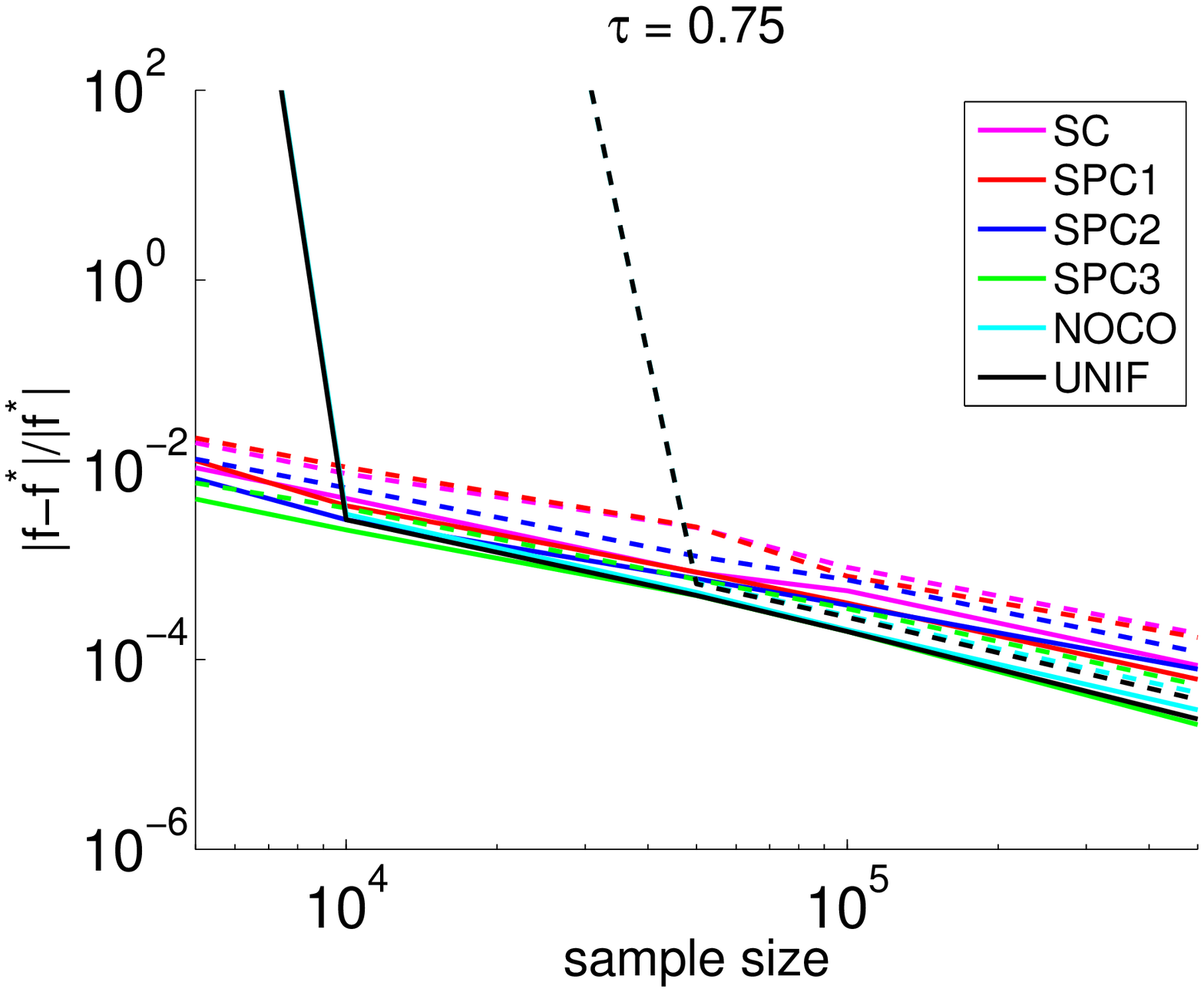}
 }
 \\
 \subfigure[$\tau = 0.5$, $\|x-x^*\|_2/\|x^*\|_2$]{
   \includegraphics[width=0.45\textwidth] {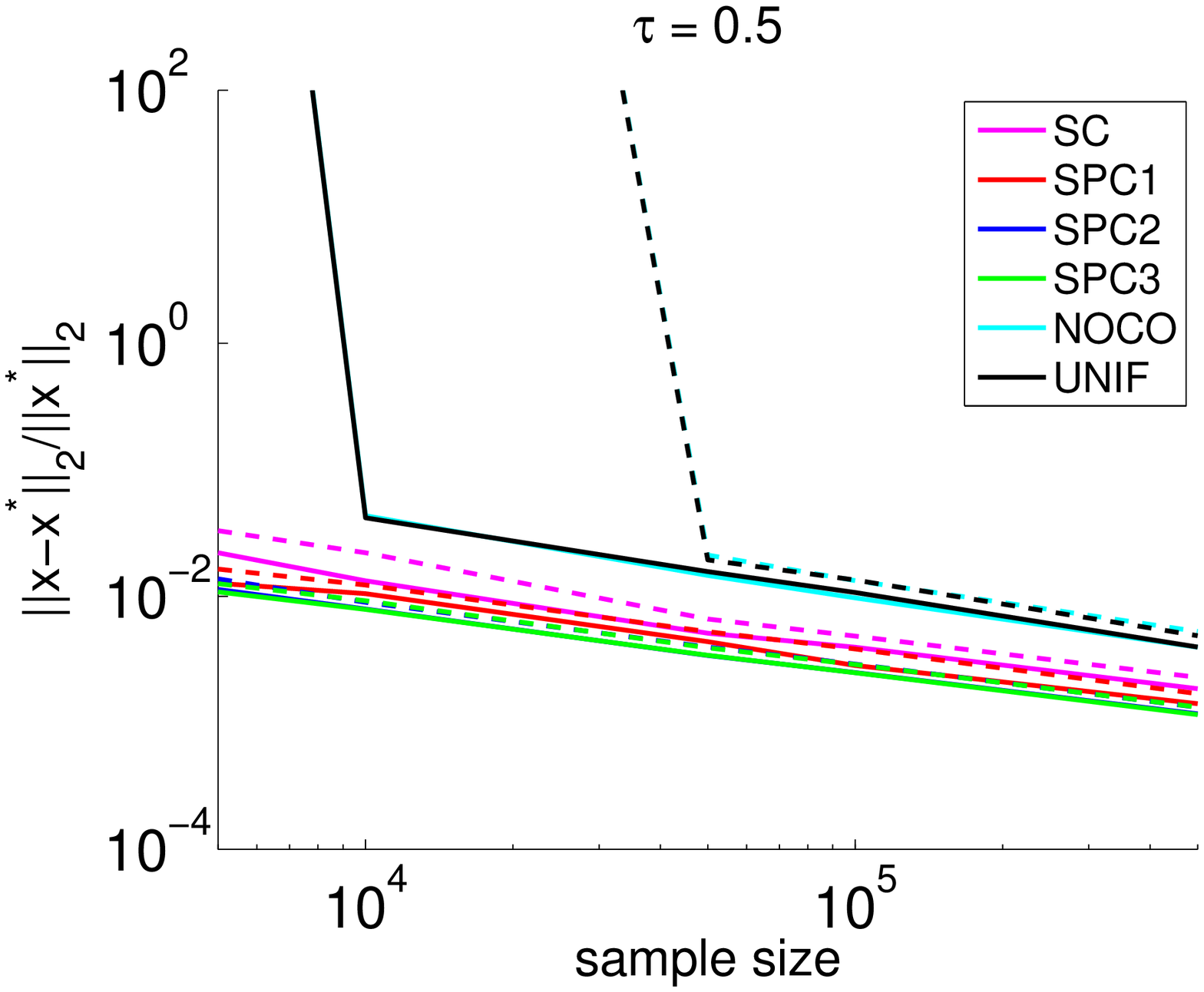}
 }
 &
 \subfigure[$\tau = 0.75$, $\|x-x^*\|_2/\|x^*\|_2$]{
   \includegraphics[width=0.45\textwidth] {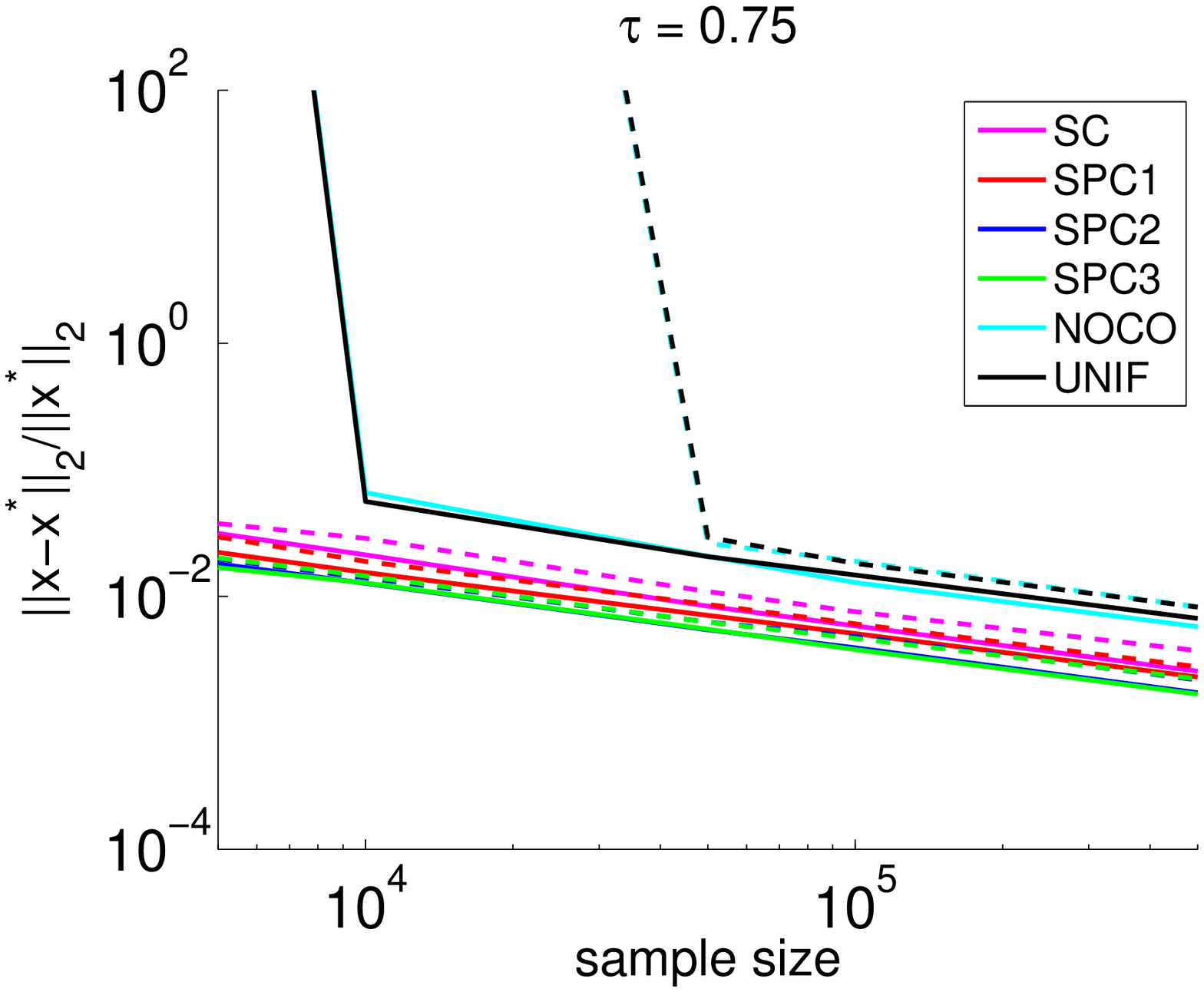}
 }
 \end{tabular}
 \end{center}
 \caption{ 
   The first (solid lines) and the third (dashed lines) quartiles of the relative errors of the objective value
   (namely, $|f-f^*|/|f^*|$) and solution vector (namely, $\|x-x^*\|_2/\|x^*\|_2$),
   by using 6 different methods, among 30 independent trials, as a function of the sample size $s$.
   The test is on replicated skewed data with size $2.5e9$ by 50.
   The three different columns correspond to $\tau = 0.5, 0.75$, respectively.
    }
   \label{err_s_large}
\end{figure}

As can be seen, the method preserves the same structure as when the method 
is applied to the medium-scale data.
Still, SPC2 and SPC3 performs slightly better than other methods when $s$ is large enough.
In this case, as before, NOCO and UNIF are not reliable when $s<1e4$.
When $s>1e4$, NOCO and UNIF perform sufficiently closely to the conditioning-based methods
on approximating the objective value.
However, the gap between the performance on approximating the solution vector is significant.

In order to show more detail on the quartiles of the relative errors, we 
generated a table similar to Table~\ref{sol_table} which records the 
quartiles of relative errors on vectors, measured in $\ell_1$, $\ell_2$, and 
$\ell_\infty$ norms by using the six methods when the sampling size $s = 5e4$ and $\tau = 0.75$.
Table~\ref{sol_table_large} shows similar quantities to and should be compared with Table~\ref{sol_table}.
Conditioning-based methods can yield 2-digit accuracy when $s = 5e4$ while NOCO and UNIF cannot.
Also, the relative error is somewhat higher when measured in $\ell_\infty$ norm.

\begin{table}[ht]
\begin{center}
\begin{sc}
\small
\begin{tabular}{c|ccc}
   &  $\|x - x^*\|_2/\|x^*\|_2$ & $\|x - x^*\|_1/\|x^*\|_1$ & $\|x - x^*\|_\infty/\|x^*\|_\infty$ \\
\hline
SC &  [0.0084,  0.0109] & [0.0075,  0.0086] & [0.0112,  0.0159] \\
SPC1 &  [0.0071,  0.0086] & [0.0066, 0.0079] & [0.0080, 0.0105] \\
SPC2 &  [0.0054,   0.0063] &  [0.0053,  0.0061] & [0.0050, 0.0064] \\
SPC3 & [0.0055,   0.0062]  &  [0.0054,  0.0064] &  [0.0050,  0.0067] \\
NOCO & [0.0207,  0.0262] &  [0.0163,  0.0193] &  [0.0288,  0.0397] \\
UNIF & [0.0206,   0.0293]  &  [0.0175,  0.0200] &   [0.0242,  0.0474] \\
\end{tabular}
\end{sc}
\end{center}
\caption{The first and the third quartiles of relative errors of the 
solution vector, measured in $\ell_1$, $\ell_2$, and $\ell_\infty$ norms.
The test is on replicated synthetic data with size $2.5e9$ by 50,
 the sampling size $s = 5e4$, and $\tau = 0.75$.
}
\label{sol_table_large}
\end{table}

Next, we will explore how the accuracy may change as the lower dimension $d$ varies,
and the capacity of our large-scale version algorithm.
In this experiment, we fix the higher dimension of the replicated skewed data to be $1e9$,
and let $d$ take values in $10, 50, 100, 150$.
We will only use SPC2 as it has the relative best condition number.
Figure~\ref{err_s_large_spc2} shows the results of the experiment described above.

\begin{figure}[h!tbp]
 \begin{center}
 \begin{tabular}{ccc}
\subfigure[$\tau = 0.5$, $|f-f^*|/|f^*|$]{
   \includegraphics[width=0.45\textwidth] {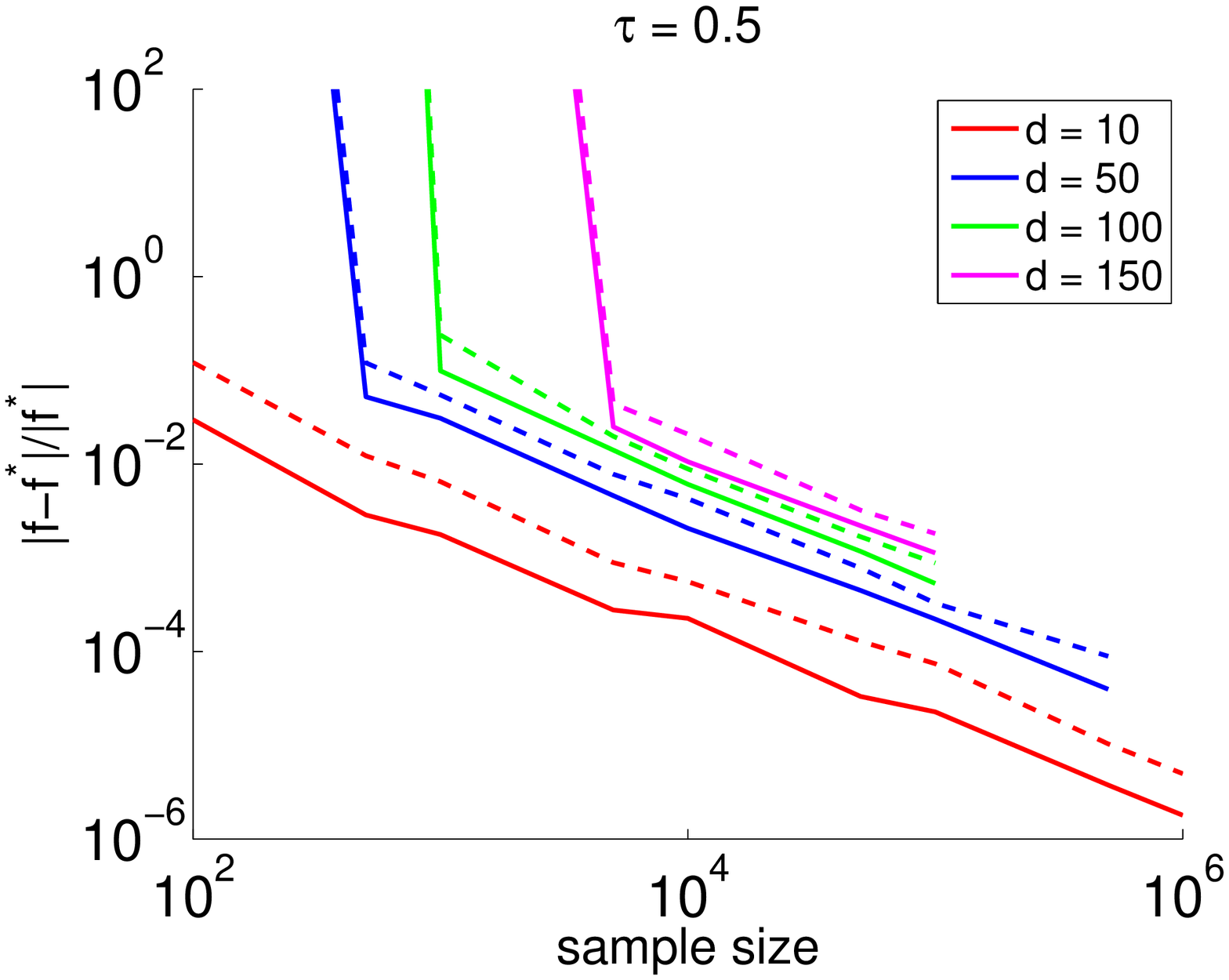}
 }
 &
 \subfigure[$\tau = 0.75$,  $|f-f^*|/|f^*|$]{
   \includegraphics[width=0.45\textwidth] {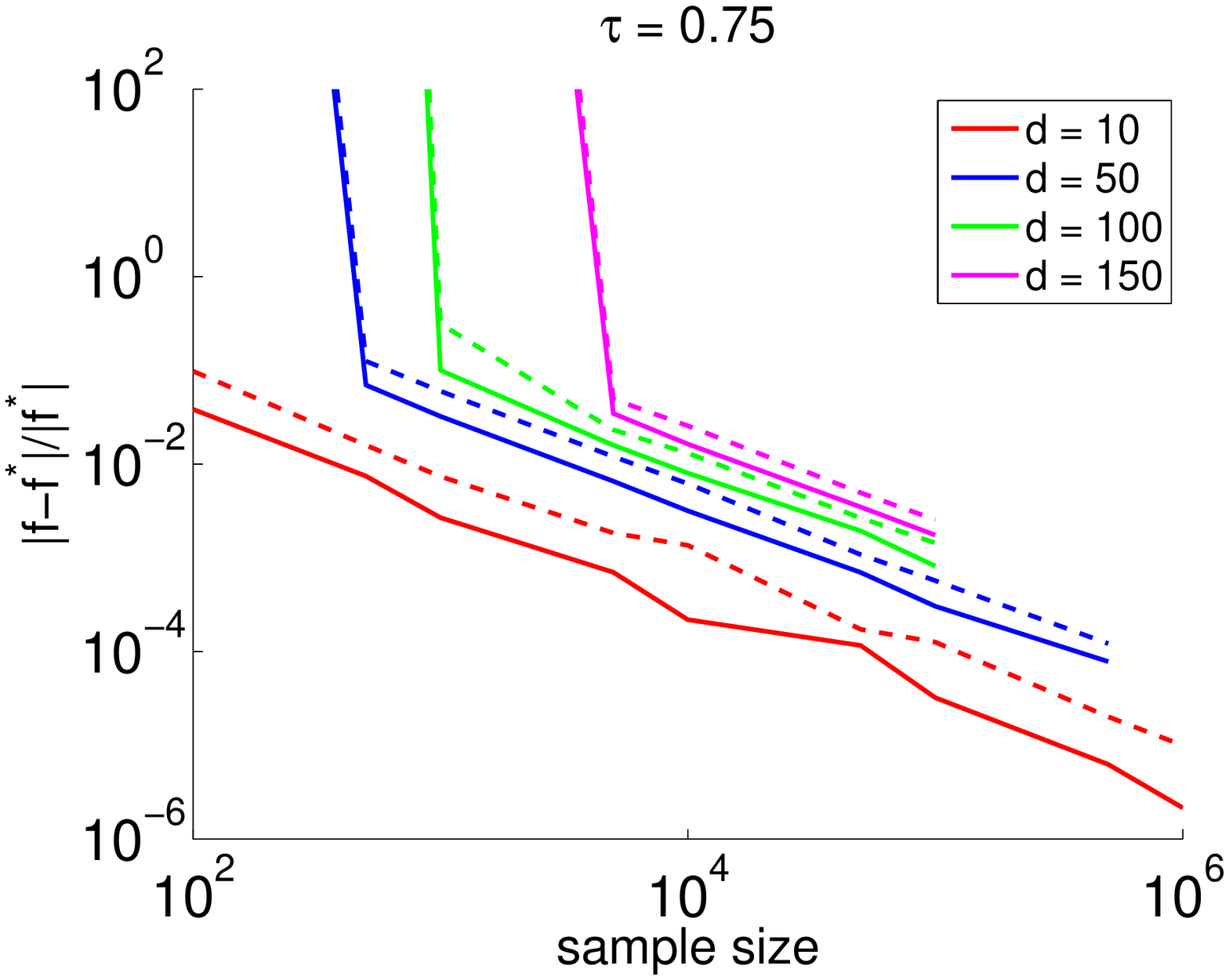}
 }
 \\
 \subfigure[$\tau = 0.5$, $\|x-x^*\|_2/\|x^*\|_2$]{
   \includegraphics[width=0.45\textwidth] {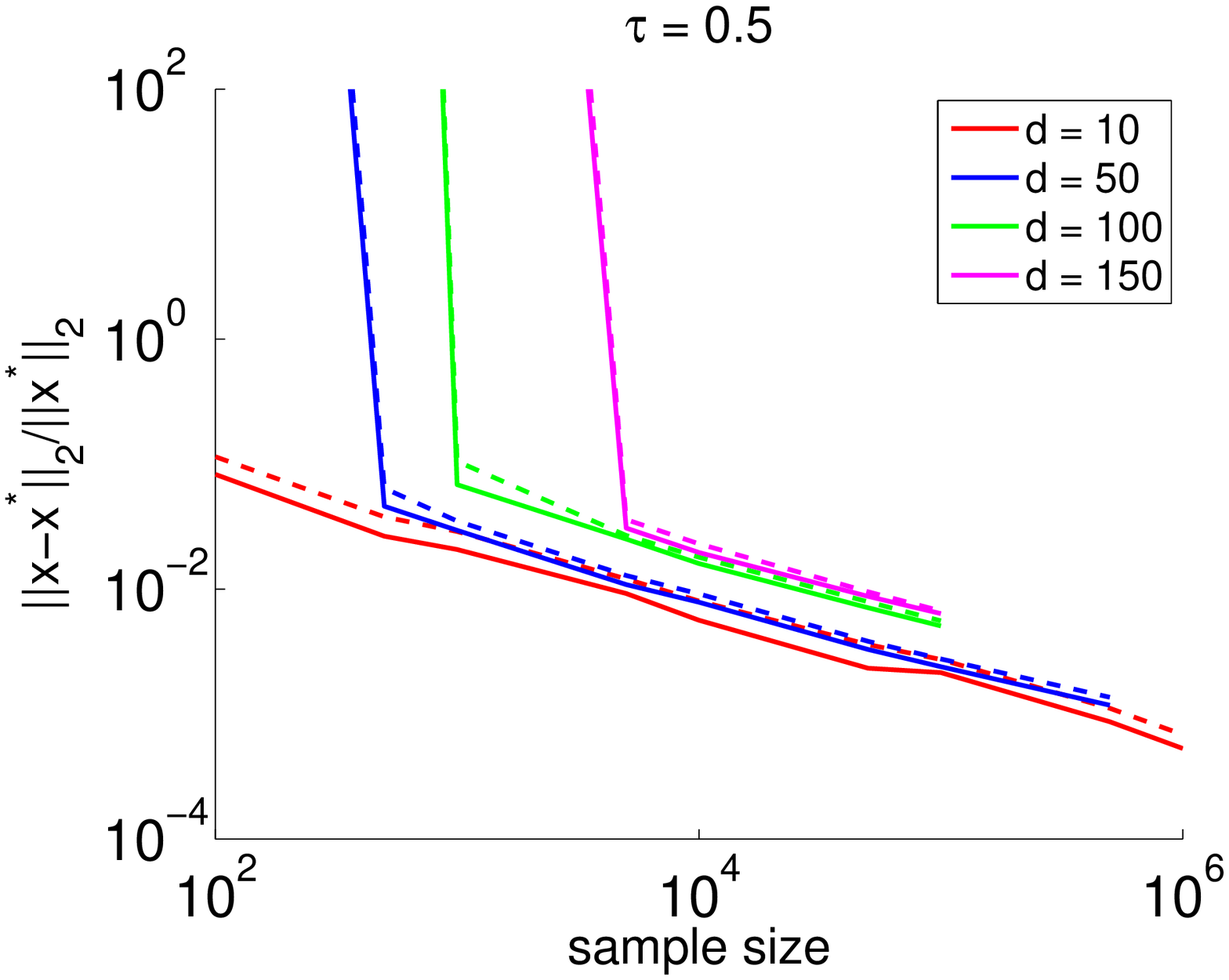}
 }
 &
 \subfigure[$\tau = 0.75$, $\|x-x^*\|_2/\|x^*\|_2$]{
   \includegraphics[width=0.45\textwidth] {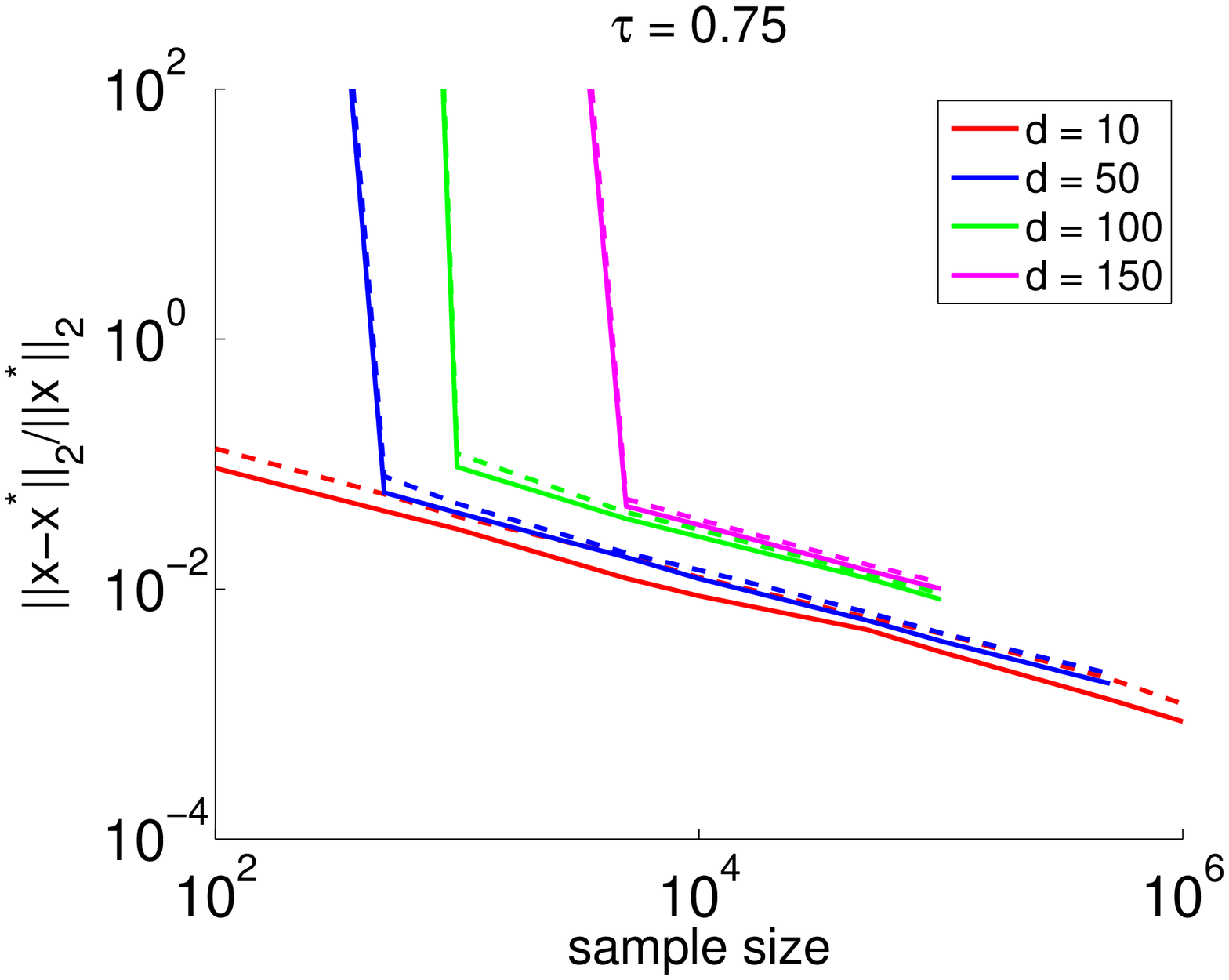}
 }
 \end{tabular}
 \end{center}
 \caption{ 
   The first (solid lines) and the third (dashed lines) quartiles of the relative errors of the objective value
   (namely, $|f-f^*|/|f^*|$) and solution vector (namely, $\|x-x^*\|_2/\|x^*\|_2$),
   by using SPC2, among 30 independent trials, as a function of the sample size $s$.
   The test is on replicated skewed data with $n=1e9$ and $d=10,50,100,150$.
   The two different columns correspond to $\tau = 0.5, 0.75$, respectively.
   The missing points mean that the subproblem on such sampling size
   with corresponding $d$ is unsolvable in RAM.
    }
   \label{err_s_large_spc2}
\end{figure}

From Figure~\ref{err_s_large_spc2},
except for some obvious fact such as the accuracies become lower as $d$ increases
when the sampling size is unchanged, we should also notice that,
the lower $d$ is, the higher the minimum sampling size required to yield acceptable relative errors will be.
For example, when $d=150$, we need to sample at least $1e4$ rows in order to obtain at least one digit accuracy.

Notice also that, there are some missing points in the plot.
That means we cannot solve the subproblem at that sampling size with certain $d$.
For example, solving a subproblem with size $1e6$ by 100 is unrealistic on a single machine.
Therefore, the corresponding point is missing.
Another difficulty we encounter is the capability of conditioning on a single machine.
Recall that, in Algorithm~\ref{cond_alg}, we need to perform QR factorization or ellipsoid rounding on
a matrix, say $SA$, whose size is determined by $d$.
In our large-scale version algorithm, since these two procedures are not parallelizable,
we have to perform these locally.
When $d=150$, the higher dimension of $SA$ will be over $1e7$.
Such size has reached the limit of RAM for performing QR factorization or ellipsoid rounding.
Hence, it prevents us from increasing the lower dimension $d$.

For the census data, we stack it vertically 2000 times to construct 
a realistic data set whose size is roughly $1e10 \times 11$.
In Table~\ref{census_result}, we present the solution computed by our 
randomized algorithm with a sample size $1e5$ at different quantiles, along 
with the corresponding optimal solution. 
As can be seen, for most coefficients, our algorithm provides at least 2-digit 
accuracy. 
Moreover, in applications such as this, the quantile regression result reveals some 
interesting facts about these data.
For example, for these data, marriage may entail a higher salary in lower 
quantiles; Education$^2$, whose value ranged from $0$ to $256$, has a strong 
impact on the total income, especially in the higher quantiles; and
the difference in age doesn't affect the total income much in lower 
quantiles, but becomes a significant factor in higher quantiles.

\begin{table}[h!tpb]
\begin{center}
\begin{sc}
\scriptsize
\begin{tabular}{c|ccccc} 
  Covariate  & $\tau = 0.1$ & $\tau = 0.25$ & $\tau = 0.5$ & $\tau = 0.75$  & $\tau = 0.9$ \\
\hline
    \multirow{2}{*}{intercept} &   8.9812 &    9.3022 &    9.6395 &   10.0515  &  10.5510  \\
            &    [8.9673, 8.9953] &   [9.2876, 9.3106]   &  [9.6337,  9.6484] & [10.0400, 10.0644]  & [10.5296,  10.5825] \\     
  \multirow{2}{*}{female}   & -0.2609 &  -0.2879  &  -0.3227 &   -0.3472  &  -0.3774  \\
   &   [ -0.2657,   -0.2549] &  [ -0.2924,   -0.2846]  &  [-0.3262,   -0.3185]  &  [-0.3481,   -0.3403] &  [ -0.3792,   -0.3708] \\
   \multirow{2}{*}{ Age $\in$ [30, 40)} &  0.2693 &   0.2649  &  0.2748  &  0.2936  &  0.3077   \\
  &  [0.2610,    0.2743]   &  [0.2613,    0.2723]  &   [0.2689,    0.2789]  &   [ 0.2903,    0.2981]  &   [0.3027,    0.3141]     \\
   \multirow{2}{*}{ Age $\in$ [40, 50) } &   0.3173  &  0.3431 &   0.3769  &  0.4118  &  0.4416  \\
  &  [0.3083,    0.3218]  &  [ 0.3407,    0.3561]  &   [ 0.3720,    0.3821]  &   [ 0.4066,    0.4162]  &   [ 0.4386,    0.4496]    \\
  \multirow{2}{*}{Age $\in$ [50, 60) } &   0.3316  &  0.3743  &  0.4188 &   0.4612   & 0.5145  \\
  &  [ 0.3190,    0.3400]  &   [ 0.3686,    0.3839]  &    [0.4118,    0.4266]  &   [0.4540,    0.4636]  &   [ 0.5071,    0.5230]   \\
  \multirow{2}{*}{ Age $\in$ [60, 70) } &  0.3237  &  0.3798   & 0.4418  &  0.5072 &   0.6027  \\
  &  [0.3038,    0.3387]  &    [0.3755,  0.3946]  &   [0.4329,    0.4497]  &   [0.4956,    0.5162]  &    [0.5840,    0.6176]   \\
  \multirow{2}{*}{ Age $\geq$ 70 } &  0.3206  &  0.4132 &   0.5152  &  0.6577 &   0.8699  \\
   &  [0.2962,    0.3455]  &   [0.4012,    0.4359]  &    [0.5036,    0.5308]  &   [ 0.6371,    0.6799]  &   [ 0.8385,    0.8996]  \\
   \multirow{2}{*}{ non\_white } &  -0.0953 &  -0.1018 &   -0.0922 &  -0.0871 &  -0.0975  \\
   & [-0.1023,   -0.0944]  &  [-0.1061,   -0.0975]  &  [-0.0985,   -0.0902]  &  [-0.0932,   -0.0860]  &   [-0.1041,   -0.0932]  \\
 \multirow{2}{*}{ married}  & 0.1175 &   0.1117 &   0.0951  &  0.0870  &  0.0953  \\
  &  [0.1121,    0.1238]  &   [ 0.1059,    0.1162 ]  &  [ 0.0918,    0.0989]  &    [0.0835,    0.0914]  &   [ 0.0909,    0.0987]  \\
  \multirow{2}{*}{ education } &  -0.0152 &   -0.0175 &   -0.0198 &  -0.0470 &  -0.1062  \\
 &  [ -0.0179,   -0.0117]  & [-0.0200,   -0.0149] &  [-0.0225,   -0.0189] &  [-0.0500,   -0.0448] &   [-0.1112,   -0.1032]    \\
  \multirow{2}{*}{ education$^2$ }  & 0.0057  &  0.0062  &  0.0065  &  0.0081 &   0.0119  \\
   & [0.0055,   0.0058]  &  [0.0061,    0.0064] &   [0.0064,    0.0066]  &  [0.0080,    0.0083] &   [0.0117,    0.0122]  \\
\end{tabular}
\end{sc}
\end{center}
  \caption{Quantile regression results for the U.S.~Census 2000 data. The response is the total annual income.
    Except for the intercept and the terms involved with education, all the covariates are $\{0, 1\}$ binary indicators.}
\label{census_result}
\end{table}

To summarize our large-scale evaluation, our main algorithm can handle terabyte-sized quantile regression
problems easily, obtaining, \emph{e.g.}, $2$ digits of accuracy by sampling 
about $1e5$ rows on a problem of size $1e10 \times 11$.  In addition, its running 
time is competitive with the best existing random sampling algorithms, and 
it can be applied in parallel and distributed environments.
However, its capability is restricted by the size of RAM since some steps of the algorithms
are needed to be performed locally.

\section{Conclusion}

We have proposed, analyzed, and evaluated new randomized algorithms for solving 
medium-scale and large-scale quantile regression problems.
Our main algorithm uses a subsampling technique that involves constructing 
an $\ell_1$-well-conditioned basis; and our main algorithm runs in nearly 
input-sparsity time, plus the time needed for solving a subsampled problem 
whose size depends only on the lower dimension of the design matrix.
The sampling probabilities used by our main algorithm are derived by 
calculating the $\ell_1$ norms of a well-conditioned basis; and this 
conditioning step is an essential step of our method.
For completeness, we have provided a summary of recently-proposed $\ell_1$ 
conditioning methods, and based on this we have introduced a new method 
(SPC3) in this article.

We have also provided a detailed empirical evaluation of our main algorithm.
This evaluation includes a comparison in terms of the quality of approximation
of several variants of our main algorithm that are obtained by applying 
several different conditioning methods.
The empirical results meet our expectation according to the theory.
Most of the conditioning methods, like our proposed method, SPC3, yield 
2-digit accuracy by sampling only 0.1\% of the data on our test problem.
As for running time, our algorithm is more scalable, when comparing to 
existing competing algorithms, especially when the lower dimension gets up to 
several hundred, while the large dimension is at least one million.
In addition, we show that our algorithm works well for terabytes-size data 
in terms of accuracy and solvability.
 
Finally, we should emphasize that our main algorithm relies heavily on the notion of $\ell_1$ 
conditioning, and that the overall performance of it can be improved if better 
$\ell_1$ conditioning methods are derived.


\section*{Acknowledgments}
This work was supported in part by a grant from the Army Research Office.

\bibliographystyle{plain}
\bibliography{qr_5-trfmt.bib}

\end{document}